\documentclass[a4paper, onecolumn, 10pt, accepted=2023-05-25]{quantumarticle}
\pdfoutput=1

\usepackage{geometry}
\usepackage{soul}
\usepackage{relsize}
\usepackage{slantsc}
\usepackage{bbm}
\usepackage{graphicx}
\usepackage{amsmath}
\usepackage{amssymb}
\usepackage{amsthm}
\usepackage{amsbsy}
\usepackage{mathrsfs}
\usepackage{varioref}
\usepackage{dsfont}
\usepackage{bm}
\usepackage{color}
\usepackage[colorlinks=true]{hyperref}
\usepackage{cleveref}
\usepackage[font = footnotesize, labelfont = bf, justification = RaggedRight]{caption}
\usepackage{subfigure}
\usepackage{booktabs} % for much better looking tables
\usepackage{array} % for better arrays (eg matrices) in maths
\usepackage{paralist} % very flexible & customisable lists (eg. enumerate/itemize, etc.)

\usepackage[parfill]{parskip}
\PassOptionsToPackage{compress}{natbib}
\usepackage[numbers]{natbib}
\newtheorem{definition}{Definition}
\newcommand{\defref}[1]{Definition~\ref{#1}}
\newtheorem{lemma}{Lemma}
\numberwithin{lemma}{section}
\newcommand{\lemref}[1]{Lemma~\ref{#1}}
\newtheorem{prop}{Proposition}
\numberwithin{prop}{section}
\newcommand{\propref}[1]{Proposition~\ref{#1}}
\newtheorem{corollary}{Corollary}
\numberwithin{corollary}{section}
\newcommand{\corref}[1]{Corollary~\ref{#1}}
\newtheorem{theorem}{Theorem}
\numberwithin{theorem}{section}
\newcommand{\thmref}[1]{Theorem~\ref{#1}}

\DeclareMathOperator\supp{supp}
\DeclareMathOperator\spann{span}

\newcommand{\<}{\langle}
\renewcommand{\>}{\rangle}

\newcommand{\lo}{{\mathcal{L}}}
\newcommand{\los}{{\mathcal{L}_{s}}}
\newcommand{\lop}{{\mathcal{L}_{p}}}

\newcommand{\trc}{{\mathcal{T}}}
\newcommand{\s}{{\mathcal{S}}}
\newcommand{\ee}{{\mathcal{E}}}
\newcommand{\xx}{{\mathcal{X}}}
\newcommand{\yy}{{\mathcal{Y}}}
\newcommand{\zz}{{\mathcal{Z}}}

\renewcommand{\aa}{{\mathcal{A}}}

\newcommand{\ff}{{\mathcal{F}}}

\newcommand{\qq}{{\mathcal{Q}}}
\newcommand{\mm}{{\mathcal{M}}}

\newcommand{\ii}{{\mathcal{I}}}
\newcommand{\jj}{{\mathcal{J}}}
\newcommand{\co}{\mathds{C}}
\newcommand{\re}{\mathds{R}}

\newcommand{\h}{{\mathcal{H}}}
\newcommand{\kk}{{\mathcal{K}}}
\newcommand{\hs}{{\mathcal{H}\sub{\s}}}
\newcommand{\ha}{{\mathcal{H}\sub{\aa}}}
\newcommand{\nsys}{{N\sub{\s}}}
\newcommand{\nsysk}[1]{{N\sub{\s}^{#1}}}
\newcommand{\napp}{{N\sub{\aa}}}
\newcommand{\nappk}[1]{{N\sub{\aa}^{#1}}}

\newcommand{\A}{\mathsf{A}}
\newcommand{\B}{\mathsf{B}}
\newcommand{\E}{\mathsf{E}}
\newcommand{\F}{\mathsf{F}}
\newcommand{\Z}{\mathsf{Z}}
\renewcommand{\P}{\mathsf{P}}
\newcommand{\Q}{\mathsf{Q}}
\newcommand{\R}{\mathsf{R}}
\newcommand{\G}{\mathsf{G}}

\newcommand{\nat}{\mathds{N}}

\newcommand{\one}{\mathds{1}}
\newcommand{\onesys}{\mathds{1}\sub{\s}}
\newcommand{\oneapp}{\mathds{1}\sub{\aa}}

\newcommand{\zero}{\mathds{O}}
\newcommand{\imag}{\mathfrak{i}}

\newcommand{\tr}{\mathrm{tr}}

\newcommand{\av}{\mathrm{av}}
\newcommand{\sub}[1]{_{\!\mathsmaller{\, #1}}}
\newcommand{\eq}[1]{Eq.~\eqref{#1}}
\newcommand{\fig}[1]{Fig.~\ref{#1}}
\newcommand{\sect}[1]{Sec.~\ref{#1}}
\newcommand{\app}[1]{Appendix~(\ref{#1})}
\newcommand{\ket}[1]{|{#1}\rangle}

\newcommand{\pr}[1]{P_{#1}}

\newcommand{\Ad}{\mathrm{Ad}}

\newcommand{\var}[1]{\mathrm{Var}\left({#1} \right)}

\title{Measurement disturbance and conservation laws in quantum mechanics}
\author{M. Hamed Mohammady}
\email{mohammad.mohammady@ulb.be}
\affiliation{QuIC, \'{E}cole Polytechnique de Bruxelles, CP 165/59, Universit\'{e} Libre de Bruxelles, 1050 Brussels, Belgium}
\affiliation{RCQI, Institute of Physics, Slovak Academy of Sciences, D\'ubravsk\'a cesta 9, Bratislava 84511, Slovakia}
\author{Takayuki Miyadera}
\email{miyadera@nucleng.kyoto-u.ac.jp}
\affiliation{Department of Nuclear Engineering, Kyoto University, Nishikyo-ku, Kyoto 615-8540, Japan}
\author{Leon  Loveridge}
\email{Leon.D.Loveridge@usn.no }
\affiliation{Quantum Technology Group, Department of Science and Industry Systems, University of South-Eastern Norway, 3616 Kongsberg, Norway}

\begin{document}

\begin{abstract}

Measurement error and disturbance, in the presence of conservation laws, are analysed in general operational terms. We provide novel quantitative bounds demonstrating necessary conditions under which accurate or non-disturbing measurements can be achieved, highlighting an interesting interplay between incompatibility, unsharpness, and coherence. From here we obtain a substantial generalisation of the Wigner-Araki-Yanase (WAY) theorem. Our findings are further refined through the analysis of the fixed-point set of the measurement channel, some extra structure of which is characterised here for the first time.

\end{abstract}

\maketitle

\section{Introduction}

That measurements generally disturb quantum systems is one of the fundamental aspects of quantum mechanics. The consequences of this effect range from the foundational
to the applied, sometimes entering in the guise of measurement ``back-action'',  playing a key role in quantum metrology, computation, and information processing \cite{Busch1990,  Ozawa2003, Busch2009, Heinosaari2010,  Tsang2010, Tsang2012, Rozema2012, Groen2013, Hatridge2013,  Busch2013a, Busch2014a, Kaneda2014, Blok2014, Shitara2016, Moller2017, Hamamura-Miyadera, Carmeli2019, Wu2020,  DAriano2020, Ipsen2021}. Measurement disturbance can be  seen when two observables are measured in succession, and the statistics of the second measurement depend on the first. While a well-known necessary condition for  non-disturbance is that the pair of observables must be \emph{compatible} \cite{Heinosaari2015, Guhne2021}, further necessary conditions  arise when the first measurement obeys a conservation law, i.e., when the  interaction between the measured system and measuring apparatus conserves some total quantity such as energy, charge, or angular momentum. Indeed, the  Wigner-Araki-Yanase (WAY) theorem  states that when a single sharp observable is measured in succession, the first measurement will not disturb the second only if the measured observable commutes with the system part of a conserved quantity \cite{E.Wigner1952,Busch2010,Araki1960}. 
The same constraint  holds for perfectly accurate measurements, and independently of disturbance, if the  pointer observable of the apparatus obeys the ``Yanase condition'', i.e., if it  commutes with the apparatus part of a conserved quantity  \cite{Loveridge2011}. 

The WAY theorem has evolved over the years and continues to inspire research in a variety of directions (some recent examples are \cite{Miyadera2006a,Kimura2008,Busch2011, Busch2013, uczak2016,  Tukiainen2017, Tajima2019b, Sotan2021}), having impact also in other fields of research: for instance in quantum computing \cite{Ozawa2002a,Karasawa2007,Karasawa2009},   the resource theories of asymmetry \cite{Ahmadi2013b} and coherence \cite{Aberg2014,Tajima2019}, the theory of quantum reference frames \cite{Loveridge2017a, Loveridge2020a}, quantum clocks \cite{Gisin2018}, and quantum thermodynamics \cite{Navascues2014a, Mohammady2017, Mohammady2019c, Chiribella2021a, Mohammady2021}. Despite the progress that has been made, however, the full scope of the WAY theorem is still not known. For instance, the theorem as stated pertains only to sharp observables, and has been shown only in the limited framework of ``normal'' measurement schemes, where the apparatus is prepared in a pure state and unitarily interacts with the measured system, and where the pointer observable is  sharp. But in the quantum theory of measurement \cite{Busch2016a}, observables are more properly represented by positive operator valued measures (POVMs) which can be unsharp,  measurement interactions are more generally described by channels which can be non-unitary, and the apparatus preparation  may be described by a mixed state.  Additionally, the WAY theorem  addresses disturbance only in the case where the same observable is measured in succession, and the situation where   the first and second observables in the sequence are not the same has received scant attention. In this paper, we shall investigate the role of conservation laws on measurement error and disturbance in the more general setting, stating our results in operational terms, in that the quantitative bounds we employ can be seen to arise from the probabilistic structure of quantum theory in its general measurement theoretic form \cite{PaulBuschMarianGrabowski1995}.

The paper is structured as follows. In  \sect{sec:pre}, we present the elements of operational quantum theory pertinent to our investigation. This includes a  background on the quantum theory of measurement, together with  a quantification of measurement error and disturbance. Readers familiar with these topics can jump directly to \sect{sect:Generalised-WAY}, where the main results of the paper begin; here we present a framework for describing conservation laws in quantum theory, prising apart two distinct notions of conservation---full and average---whose difference manifests for general channels and which plays a key role in interpreting our findings.  Next, we consider sequential measurements where the first measurement obeys a conservation law---whether average or full---and obtain general quantitative bounds for the error in the first measurement to realise a desired target observable, and the disturbance by the first measurement on a second, possibly different, observable. Here, we do not assume that the system observables are sharp, or that the apparatus pointer observable is sharp, or that the measurement interaction is unitary, or that the apparatus preparation is pure. In particular, the bounds demonstrate that in the case of a full conservation law, a large coherence in the apparatus preparation is in general a necessary condition for   approximately accurate and non-disturbing measurements of   observables  not commuting with the system part of a conserved quantity. These bounds are then used to prove a generalisation of the WAY theorem,  given in the form of a single quantitative bound, and capturing many essential features of the original  theorem.  Next, we provide an even stronger  generalisation of the WAY theorem, indicating a deep connection between measurability, non-disturbance, and ``definiteness'',   and demonstrating that there are unsharp observables not commuting with the conserved quantity whose measurement cannot be accurate or non-disturbing irrespective of the apparatus preparation. Finally, in \sect{sect:fixed-point-non-disturbance} we consider how the structure of the set of fixed states of the measurement channel  imposes further restrictions on non-disturbance. In particular, we show that an observable not commuting with the conserved quantity admits a non-disturbing measurement only if the measurement channel disturbs all ``faithful'' states, i.e.,  states with strictly positive eigenvalues.

\section{Preliminaries}\label{sec:pre}
In this section we  introduce the elements of operational quantum theory. This  includes  some background on  observables, instruments, and measurement schemes, as part of the quantum theory of measurement (see, e.g., \cite{Busch2016a, PaulBuschMarianGrabowski1995,Busch1996, Heinosaari2011}). In particular, an operationally motivated  quantification of measurement error and disturbance is provided, together with a review of two  special instances of non-disturbing measurements---measurements of the first kind and repeatable measurements.

\subsection{Operators on Hilbert space, operations, and channels}
Let $\h$ be a complex separable Hilbert space, with $\lo(\h) \supset \los(\h) \supset \lop(\h)$  the algebra of bounded (linear) operators, the real vector space of self-adjoint operators, and the (cone of) positive operators on $\h$, respectively. For any subset $\mathscr{A} \subseteq \lo(\h)$, the \emph{commutant} is defined as $\mathscr{A}':= \{B \in \lo(\h) : [A, B] = \zero \, \forall \, A \in \mathscr{A}\}$.  We shall denote by $\one$ and $\zero$  the identity and null operators of $\lo(\h)$, respectively, and an operator $A \in \lop(\h)$ satisfying $\zero \leqslant A \leqslant \one$ will be called an \emph{effect}.   We define by $\trc(\h) \subseteq \lo(\h)$ the two-sided ideal of trace-class operators in $\lo(\h)$. The (normal) state space  is the space of positive, unit-trace operators $\s(\h) \subset \trc(\h)$, and a state $\rho \in \s(\h)$ is called faithful if for all $A \in \lo(\h)$,  $\tr[A^*A \rho] = 0 \implies A = \zero$, which implies that all of the eigenvalues of $\rho$ are strictly positive.

Transformations of quantum systems are called \emph{operations}, defined as completely positive (CP), trace non-increasing linear maps  $\Phi: \trc(\h) \to \trc(\kk)$. Among the operations are the \emph{channels}, which preserve the trace. For any operation $\Phi : \trc(\h) \to \trc(\kk)$, there is an associated (``Heisenberg picture") dual operation $\Phi^* : \lo(\kk) \to \lo(\h)$,  defined via the duality $\tr[\Phi^*(A) T] = \tr[A \Phi(T)]$ for all $A\in \lo(\kk)$ and $T \in \trc(\h)$.
$\Phi^*$ is completely positive and sub-unital,  and unital exactly when $\Phi$ is trace-preserving. Unital  operations  $\Phi^*$  will also be referred to as channels. In \app{app:inequalities-operations} we present several properties of operations that are of central importance for the proofs of our results, most notably a Cauchy-Schwarz inequality \cite{Janssens2017}.

For channels $\Phi : \trc(\h) \to \trc(\h)$, and their duals $\Phi^* : \lo(\h) \to \lo(\h)$, we define the fixed-point sets as 
\begin{align*}
&\ff(\Phi) := \{T \in \trc(\h) : \Phi(T) = T\},  &\ff(\Phi^*) := \{A \in \lo(\h) : \Phi^*(A) = A\}.   
\end{align*}
Note that  $A \in \ff(\Phi^*)$ is equivalent to $\tr[A \Phi(\rho)] = \tr[A \rho]$ for all states $\rho \in \s(\h)$. Linearity of $\Phi^*$ ensures that $\ff(\Phi^*)$ is closed under linear combinations, and because $\Phi^*$ preserves the involution, $\ff(\Phi^*)^* = \ff(\Phi^*)$. While $\ff(\Phi^*)$ is not generally closed under multiplication,  if $\ff(\Phi)$ contains a faithful state then   multiplicative closure is guaranteed, in which case $\ff(\Phi^*)$ is a $*$-algebra; in fact, it is a von Neumann algebra \cite{Bratteli-operator-algebras-1, Bratteli1998}. See \app{app:fixed-point-von-Neumann} for further details.

\subsection{Observables}

An \emph{observable} of a quantum system $\s$, with Hilbert space $\hs$, 
 is represented by a \emph{normalised  positive operator valued measure} (POVM)  $\E: \Sigma\to \lop(\hs)$, where $\Sigma$ is a $\sigma-$algebra of subsets of some value space $\xx$, representing possible outcomes of a measurement of $\E$. For any $X\in \Sigma$, the positive operator $\zero \leqslant \E(X) \leqslant \onesys$ is referred to as an effect of $\E$.  $\E$ is sigma-additive on disjoint elements of $\Sigma$, and normalisation implies that $\E(\xx)$ is the identity operator on $\hs$. An effect  $\E(X) = \alpha \onesys$, where $\alpha \in [0,1]$, is called \emph{trivial}, and an observable $\E$ is called non-trivial if at least one of the effects in its range is non-trivial. {\it Discrete observables} are those for which $\xx = \{x_1, x_2, \dots\}$ is countable,  in which case $\E$ can be identified with the set  $ \{\E(x) \equiv \E(\{x\})  \in \lop(\hs) :   x \in \xx\} \equiv \E$. If it is not stated otherwise, observables will be assumed to be discrete.   Combined with states, observables give rise to the probabilities 
\begin{align*}
    p^{\E}_\rho(x) := \tr[\E(x) \rho],
\end{align*}
holding for all $\rho \in \mathcal{S}(\hs)$
and all $x \in \xx$, interpreted as the probability of observing outcome $x$ when the observable $\E$ is measured in the state $\rho$.

If $\E$ is a POVM acting in $\hs$, the commutant of $\E$ is denoted by $\E' := \{A \in \lo(\hs) : [\E(x), A]=\zero \, \forall \, x\in \xx\}$. Since $\E =\E^*$ is a self-adjoint set, $\E'$ is a von Neumann algebra, and $\E'' \equiv (\E')'$ is the smallest von Neumann algebra containing $\E$ (i.e., it is the von Neumann algebra generated by $\E$). For any $A \in \lo(\hs)$ such that $A \in \E'$, we write   $[\E, A]=\zero$. Similarly, for any observable $\F :=\{\F(y): y \in \yy\}$ such that $\F \subset  \E'$, we shall write  $[\E, \F]=\zero$. Among the observables are those that are \emph{commutative}, meaning that $\E \subset \E'$ (that is, all the effects $\E(x)$ mutually commute). Among the commutative observables are the {\it sharp} observables, which satisfy the additional condition that for all $x, y \in \xx$, $\E(x) \E(y) = \delta_{x,y}\E(x)$, i.e., $\E(x)$ are mutually orthogonal projection operators. These observables correspond to self-adjoint operators through the spectral theorem. Observables that are not sharp will be called {\it unsharp}, and similarly any effect $E$ which is not a projection will be called unsharp. The unsharpness of $E$ can be quantified through the operator norm as $0\leqslant \| E - E^2\|\leqslant 1/4$, which vanishes exactly when $E$ is a projection. Finally, an observable $\E$ is defined as being ``norm-1'', or having the norm-1 property, if $\| \E(x)\|=1$ for every $x$ for which $\E(x) \ne \zero$. While sharp observables are trivially  norm-1, this property may also be enjoyed by some unsharp observables.

\subsection{Instruments}

\begin{figure}[htbp!]
\begin{center}
\includegraphics[width=0.5\textwidth]{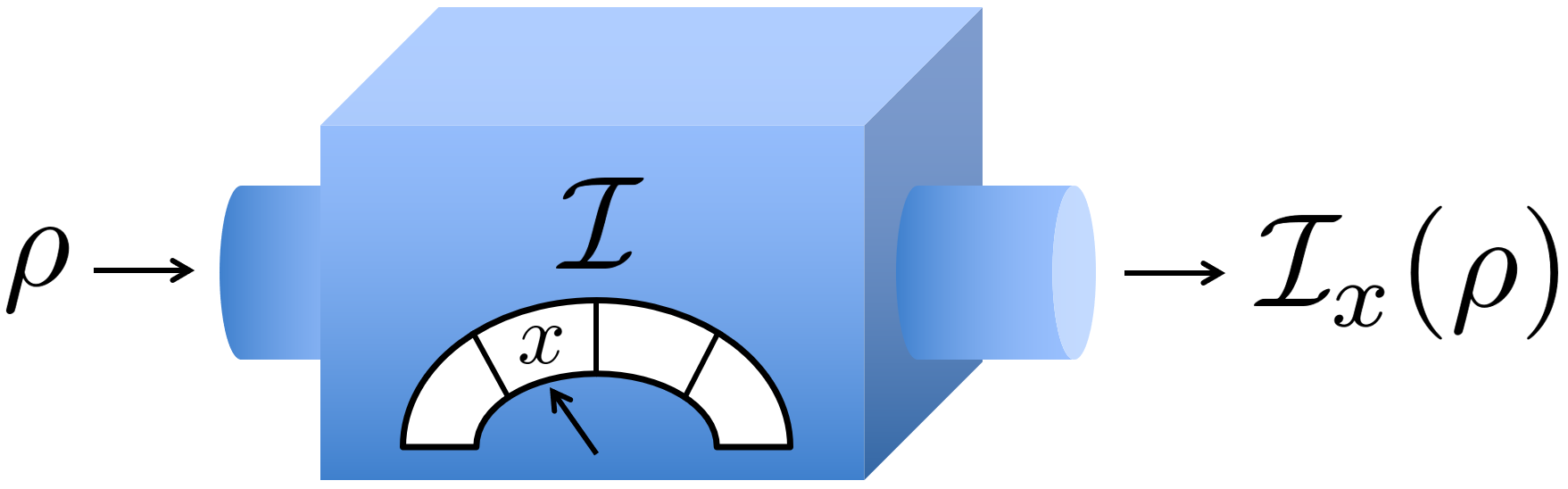}
\vspace*{-0.2cm}
\caption{An instrument measures an observable $\E$ of the system $\s$, and also transforms the system conditional on registering a given outcome. The system, initially prepared in an arbitrary state $\rho$, enters the instrument which then registers outcome $x$  with probability $p^\E_\rho(x) := \tr[\E(x) \rho] = \tr[\ii_x(\rho)]$. Subsequently, the instrument transforms the system to the (non-normalised) state $\ii_x(\rho)$. }\label{fig:Instrument}
\vspace*{-0.5cm}
\end{center}
\end{figure}

\bigskip

Though the state-observable pairings describe the totality of the measurement statistics, this is not sufficient for determining other interesting properties of a measurement, for instance the form of the associated state changes. To this end, we  make use of the notion of \emph{instrument}, or \emph{operation valued measure} \cite{Davies1970, Ozawa2000, Ozawa2001, Pellonpaa2013, Pellonpaa2013a}. A discrete instrument is a collection of operations $ \ii := \{\ii_x \equiv \ii_{\{x\}}  : x\in \xx\}$ such that $\ii_\xx(\cdot) := \sum_{x\in \xx} \ii_x(\cdot)$ is a channel. Throughout, we shall always assume that $\ii$ acts in $\hs$, that is,   $\ii_x:  \trc(\hs) \to \trc(\hs)$.   Each instrument is associated with a unique observable $\E$   via $\ii^*_x(\onesys) = \E(x)$, which implies that $p^{\E}_\rho(x) := \tr[\E(x)\rho] = \tr[\ii_x(\rho)]$. We refer to such an $\ii$ as an $\E$-compatible instrument, or an $\E$-instrument for short, and  to $\ii_\xx(\cdot)$ as the associated $\E$-channel. $\ii_x(\rho)$ is interpreted as the non-normalised state after a measurement of $\E$  has taken place and the outcome $x$ has been registered, and $\ii_\xx(\rho)$ is the normalised state after a non-selective measurement.  A schematic representation of an instrument is given in \fig{fig:Instrument}. 

We note that for every discrete observable $\E$, there are infinitely many $\E$-compatible instruments; every $\E$-instrument $\ii$ can be constructed as the set of  operations $\{\ii_x = \Phi_x \circ\ii^L_x : x\in \xx\}$ \cite{Ozawa2001, Pellonpaa2013a}, where $\Phi_x : \trc(\hs) \to \trc(\hs)$ are arbitrary channels that may depend on outcome $x$, and $\ii^L$ is the  L\"uders instrument for $\E$ \cite{Luders2006},  defined as 
\begin{align}\label{eq:luders}
    &\ii^L_x(T) := \sqrt{\E(x)} T \sqrt{\E(x)} &{\ii^L_x}^*(A) := \sqrt{\E(x)} A \sqrt{\E(x)},
\end{align}
to hold for all $x\in \xx$, $T \in \trc(\hs)$, and $A \in \lo(\hs)$.

\subsection{Measurement schemes}

\begin{figure}[htbp!]
\begin{center}
\includegraphics[width=0.4\textwidth]{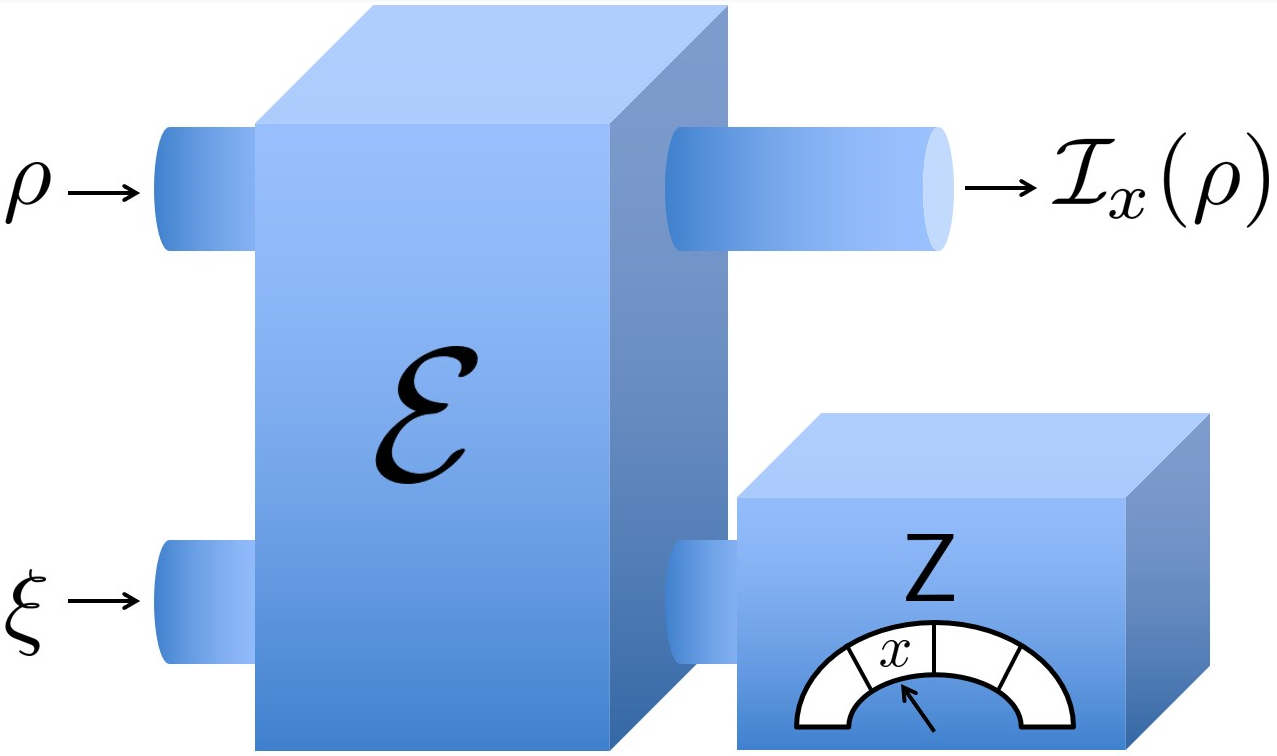}
\vspace*{-0.2cm}
\caption{An $\E$-instrument $\ii$ is implemented on the system $\s$ via a measurement scheme. The system, initially prepared in an arbitrary state $\rho$, and a measuring apparatus $\aa$, prepared in a fixed state $\xi$, undergo a joint evolution by the channel $\ee$. Subsequently, a pointer observable $\Z$ of the apparatus is measured. With probability $p^\E_\rho(x) := \tr[\E(x) \rho] = \tr[\ii_x(\rho)]$ the apparatus registers outcome $x$, thereby transforming the system to the non-normalised state $\ii_x(\rho)$.  }\label{fig:Measurement-Scheme}
\vspace*{-0.5cm}
\end{center}
\end{figure}

\bigskip

An even more comprehensive description of the measurement process involves the modelling of a measuring apparatus $\aa$ and a specification of how it couples to the system under investigation. A  \emph{measurement scheme}  is characterised by the tuple $\mm:= (\ha, \xi, \ee, \Z)$ where: $\ha$ is the Hilbert space for the measuring apparatus $\aa$ and $\xi \in \s(\ha)$ is a state on $\ha$; $\ee: \trc(\hs\otimes \ha) \to \trc(\hs\otimes \ha)$ is a channel   which serves to correlate system and apparatus; and $\Z := \{\Z(x)  : x \in \xx\}$ is a  ``pointer'' observable of the apparatus. The operations of the instrument implemented by $\mm$ can be written as
\begin{align}\label{eq:instrument-dilation}
\ii_x(\cdot) = \tr\sub{\aa}[(\onesys \otimes \Z(x) ) \ee(\cdot \otimes \xi)],
\end{align}
 where  $\tr\sub{\aa}: \trc(\hs\otimes \ha)\to \trc(\hs)$ is the partial trace channel over the apparatus, defined as $\tr[A \tr\sub{\aa}[T]] = \tr[(A \otimes \oneapp) T]$ for all $A\in \lo(\hs)$ and $T \in \trc(\hs\otimes \ha)$. The channel implemented by $\mm$ may thus be written as $\ii_\xx(\cdot) = \tr\sub{\aa}[\ee(\cdot \otimes \xi)]$. A schematic representation of a measurement scheme is given in \fig{fig:Measurement-Scheme}. We note that every $\E$-compatible instrument admits infinitely many \emph{normal} measurement schemes, where $\xi$ is chosen to be pure, $\ee$ is chosen to be unitary, and $\Z$ is chosen to be sharp   \cite{Ozawa1984}. However, unless stated otherwise, we shall consider the more general situation where $\xi$ may be mixed, $\ee$ may be  non-unitary, and $\Z$ may be unsharp. 

We now introduce the unital, completely positive normal conditional expectation  $\Gamma_\xi: \lo(\hs\otimes \ha) \to \lo(\hs)$. $\Gamma_\xi$, called a  {\it restriction map} for $\xi$, is defined as the dual of the isometric embedding (or the {\it preparation map}) $T \mapsto T \otimes \xi$, and  satisfies $\tr[\Gamma_\xi(B)T] = \tr[B (T\otimes \xi)]$ for all $B \in \lo(\hs\otimes \ha)$ and $T \in \trc(\hs)$.  We may use the restriction map to define the channel $\Gamma_\xi^\ee : \lo(\hs\otimes \ha) \to \lo(\hs)$ as 
\begin{align}\label{eq:Gamma-U}
    \Gamma_\xi^\ee(\cdot) := \Gamma_\xi\circ \ee^*(\cdot).
\end{align}
 Using \eq{eq:Gamma-U}, we may express the duals of the operations defined in  \eq{eq:instrument-dilation} as 
\begin{align}\label{eq:instrument-dual-dilation}
    \ii_x^*(\cdot) = \Gamma_\xi^\ee( \cdot \otimes \Z(x) ).
\end{align}
In particular, we may write the dual channel as $\ii^*_\xx(\cdot) =  \Gamma_\xi^\ee(\cdot \otimes \oneapp)$.  

We may also be interested in asking how the apparatus is transformed as a result of the measurement interaction. To this end, we introduce the channel  $\Lambda: \trc(\hs) \to \trc(\ha)$ and its dual $\Lambda^*: \lo(\ha) \to \lo(\hs)$, referred to as \emph{conjugate} channels to $\ii_\xx$ and $\ii_\xx^*$, respectively, defined as
\begin{align}\label{eq:conjugate-channel}
&\Lambda(T) := \tr\sub{\s}[\ee(T \otimes \xi)], &\Lambda^*(A) := \Gamma_\xi^\ee(\onesys \otimes A),
\end{align}
to hold for all $T \in \lo(\hs)$ and $A \in \lo(\ha)$.  That is, $\Lambda(\rho)$ is the state of the apparatus after it has interacted with the system, when the system is initially prepared in state $\rho$. On the other hand, for an initial system state $\rho$,   the expected value of   $A \in \lo(\ha)$ in the  state of the apparatus after the measurement interaction  can be obtained by evaluating the expected value of $\Lambda^*(A)$ in $\rho$.

\subsection{Quantifying measurement error and measurement disturbance}

In order to quantify  \emph{measurement error} and \emph{measurement disturbance},  we shall first provide a quantification of the difference, or discrepancy, between two effects $E$ and $F$ acting in a generic space $\h$. For any state $\rho \in \s(\h)$,  the probabilities  that the properties corresponding to $E$ and $F$ are realised can be compared as $|\tr[\rho (E - F) ]|$, which can be estimated through repeated measurements of $E$ and $F$ in the state $\rho$. Given that we wish to quantify the sense in which $E$ and $F$ differ as effects, i.e., independently of the state, it is natural to take the supremum over all states, and note that 
\begin{equation}\label{eq:effect-difference-op}
    \sup_{\rho \in \s(\h)}\left|\tr[ \rho (E - F) ] \right | = \|E-F\|,
\end{equation}
the right hand side denoting the operator norm, which of course vanishes  when $E=F$ and is non-zero otherwise. \eq{eq:effect-difference-op} gives an operationally motivated---in the sense of being derived directly from the experimental probabilities---quantification of the discrepancy between two effects, which will be utilised in the analysis of measurement error and disturbance.

Let us first address the question of measurement error.  Note that by \eq{eq:instrument-dual-dilation} and \eq{eq:conjugate-channel},  the  observable that is measured by a scheme $\mm:=(\ha, \xi, \ee, \Z)$  has the effects  $\ii_x^*(\onesys) = \Gamma_\xi^\ee(\onesys \otimes \Z(x)) \equiv \Lambda^*(\Z(x))$. Now, let $\E$ be the \emph{target} observable, i.e., the observable we wish to measure, which may be different to the measured observable, but has the same value space $\xx$. By \eq{eq:effect-difference-op}, the measurement error for each effect of the target observable can be quantified through the operator norm as $\| \epsilon(x)\|$, where 
\begin{align}\label{eq:error-equality}
    \epsilon(x) := \Lambda^*(\Z(x)) - \E(x).
\end{align}
A global quantification of measurement error may thus be defined as the largest error over all effects of $\E$, 
\begin{align*}
    \epsilon := \max_{x\in \xx} \|\epsilon(x)\|.
\end{align*}
$\mm$ is a measurement scheme for $\E$ if $\epsilon = 0$, that is, if the target observable is what is actually measured, so that  $\Lambda^*(\Z(x)) = \E(x)$ for all $x$. In the absence of any constraints, perfectly accurate measurements are possible for any target observable $\E$.

\begin{figure}[htbp!]
\begin{center}
\includegraphics[width=0.5\textwidth]{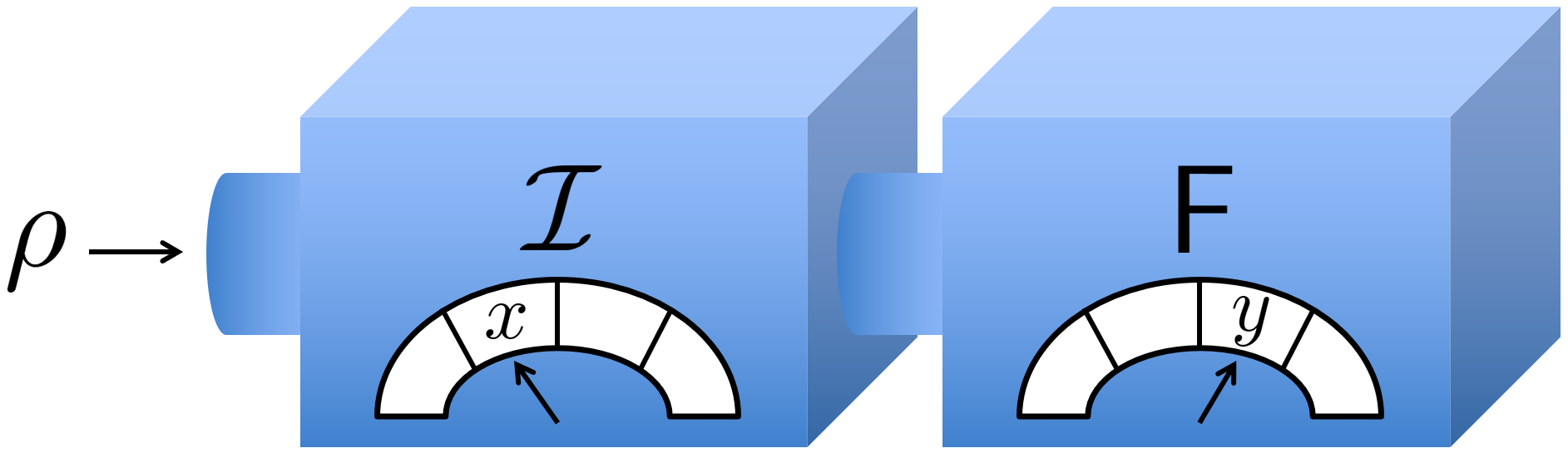}
\vspace*{-0.2cm}
\caption{In a sequential measurement, an observable $\E$  is measured by an instrument $\ii$, and subsequently a second observable $\F$ is measured. $\ii$ does not disturb $\F$ if for all input states $\rho$, the statistics of $\F$ do not depend on whether an $\E$-measurement took place or not.    }\label{fig:Sequential-Measurement}
\vspace*{-0.5cm}
\end{center}
\end{figure}

\bigskip

Now let us address the question of measurement disturbance. Let $\E := \{\E(x) : x \in \xx\}$ and  $\F:= \{\F(y) : y \in \yy\}$ be two observables acting in $\hs$. Consider the  sequential measurement of these observables, as depicted in \fig{fig:Sequential-Measurement},  where at first $\E$ is measured by the instrument $\ii$ (implemented by some measurement scheme $\mm$), and subsequently $\F$ is measured.  For any initial state $\rho \in \s(\hs)$, the probability of observing outcome $y$ of $\F$ after a non-selective measurement by the $\E$-instrument $\ii$ is given as 
\begin{align*}
    \tr[\F(y) \ii_\xx(\rho)] \equiv  \tr[\ii^*_\xx(\F(y)) \rho].
\end{align*}
That is, the prior $\E$-measurement implies that  we perform a measurement of the disturbed observable $\{\ii_\xx^*(\F(y)) : y \in \yy\}$ in the state $\rho$.  By \eq{eq:effect-difference-op}, the disturbance of each effect of $\F$ may be quantified through the operator norm as $\| \delta(y) \|$,   where 
\begin{align}\label{eq:disturbance-quantification}
 \delta(y) &:= \ii^*_\xx(\F(y)) - \F(y).  
 \end{align}
Note that if $\ii$ is implemented by the measurement scheme $\mm:=(\ha, \xi, \ee, \Z)$, then we may equivalently write $\delta(y) = \Gamma_\xi^\ee(\F(y)\otimes \oneapp) - \F(y)$. A global quantification of the disturbance of $\F$ can then be defined as the largest disturbance over all the effects,
\begin{align*}
    \delta := \max_{y \in \yy} \| \delta(y)\|,
\end{align*}
and $\F$ is non-disturbed by $\ii$ exactly when  $\delta = 0$, which is the case when $\ii_\xx^*(\F(y)) = \F(y)$  for all $y \in \yy$. In other words, $\ii$ does not disturb $\F$  exactly when each  $\F(y)$  is a fixed point of the $\E$-channel $\ii^*_\xx$, i.e.,  $\F \subset \ff(\ii_\xx^*)$. In such a case, for any initial state $\rho$, a non-selective measurement  by $\ii$ does not affect the subsequent measurement statistics of   $\F$. In the absence of any constraints,  non-disturbance is always possible when the pair of observables commute; since $\E' \subset \ff({\ii^L_\xx}^*)$ always holds, where $\ii^L$ is the L\"uders $\E$-instrument defined in \eq{eq:luders}, then a L\"uders measurement of $\E$ is guaranteed not to disturb any $\F$ commuting with $\E$ \cite{Busch1998}. While the fixed-point set of the $\E$-channel $\ii_\xx^*$ is not always contained in the commutant of $\E$,  in \app{app:instrument-fixed-points} we present some cases where  $\ff(\ii_\xx^*) \subset \E'$ necessarily holds. For a wider discussion on the relationship between disturbance, commutation, and compatibility, and a quantitative bound relating the minimum disturbance in terms of the commutation between the pair of observables, and the unsharpness of each,  see \app{app:disturbance-compatibility}.

\subsection{Measurements of the first kind, and repeatable measurements}

 A special instance of a non-disturbing measurement is when an $\E$-instrument $\ii$ does not disturb $\E$ itself,  i.e., when $\E \subset \ff(\ii_\xx^*)$. Such measurements are referred to as \emph{measurements of the first kind}.  A subclass of  measurements of the first kind are those which are \emph{repeatable}. Though repeatability is a standard assumption in many textbook treatments of quantum mechanics, that it is a property which a measurement may or may not enjoy appeared already in Wigner's 1952 contribution on the WAY theorem. However, within the general framework presented thus far, repeatability corresponds to a very special form of state change, possible only for a privileged class of observables---an observable $\E$ admits a repeatable measurement only if it is discrete \cite{Ozawa1984}, and all   the effects  have at least one eigenvector with eigenvalue 1 \cite{Busch1995}. $\ii$ is a repeatable  $\E$-instrument if 
\begin{align*}
    \tr[\ii_y \circ \ii_x(\rho) ] =  \delta_{x,y}\tr[\ii_x(\rho)] \qquad \forall \, \rho \in \s(\hs), \,  x,y \in \xx,
\end{align*}
 which implies that 
\begin{align}\label{eq:measurement-repeatable}
    \ii_x^*(\E(y)) =   \delta_{x,y}\E(x)\qquad   \forall \, x,y\in \xx.
\end{align}
The above definition is equivalent to $\ii^*_x(\E(x)) = \E(x)$ for all $x$, since if this holds then  $\ii_x^*(\onesys - \E(x)) = \E(x) - \E(x) = \zero$ \cite{Busch1990}. In other words, if $\ii$ is a repeatable instrument, then repeated measurements by $\ii$ are guaranteed (with probability one) to  produce the same result.  It is straightforward to verify that if a measurement of $\E$  is repeatable, then it is also of the first kind, since 
\begin{align*}
\ii_\xx^*(\E(y)) = \sum_{x \in \xx}\ii_x^*(\E(y))= \sum_{x \in \xx}\delta_{x,y}\E(x) = \E(y).
\end{align*}
While the converse relation does not hold in general---a measurement can be of the first kind and not repeatable, such as is the case for a L\"uders instrument compatible with a commutative but unsharp observable---in the special case of  sharp observables repeatability and first-kindness coincide (Theorem 1 in Ref. \cite{Lahti1991}). In \app{app:repeatable}, we provide a series of results regarding the structure of repeatable instruments. 

\section{Generalisation of the Wigner-Araki-Yanase theorem}\label{sect:Generalised-WAY}
The Wigner-Araki-Yanase (WAY) theorem is the classic result connecting measurement, conservation, and disturbance. This theorem was formulated by Araki and Yanase in 1960 \cite{Araki1960}, capturing in a fairly general setting an observation due to Wigner given in 1952 \cite{E.Wigner1952, Busch2010} regarding spin measurements in the presence of angular momentum conservation. The WAY theorem as formulated in Ref. \cite{Loveridge2011} states that for any discrete sharp observable  represented as a self-adjoint operator $A$ not commuting with the system part of a (bounded, additive) conserved quantity,  the measurement---described by a normal measurement scheme---cannot be repeatable and must violate the Yanase condition, i.e., the pointer observable of the apparatus must fail to commute with the apparatus part of the conserved quantity. In other words, if the Yanase condition is satisfied, then   the measurement cannot  be ``accurate", in the sense that $A$ is not measured by the scheme.  But the WAY theorem does not rule out approximate measurements with approximate repeatability properties, where approximate measurement is understood to mean that the unsharp observable which is actually measured can be made statistically close to $A$. Therefore, WAY has both a strict impossibility part, along with the provision of conditions under which approximate measurements may be possible; as already hinted at by Wigner's original observation, and subsequently refined by  Yanase \cite{Yanase1961} and Ozawa \cite{Ozawa2002} in the form of  quantitative bounds, a normal measurement scheme obeying a conservation law and the Yanase condition can achieve approximately accurate measurements for $A$ only if the uncertainty of the apparatus preparation in the conserved quantity is large. 

While the WAY theorem has developed over the years, its full scope  is still not known. In particular, while much of the previous work around the WAY theorem has focused on the ``measurability question''---upon which observables cannot be measured, or can be only approximately measured given the conservation law---the role of disturbance has been much less fully examined. Moreover, previous proofs of the WAY theorem concerned only  sharp target observables, and were shown in the limited framework of normal measurement schemes, leaving open the question as to whether the implications of the theorem will carry over to the more general setting. In this section, we shall close this gap. 

We begin by  first  introducing two operational definitions of conservation laws for channels---full conservation,  and the generally weaker notion of average conservation. Next, we obtain quantitative bounds for  measurement error and measurement disturbance in the presence of conservation---whether average or full---in the general setting, i.e., without assuming sharpness of the measured observable, the sharpness of the pointer observable, the unitarity of the interaction or the purity of the apparatus preparation. These bounds allow us to prove a generalisation of  WAY, which is presented as a single quantitative   bound that contains both the possibility and impossibility statements of the theorem. A further generalisation of the WAY theorem is also provided, this time presented as an equality that serves to further strengthen the  strict impossibility statement of the WAY theorem for observables that may be unsharp, but are ``definite''.   The section concludes with a demonstration that by  imposing conservation laws on  pointer objectification in addition to  the measurement interaction between system and apparatus, the measurability part of WAY can be recovered without the Yanase condition.

\subsection{Measurement schemes in the presence of  conservation laws}

\begin{figure}[htbp!]
\begin{center}
\includegraphics[width=0.4\textwidth]{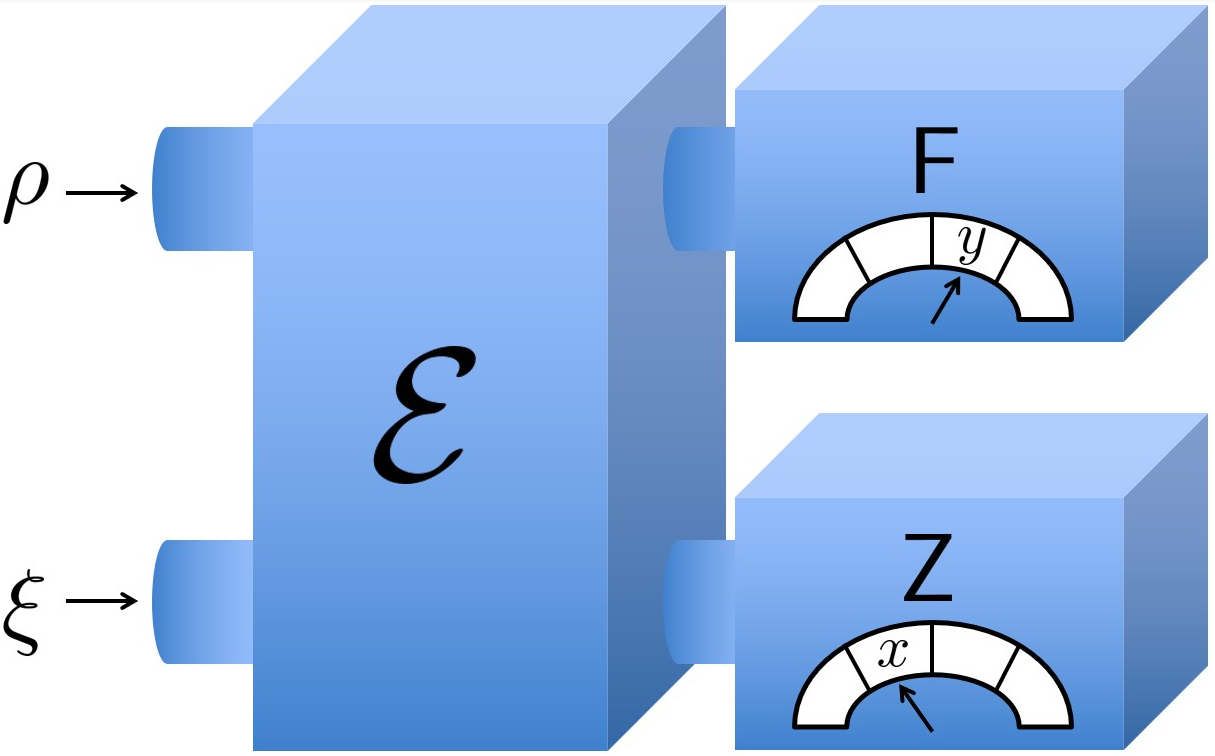}
\vspace*{-0.2cm}
\caption{The system is measured in succession, with the first measurement   implemented by a scheme  $\mm:= (\ha, \xi, \ee, \Z)$ where $\ee$ conserves an additive quantity $N = \nsys \otimes \oneapp + \onesys \otimes \napp$. The conservation law imposes  restrictions on the possibility of $\mm$ to accurately measure a target observable $\E$, and on the possibility of non-disturbance for the second observable $\F$. }\label{fig:Sequential-Measurement-Scheme}
\vspace*{-0.5cm}
\end{center}
\end{figure} 

\bigskip

While every $\E$-compatible instrument $\ii$ admits some  measurement scheme $\mm := (\ha, \xi, \ee, \Z)$,  any restrictions imposed on the elements of $\mm$ will in turn restrict the types of instruments that can be implemented, and hence the class of observables $\E$ that can be accurately measured, and the class of observables $\F$ that will be  non-disturbed.  One such restriction is given by conservation laws---for example, the interaction channel $\ee$ between system and apparatus may be restricted so that the total energy, charge, or angular momentum must be conserved. Before investigating how conservation limits measurements, let us first  consider two operational definitions of  conservation laws for channels, where the conserved quantity $N$ is always assumed to be a bounded self-adjoint operator. In the first analysis, a conservation law can be defined by equality of expectation values  before and after the action of the channel, i.e., average conservation:

\begin{definition}\label{defn:average-conservation}
A channel $\Phi: \trc(\h) \to \trc(\h)$ conserves a self-adjoint operator  $N\in \los(\h)$ on average if  for all $\rho \in \s(\h)$,
\begin{align*}
    \tr[N \Phi(\rho)] = \tr[N \rho],
\end{align*}
i.e., $N \in \ff(\Phi^*)$. 
\end{definition}

However, this does not rule out the higher moments of the ``conserved'' quantity changing their values.  Thus we may strengthen the definition in the following way:

\begin{definition}\label{defn:conservation-law}
A channel $\Phi: \trc(\h) \to \trc(\h)$ fully conserves a self-adjoint operator $N \in \los(\h)$ if for all $\rho \in \s(\h)$ and $k \in \nat$, 
\begin{align*}
      \tr[N^k \Phi(\rho)] = \tr[N^k \rho],
\end{align*}
i.e., $N^k \in \ff(\Phi^*)$ for all $k \in \nat$.
\end{definition}

As shown in \app{app:conservation-law}, full conservation is in fact equivalent to just the first two moments being conserved, i.e., $ N^k \in \ff(\Phi^*)$ for $k=1,2$. Moreover, full conservation is also shown to be equivalent to ``invariance'' of the unitary group generated by $N$ under the action of the channel, i.e., $\Phi^*(e^{\imag t N}) = e^{\imag t N}$ for all $t \in \re$. We note that invariance implies (but is not equivalent to) ``covariance'', i.e., $\Phi^*(e^{\imag t N} A e^{-\imag t N}) = e^{\imag t N} \Phi^*(A) e^{-\imag t N}$ for all $t \in \re$ and $A \in \lo(\h)$.  While full conservation trivially implies average conservation, however,  it is shown that  in general a channel may conserve $N$ on average but not fully. Indeed, it is possible for a channel to conserve $N$ on average while not  being covariant.  Therefore, average conservation is generally a weaker form of conservation law, and is logically distinct from the concept of ``symmetry'' \cite{Marvian2014,Cirstoiu2020}. However, in the special case where $\Phi(\cdot) := U \cdot U^*$ is a unitary channel, average and full conservation coincide, and   are both equivalent to the commutation relation   $[U, N]=\zero$. Since a normal measurement scheme uses a unitary interaction channel, it follows that in such cases  there is no distinction to be drawn between the two notions of conservation law. But if a measurement scheme is not normal, i.e., if the interaction channel is non-unitary, then the distinction between average and full conservation will no longer be void and, as we shall see, leads to interesting consequences.

Throughout what follows, we shall only consider the case where the interaction channel $\ee$ conserves a quantity $N$ that is a bounded,  additive, self-adjoint operator. That is, $N = \nsys \otimes \oneapp + \onesys\otimes \napp$, where  $\nsys  \in \los(\hs)$ and $\napp\in \los(\ha)$ are respectively bounded quantities of the system and apparatus alone.  Note that   conservation of  an additive $N$ by the interaction channel $\ee$ does not generally imply  conservation of $\nsys$ by the channel $\ii_\xx$, since  $\ee$ may allow for an ``exchange'' of the conserved quantity between system and apparatus; specifically, average conservation of $N$ by $\ee$ implies that 
\begin{align}\label{eq:conservation-compensation}
    \tr[\nsys (\ii_\xx(\rho) - \rho)] = \tr[\napp (\xi - \Lambda(\rho))]
\end{align}
holds for all $\rho \in \s(\hs)$, where $\Lambda$ is the conjugate channel of $\ii_\xx$ defined in \eq{eq:conjugate-channel}, so that $\Lambda(\rho)$ is the state of the apparatus after the measurement interaction. We can see that it is possible for the expected value of $\nsys$ to increase (decrease), provided that the expected value of $\napp$  decreases (increases) by an equal amount. Indeed,   such a ``compensation'' by the measuring apparatus is in general  necessary for the instrument $\ii$  to accurately measure some observable $\E$ not commuting with the conserved quantity: if $\E$ is a sharp observable and the $\E$-channel $\ii_\xx$ conserves $\nsys$ on average,  by item (i) of \lemref{lemma:fixed-points-instrument} it holds that $\E$ must commute with $\nsys$. Additionally,   if  $\ii_\xx$ fully conserves $\nsys$ then by item (iii) of \lemref{lemma:fixed-points-instrument}  $\E$ must commute with $\nsys$, even when $\E$ is unsharp.

\subsection{Measurement error and disturbance  under  conservation laws}

Consider  the case where the system is measured in succession, but where the first measurement is constrained by a conservation law, as shown in \fig{fig:Sequential-Measurement-Scheme}. We now present our first main result, providing quantitative bounds for the error of the first measurement in achieving the desired target observable $\E$, and for the disturbance caused by the first measurement on the second observable $\F$. These bounds will be used to obtain our generalisation of the WAY theorem in the sequel. 

\begin{theorem} \label{theorem:error-disturbance-main}
 Let  $\mm:=(\ha, \xi, \ee, \Z)$ be a measurement scheme for an instrument $\ii:=\{\ii_x : x\in \xx\}$ acting in $\hs$, and assume that $\ee$  conserves an additive quantity $N = N\sub{\s} \otimes \oneapp + \onesys\otimes N\sub{\aa}$ on average, where $\nsys \in \los(\hs)$ and $\napp \in \los(\napp)$. Let $\|\epsilon(x)\|$ be the error in measuring the effects of a target observable $\E:=\{\E(x) : x\in \xx\}$ by $\mm$, as defined in \eq{eq:error-equality}. Then for all $x\in \xx$ it holds that
\begin{align}\label{eq:error-bound}
    \|[\E(x), \nsys] - \Lambda^*([\Z(x), \napp]) \| \leqslant 2 \| \nsys\| \|\epsilon(x) \| +  2 \|\Gamma_\xi^\ee(N^2) - \Gamma_\xi^\ee(N)^2 \|^{\frac{1}{2}} \bigg(  2\|\epsilon(x)\| +  \|\E(x) - \E(x)^2 \| \bigg)^{\frac{1}{2}} ,
\end{align}
where $\Gamma_\xi^\ee$ is the channel defined in \eq{eq:Gamma-U}, and $\Lambda$ is the conjugate channel to $\ii_\xx$ defined in \eq{eq:conjugate-channel}. 

Let $\|\delta(y)\|$ be the disturbance of the effects of an observable $\F:=\{\F(y) : y \in \yy\}$ caused by $\ii$, as defined in \eq{eq:disturbance-quantification}. Then  for all $y \in \yy$ it holds that
\begin{align}\label{eq:disturbance-bound}
 \| [\F(y), \nsys] - \ii^*_\xx([\F(y), \nsys]) \| & \leqslant  2  \|\nsys \| \|\delta(y) \| +  2 \|\Gamma_\xi^\ee(N^2) - \Gamma_\xi^\ee(N)^2 \|^{\frac{1}{2}} \bigg(  2\|\delta(y)\| +  \|\F(y) - \F(y)^2 \| \bigg)^{\frac{1}{2}}.
\end{align}

\end{theorem}

The proof for the error bound \eq{eq:error-bound} is provided in \app{app:error-bounds}, and the proof for the disturbance bound \eq{eq:disturbance-bound} is given in \app{app:disturbance-bounds}.  In \app{app:disturbance-compatibility}, we also provide similar bounds for disturbance that are independent of conservation, but take into account the commutation between $\F$ and the observable that is measured by $\ii$.    Note that while the upper bounds of both inequalities above (the terms on the right hand side) are structurally the same, the lower bounds (the terms on the left hand side) are not. Specifically, while the constraints on measurement error depend on the choice of pointer observable, the constraints on disturbance are independent of this. It follows that the implications of these inequalities differ markedly.  We shall illustrate this by considering when the inequalities impose no constraints, i.e., when the lower bounds vanish.

Let us first consider  \eq{eq:error-bound}. This  will not impose any constraints on perfectly accurate measurements, i.e., $\epsilon = 0$, if it holds that $\Lambda^*([\Z(x), \napp]) = [\E(x), \nsys]$, as in such a case the lower bound vanishes. If the pointer observable obeys the Yanase condition, i.e., if $\Z$ commutes with $\napp$, it follows that there will not be any constraints on perfectly accurate measurements  if  $\E$ commutes with $\nsys$, since by complete positivity it holds that $[\Z(x), \napp] =[\E(x), \nsys] = \zero \implies \Lambda^*([\Z(x), \napp]) = [\E(x), \nsys]$. But if we are free to choose the pointer observable, then there exists a measurement scheme obeying a conservation law that can accurately measure all target observables $\E$. For example, let us assume that $\mm$ is a ``trivial'' measurement scheme where the system and apparatus are identical, and where  $\ee$ is a unitary swap channel. That is, choose $\ha \simeq \hs$, $\napp = \nsys$, and  $\ee^*(A \otimes B) = B \otimes A$ for all $A,B$. It is clear that  $\ee$  conserves $\nsys \otimes \oneapp + \onesys \otimes \nsys$. In such a case, $\Lambda^*(A) = A$ for all $A$, and so by choosing the pointer observable so that $\Z = \E$, we obtain $\epsilon(x) = \Lambda^*(\E(x)) - \E(x) = \zero$, and so all target observables are measurable. This is perfectly consistent with \eq{eq:error-bound}, since in such a case  $\Lambda^*([\E(x), \nsys]) = [\E(x), \nsys]$, and so the lower bound  vanishes. But note that  if $\mm$ is trivial, then the instrument that it implements is also trivial, i.e., it will hold that $\ii_\xx^*(\cdot) = \tr[\cdot \xi] \onesys$. In such a case, we have $\delta(y) = \tr[\F(y) \xi] \onesys - \F(y)$, so all non-trivial observables will be disturbed.  

Let us now consider  \eq{eq:disturbance-bound}. This will not impose any constraints on non-disturbance for $\F$, i.e., $\delta = 0$, if    $\F$ commutes with the conserved quantity. This is because by  complete positivity, it holds that $[\F(y), \nsys] = \zero \implies \ii^*_\xx([\F(y), \nsys]) = [\F(y), \nsys]$, in which case the lower bound vanishes. But note that  $ \ii^*_\xx([\F(y), \nsys]) = [\F(y), \nsys]$ does not in general imply that $[\F(y), \nsys] = \zero$. For example, let us assume that $\ii_\xx$ fully conserves  $\nsys$, i.e., $\nsysk{k} \in \ff(\ii_\xx^*)$ for $k=1,2$. In such a case, by the multiplicability theorem (\corref{corollary:multiplicability}) it will hold that $\ii_\xx^*([A, \nsys]) = [\ii_\xx^*(A), \nsys]$ for all $A \in \lo(\hs)$.  If it additionally holds that $\F \subset \ff(\ii_\xx^*)$,  then $ \ii^*_\xx([\F(y), \nsys]) = [\F(y), \nsys]$, so that the lower bound of \eq{eq:disturbance-bound} will vanish and  non-disturbance will not be ruled out for  $\F$, whether it commutes with $\nsys$ or not. But  by item (iii) of \lemref{lemma:fixed-points-instrument} such an instrument $\ii$ will accurately measure $\E$ only if $[\E,\nsys]=\zero$.  Indeed, as a result of the above arguments, and as shown in \corref{corrollary:Luders-commuting-no-constraint}, if the measurement scheme $\mm$ implements the L\"uders instrument $\ii^L$ compatible with an observable $\E$ commuting with $\nsys$, then  non-disturbance will not be ruled out for  any observable $\F$ that commutes with $\E$. 

Notwithstanding the special cases discussed above, when $\E$ and $\F$ do not commute with the conserved quantity, the lower bounds in  \thmref{theorem:error-disturbance-main} will  not vanish in general, in which case the upper bounds must also not vanish. It follows that   a large value of $\|\Gamma_\xi^\ee(N^2) - \Gamma_\xi^\ee(N)^2 \|$ is a necessary condition for achieving an arbitrarily small measurement  error for $\E$ when $[\E(x), \nsys] \ne \Lambda^*([\Z(x), \napp])$ and an arbitrarily small disturbance for $\F$  when $[\F(y), \nsys] \ne \ii^*_\xx([\F(y), \nsys])$. If the error and disturbance are to be  exactly zero, then $\E$ and $\F$ must also be unsharp.    The term $\|\Gamma_\xi^\ee(N^2) - \Gamma_\xi^\ee(N)^2 \|$ is clearly dependent on the choice of apparatus preparation $\xi$ and, as we show below, under the stronger constraint of a full conservation law  this quantity  obtains a clearer interpretation as the uncertainty of $\napp$ in the apparatus preparation, as quantified by the variance.

\begin{lemma}\label{lemma:conservation-variance-condition}
If the channel $\ee$ fully conserves an additive quantity $N = N\sub{\s} \otimes \oneapp + \onesys\otimes N\sub{\aa}$, then   
\begin{align*}
\|\Gamma_\xi^\ee(N^2) - \Gamma_\xi^\ee(N)^2 \| = \var{\napp, \xi},    
\end{align*} 
where  $\var{\napp, \xi}:= \tr[N\sub{\aa}^2 \xi] - \tr[N\sub{\aa} \xi]^2$ denotes the variance of $N\sub{\aa}$ in the state $\xi$. 
\end{lemma}
\begin{proof}
If $N$ is fully conserved by $\ee$, then by \defref{defn:conservation-law} we have  $\ee^*(N^k) = N^k$ for  $k=1,2$. It follows that $\Gamma_\xi^\ee(N^k) = \Gamma_\xi(N^k)$ for $k=1,2$, and hence  $\Gamma_\xi^\ee(N^2) - \Gamma_\xi^\ee(N)^2 = \Gamma_\xi(N^2) - \Gamma_\xi(N)^2$.  Recall that the restriction map satisfies $\Gamma_\xi(A\otimes B) = \tr[B \xi] A$ for all $A \in \lo(\hs)$ and $B \in \lo(\ha)$. It follows  that $\Gamma_\xi(N) = \nsys + \tr[\napp \xi]\onesys$. Now note that  $N^2 = N\sub{\s}^2\otimes \oneapp + 2 \nsys \otimes \napp + \onesys \otimes N\sub{\aa}^2$.  Therefore, 
\begin{align*}
\Gamma_\xi(N^2) &=  N\sub{\s}^2 + 2 \tr[\napp \xi] \nsys + \tr[N\sub{\aa}^2 \xi] \onesys, \nonumber \\
\Gamma_\xi(N)^2 &= N\sub{\s}^2 + 2 \tr[\napp \xi] \nsys + \tr[N\sub{\aa} \xi]^2 \onesys,
\end{align*}
which gives
\begin{align*}
\Gamma_\xi( N^2) - \Gamma_\xi(N)^2 & = \tr[N\sub{\aa}^2 \xi] \onesys - \tr[\napp \xi]^2 \onesys = \var{\napp, \xi} \onesys.    
\end{align*}
\end{proof}

If $\xi$ is a pure state, then a large variance  implies a large coherence. This is because  $\var{\napp,\xi} = 0 \iff [\napp, \xi] = \zero$ for pure states.  Of course, if $\xi$ is a mixed state then it may still be the case that $\var{\napp, \xi}$ is large even if $\xi$ commutes with $\napp$, and hence has zero coherence in the conserved quantity. A quantifier of coherence (or asymmetry) for general states is given by the quantum Fisher information  \cite{PETZ2011, Streltsov2016a, Takagi2018, Marvian2021},  which is equal to four times the convex roof of the  variance \cite{Toth2013, QFI-variance}. Let $\{q_i, \phi_i\}$ be an arbitrary ensemble of (not necessarily orthogonal) unit vectors $\phi_i \in \ha$, with $\{q_i\}$ a probability distribution. The quantum Fisher information of $\napp$ in $\xi$ can be written as
 \begin{align}\label{eq:QFI-defn}
  \qq(\napp, \xi)   = 4 \inf_{\{q_i, \phi_i\}}\left \{\sum_i q_i \var{\napp, \phi_i} : \xi = \sum_i q_i \pr{\phi_i} \right \}.
\end{align}
Here,  $\pr{\psi} \equiv |\psi\>\<\psi|$ denotes the  projection on $\psi$, and we use the short-hand notation $\var{\napp, \phi_i} \equiv \var{\napp, \pr{\phi_i}}$. It is clear that $\qq(\napp, \xi)=0$ if and only if $[ \napp, \xi]=\zero$, while $\qq(\napp, \xi) = 4 \var{\napp, \xi}$ if $\xi$ is a pure state. The following demonstrates that  a large coherence of the conserved quantity in the initial state of the apparatus, when such a state may be mixed, is a necessary condition for accurate and non-disturbing measurements in the presence of a full conservation law.

\begin{prop}\label{prop:First-kind-WAY-Fisher}
Consider again the setup of \thmref{theorem:error-disturbance-main}, and assume that the interaction channel $\ee$ also fully conserves  $N = N\sub{\s} \otimes \oneapp + \onesys\otimes N\sub{\aa}$. Then for all $x\in \xx$ it also holds that
\begin{align}\label{eq:error-bound-QFI}
    \|[\E(x), \nsys] - \Lambda^*([\Z(x), \napp]) \| \leqslant 2 \| \nsys\| \|\epsilon(x) \| +  \frac{1}{2}\qq(\napp, \xi)^{\frac{1}{2}},
\end{align}
and for all $y \in \yy$ it also holds that
\begin{align}\label{eq:disturbance-bound-QFI}
 \| [\F(y), \nsys] - \ii^*_\xx([\F(y), \nsys]) \| & \leqslant 2 \| \nsys\| \|\delta(y) \| + \frac{1}{2}\qq(\napp, \xi)^{\frac{1}{2}}.
\end{align}
\end{prop}
The proof for \eq{eq:error-bound-QFI} is provided in \app{app:error-bounds}, and the proof for \eq{eq:disturbance-bound-QFI} is given in \app{app:disturbance-bounds}. Let us stress that since full conservation implies average conservation, then the inequalities in the above proposition are \emph{complementary} to those given in \thmref{theorem:error-disturbance-main}. That is, when the observables are sharp,  then zero error and zero disturbance are  possible only if $[\E(x), \nsys] = \Lambda^*([\Z(x), \napp])$ and $[\F(y), \nsys] = \ii^*_\xx([\F(y), \nsys])$, respectively, independently  of the coherence in the apparatus preparation. This is because when the observables are sharp and there is zero error for $\E$ or zero disturbance for $\F$, while the upper bounds in \propref{prop:First-kind-WAY-Fisher} may be large, the upper bounds  in \thmref{theorem:error-disturbance-main} vanish. Notwithstanding, we see that so long as the apparatus preparation has a large coherence, approximately accurate measurements for $\E$ and approximate non-disturbance for $\F$ will not be ruled out, even when these observables are sharp.

\subsection{The generalised Wigner-Araki-Yanase theorem}

We are now ready to give two formulations of the WAY theorem which go beyond existing work in several respects. The first formulation is a direct consequence of the quantitative bounds given above.

\begin{theorem}[Generalised  WAY theorem 1]\label{theorem:Generalized-WAY}
Let $\mm := (\ha, \xi, \ee, \Z)$ be a measurement scheme for an $\E$-instrument $\ii$ acting in $\hs$, and  assume that $\ee$  conserves an additive quantity  $N = N\sub{\s} \otimes \oneapp + \onesys\otimes N\sub{\aa}$ on average, where $\nsys \in \los(\hs)$ and $\napp \in \los(\ha)$. If either $\ii$ is repeatable, or the Yanase condition $[\Z, \napp]=\zero$ is satisfied, then   for all $x\in \xx$ it holds that
\begin{align}\label{eq:WAY-bound-1}
    \| [\E(x), \nsys] \| \leqslant  2 \|\Gamma_\xi^\ee(N^2) - \Gamma_\xi^\ee(N)^2 \|^{\frac{1}{2}} \|\E(x) - \E(x)^2 \|^{\frac{1}{2}},
\end{align}
where $\Gamma_\xi^\ee$ is the channel defined in \eq{eq:Gamma-U}. If $\ee$ also fully conserves $N$, and if either $\ii$ is repeatable or the Yanase condition is satisfied, then for all $x\in \xx$ it also holds that
\begin{align}\label{eq:WAY-bound-1-Fisher}
    \| [\E(x), \nsys] \| \leqslant  \frac{1}{2} \qq(\napp, \xi)^\frac{1}{2}, 
\end{align}
where $\qq(\napp, \xi)$ is the quantum Fisher information of $\napp$ in $\xi$ as defined in \eq{eq:QFI-defn}.
\end{theorem}

\begin{proof}

We first prove \eq{eq:WAY-bound-1}.  By  \thmref{theorem:error-disturbance-main}, and setting $\epsilon = 0$, if $\mm$ is a measurement scheme for $\E$ then it must hold that
\begin{align*}
 \|[\E(x), \nsys] - \Lambda^*([\Z(x), \napp]) \| \leqslant  2 \|\Gamma_\xi^\ee(N^2) - \Gamma_\xi^\ee(N)^2 \|^{\frac{1}{2}}  \| \E(x) - \E(x)^2 \|^{\frac{1}{2}}  \qquad \qquad  \qquad \forall \, x \in \xx.
\end{align*}
If the Yanase condition is satisfied, then  $\Lambda^*([\Z(x), \napp]) = \Lambda^*(\zero) = \zero$, and so  we obtain  \eq{eq:WAY-bound-1}.

Now let $\Z$ be an arbitrary pointer observable, but assume that  $\ii$ is a repeatable $\E$-instrument. Recall that repeatability implies first-kindness, which is a specific instance of non-disturbance. Then by \thmref{theorem:error-disturbance-main},  identifying $\F$ with $\E$, and setting $\delta=0$,  it must hold that 
\begin{align*}
\| [\E(x), \nsys] - \ii^*_\xx([\E(x), \nsys]) \| & \leqslant  2 \|\Gamma_\xi^\ee(N^2) - \Gamma_\xi^\ee(N)^2 \|^{\frac{1}{2}} \|\E(x) - \E(x)^2 \|^{\frac{1}{2}} \qquad \qquad  \qquad \forall \, x \in \xx.    
\end{align*}
 By item (vi) of \propref{prop:repeatable-instrument-identity}, if $\ii$ is a repeatable measurement of $\E$, then   $\ii^*_\xx(\E(x) A) = \ii^*_\xx(A \E(x))$ for all $A \in \lo(\hs)$. It follows that  $\ii^*_\xx([\E(x), \nsys])  = \zero$, and so once again we   obtain \eq{eq:WAY-bound-1}. Indeed, let us note that if the measurement is repeatable,  then for any choice of pointer observable $\Z$ it will also hold that $\Lambda^*(\Z(x) B) = \Lambda^*(B \Z(x))$ for all $B \in \lo(\ha)$, and so $\Lambda^*([\Z(x), \napp]) = \zero$ even if $\Z$ violates the Yanase condition. It follows that we can obtain \eq{eq:WAY-bound-1} under the repeatability assumption directly from the measurability bound.     By the same arguments as above, \eq{eq:WAY-bound-1-Fisher} is obtained from  \propref{prop:First-kind-WAY-Fisher}. 
\end{proof}

This theorem goes beyond the original WAY theorem (and its descendants) in the setting of  bounded conserved quantities in the following respects: it holds for general interaction channels, unsharp target observables, unsharp pointer observables, and mixed apparatus states. It also provides an operationally motivated quantitative bound from which the original theorem can be obtained as a special case: if $\E$ is a sharp observable, the upper bound of \eq{eq:WAY-bound-1} vanishes,  in which case an additive  conservation law together with either repeatability or the Yanase condition necessitates commutation of $\E$ with the system part of the conserved quantity. Note that the impossibility statement of the original WAY theorem holds even under the weaker notion of average conservation.  This shows that the impossibility of perfect measurements, for sharp observables not commuting with the conserved quantity, holds in much broader contexts than previously assumed. Indeed, such constraint holds even when the measurement is not constrained by ``symmetry'';  recall that while full conservation of $N$ by $\ee$ implies that $\ee$ is covariant with respect to unitary evolution generated by $N$ (in fact, it is also invariant), it may be the case that $\ee$ conserves $N$ on average without being covariant.

\thmref{theorem:Generalized-WAY} does not rule out  accurate (under the Yanase condition) or repeatable measurements for unsharp observables not commuting with the conserved quantity, provided an appropriate apparatus preparation: in the special case of a full conservation law, the apparatus preparation must have a large coherence in the conserved quantity. However, this does not imply that coherence allows for accurate  or repeatable measurements of \emph{all} unsharp observables. We now present a  further generalisation of the WAY theorem,  providing additional necessary conditions for perfect measurements that are independent of the apparatus preparation.  

\begin{theorem}[Generalised WAY theorem 2]\label{theorem:Strong-WAY}
Let $\mm := (\ha, \xi, \ee, \Z)$ be a measurement scheme for an $\E$-instrument $\ii$ acting in $\hs$, and  assume that $\ee$  conserves an additive quantity  $N = N\sub{\s} \otimes \oneapp + \onesys\otimes N\sub{\aa}$ on average, where $\nsys \in \los(\hs)$ and $\napp \in \los(\ha)$.  If either $\ii$ is a measurement of the first kind, or the Yanase condition $[\Z, \napp]=\zero$ is satisfied, then for any effect $\E(x)$ that has both eigenvalue 1 and 0, it holds that
\begin{align}\label{eq:WAY-bound-2}
  \P [\E(x), \nsys] \P =   [\E(x), \P \nsys \P] = \zero,
\end{align}
where $\P := \P_0(x) + \P_1(x)$, with $\P_0(x)$ and $\P_1(x)$ orthogonal projections onto the eigenvalue-0 and eigenvalue-1 eigenspaces of $\E(x)$, respectively. 
\end{theorem}

For a proof, see \app{app:strong-WAY-proof}. The first equality in \eq{eq:WAY-bound-2}  follows from the fact that $\E(x)\P = \P \E(x) = \P_1(x)$, so that $\E(x)$ commutes with $\P$, and the second equality states that   while the commutator $[\E(x), \nsys]$ may not vanish entirely, it does vanish when projected onto the subspace $\P \hs$. The above theorem is an even stronger extension of the original WAY theorem, as it relaxes the repeatability condition to that of first-kindness; recall that while a measurement that is repeatable is also of the first kind,  repeatability and first-kindness coincide only for sharp observables, and a measurement of an unsharp observable may be of the first kind but not repeatable. Moreover, note that an observable admits a repeatable measurement only if all effects have eigenvalue 1 which, by normalisation, implies that all effects have both eigenvalue 1 and 0. In such a case,  \eq{eq:WAY-bound-2} applies to every effect, and $\P$ may be interpreted as the projection onto the union of  eigenvalue-1 eigenspaces of all the effects of $\E$, i.e., $\P = \sum_{x\in \xx} \P_1(x)$. Finally, note that if $\E$ is sharp, then $\P = \onesys$, in which case  the original WAY theorem is once again recovered. 

The condition of an effect $\E(x)$ having both eigenvalue 1 and 0 implies that the effect is \emph{definite}, or admits definite values. Specifically, such a condition implies that there exist states $\rho$ for which outcome $x$ can be predicted to obtain with probabilistic certainty, i.e., $\tr[\E(x) \rho] = 1$, and that there exist states $\sigma$ for which outcome $x$ can be predicted to \emph{not} obtain with probabilistic certainty, i.e., $\tr[\E(x) \sigma] = 0$. Therefore, a measurement of such an $\E$ allows for perfect distinguishability of states $\rho$ and $\sigma$; if outcome $x$ is observed, we know with probabilistic certainty that the system was not prepared in state $\sigma$. Conversely, if any outcome $y \ne x$ is observed, we know with probabilistic certainty that the system was not prepared in state $\rho$.   \thmref{theorem:Strong-WAY} therefore demonstrates that the impossibility part of the WAY theorem---originally pertaining to sharpness---is more properly understood as concerning observables with definite values, even if unsharp.   That is to say, if we wish to achieve perfectly accurate (under the Yanase condition) or first-kind measurements of an observable that truly does not commute with the conserved quantity, i.e., such that  $[\E(x), \nsys]$ does not vanish even when projected onto a subspace of $\hs$, then not only must such an observable be unsharp, but it must also not admit definiteness.

 \thmref{theorem:Strong-WAY} imposes stronger constraints than \thmref{theorem:Generalized-WAY} and  demonstrates that in general, and irrespective of the apparatus preparation,  there exist unsharp observables not commuting with the conserved quantity that do not admit a repeatable or first-kind measurement, and which cannot be accurately measured if the Yanase condition holds. To illustrate this, let us introduce the following model. Consider a system $\hs \simeq \co^3$ with the orthonormal basis $\{\ket{-1}, \ket{0}, \ket{1}\}$, and the conserved quantity $\nsys = \sum_{n} n |n\>\<n| \equiv |1\>\<1| - |-1\>\<-1|$. Consider also the class of binary observables $\E_\lambda  := \{\E_\lambda(+), \E_\lambda(-)\}$ acting in $\hs$, defined by
\begin{align}\label{eq:toy-model-strong-WAY}
\E_\lambda(\pm) := \lambda |\pm\>\<\pm| + (1-\lambda) |\mp\>\<\mp| + \frac{1}{2} |0\>\<0|    ,
\end{align}
where $1/2 < \lambda \leqslant  1$ and  $\ket{\pm} := \frac{1}{\sqrt{2}}(\ket{1} \pm \ket{-1})$. In the absence of any constraints, all observables in this  class  admit  first-kind measurements. For example, since $\E_\lambda$ is commutative, then the corresponding L\"uders instrument is a first-kind (but not repeatable) measurement of $\E_\lambda$.   On the other hand, $\E_\lambda$ admits a repeatable measurement if and only if $\lambda = 1$,  in which case an instrument with operations $\ii_\pm(\cdot) = \tr[\E_\lambda(\pm) \cdot ] |\pm\>\<\pm|$ is a repeatable measurement of $\E_\lambda$. Since $\E_\lambda$ is unsharp, even when $\lambda=1$, then \thmref{theorem:Generalized-WAY} does not rule out accurate or repeatable measurements for such an observable, provided an appropriate apparatus preparation. But now note that when $\lambda = 1$, both effects have eigenvalue 1 and 0, and we have  $\P_1(\pm) = |\pm\>\<\pm|$ and $\P_0(\pm) = |\mp\>\<\mp|$, and so  $\P = \P_0(\pm) + \P_1(\pm) = |1\>\<1| + |-1\>\<-1|$. It is easily verified that in such a case, 
\begin{align*}
  [\E_\lambda(\pm),  \P \nsys \P ] =    [\E_\lambda(\pm),  \nsys ] =  |\mp 1\>\< \pm 1| - |\pm 1\>\<\mp 1| \ne \zero. 
\end{align*}
 By \thmref{theorem:Strong-WAY}, it follows that when $\lambda = 1$,  a measurement of $\E_\lambda$ that is constrained by an average conservation law cannot be repeatable or even first-kind, and must violate the Yanase condition. However, the effects of $\E_\lambda$ when $\lambda <1$ do not commute with the conserved quantity, and have neither eigenvalue 1 nor eigenvalue 0. In such a case,  \thmref{theorem:Strong-WAY} does not rule out accurate or first-kind (but not repeatable) measurements.

\subsection{The Wigner-Araki-Yanase theorem without the Yanase condition}

 Traditionally, the Yanase condition  is justified by applying the repeatability part of the WAY theorem to the pointer observable; if the pointer observable is sharp, and we consider its measurement as being implemented by a conservative interaction between one measuring apparatus and another, then the pointer observable will admit a repeatable measurement only if it commutes with the conserved quantity. Repeatability of the measurement of the pointer observable is deemed a natural requirement for the possibility of measurement, since an experimenter should be able to confirm the measurement outcome by repeated observations of the apparatus: there must be a stable record of the measurement outcomes.  However, such an argument suffers from two drawbacks. Firstly, it  applies only to sharp pointer observables. Secondly, it runs into the problem of infinite regress, since we have now shifted the role of the ultimate pointer observable from the first apparatus to the second; repeatability of the first pointer observable can be abandoned if the second admits a repeatable measurement, in which case the experimenter may continue to verify the measurement outcomes. In \app{app:Weak-Yanase}, we  show that  the measurability part of the WAY theorem---\thmref{theorem:Generalized-WAY} and \thmref{theorem:Strong-WAY}---can be justified without an appeal to the Yanase condition, but rather by imposing a conservation law on  the total  measurement process, i.e., including pointer objectification. Such conservation is shown to   give rise to the so-called ``weak'' Yanase condition $[\Z^\tau, N]=\zero$, where $\Z^\tau(x) := \ee^*(\onesys \otimes \Z(x))$ is the ``Heisenberg-evolved'' pointer observable \cite{Tukiainen2017}. Subsequently, it is shown that if the weak Yanase condition is satisfied, then \eq{eq:WAY-bound-1} (with $\|\Gamma_\xi^\ee(N^2) - \Gamma_\xi^\ee(N)^2 \|$ replaced by $\var{\napp, \xi}$), \eq{eq:WAY-bound-1-Fisher}, and \eq{eq:WAY-bound-2} will hold.

\section{Fixed points and non-disturbing measurements in the presence of conservation laws}\label{sect:fixed-point-non-disturbance}

In \sect{sect:Generalised-WAY} we provided general quantitative bounds for measurement disturbance, when  observables $\E$ and $\F$ are measured in succession and when the first measurement is subject to a conservation law. As we saw, these bounds generally do not prohibit non-disturbance for an unsharp $\F$ that does not commute with the conserved quantity.   In this section, we provide tighter restrictions on the possibility of non-disturbance that depend on the structure of the fixed-point set $\ff(\ii_\xx^*)$, which intimately depends on the properties of observable $\E$ and the states that are left invariant by the $\E$-channel $\ii_\xx$. We first consider the case where $\ff(\ii_\xx^*)$ is a von Neumann algebra, which is guaranteed to be the case when $\ff(\ii_\xx)$ contains a faithful state. Next, we relax the  faithfulness condition on the states in $\ff(\ii_\xx)$,  and obtain similar restrictions for non-disturbance in the finite-dimensional setting. Finally, we show that in the finite-dimensional case, the first-kindness statement of our generalisation of WAY in \thmref{theorem:Strong-WAY} can be extended to a quantitative bound.  

\subsection{Non-disturbance and von Neumann algebras}

By \thmref{theorem:error-disturbance-main},  when a sharp observable $\F$ is not disturbed, then  the upper bound of \eq{eq:disturbance-bound} vanishes, implying that non-disturbance is  possible only if $[\F(y), \nsys] \in \ff(\ii_\xx^*)$. But in \app{app:disturbance-bounds} we provide a tighter upper bound than that of \eq{eq:disturbance-bound}, which vanishes if both $\F \subset \ff(\ii_\xx^*)$ and $\F^2 := \{\F(y)^2 : y \in \yy\} \subset \ff(\ii_\xx^*)$ hold, in which case non-disturbance will  be possible only if $[\F(y), \nsys] \in \ff(\ii_\xx^*)$. While non-disturbance of $\F$ trivially implies that $\F^2 \subset \ff(\ii_\xx^*)$ when $\F$ is sharp,  the implication  $\F \subset \ff(\ii_\xx^*)\implies \F^2 \subset \ff(\ii_\xx^*)$ holds for all observables whenever $\ff(\ii_\xx^*)$ is a von Neumann algebra, which is guaranteed to be the case whenever  $\ff(\ii_\xx)$ contains at least one faithful state.  We now show that in the presence of conservation laws, if the fixed-point set of the measurement channel $\ii_\xx^*$ is a von Neumann algebra, there are strong constraints imposed on the possibility of non-disturbance.

\begin{theorem}\label{theorem:algebra-fixed-point-WAY}
Let $\E :=\{\E(x) : x\in \xx\}$ and $\F:=\{\F(y) : y \in \yy\}$ be observables acting in $\hs$. Let $\mm:=(\ha, \xi, \ee, \Z)$ be a measurement scheme for an  $\E$-instrument $\ii$, and assume that $\ee$ conserves an additive quantity $N = \nsys \otimes \oneapp + \onesys \otimes \napp$ on average, where $\nsys \in \los(\hs)$ and $\napp \in \los(\ha)$. If $\ff(\ii_\xx^*)$ is a von Neumann algebra, then the following hold:
\begin{enumerate}[(i)]
    \item $\F \subset \ff(\ii_\xx^*)$ only if  $\F$ commutes with $\E$, with  $\Delta \nsys := \ii_\xx^*(\nsys) - \nsys$, and with  $\{[\E(x),\nsys] :x\in \xx\}$.
    
    \item $\ii$ is a measurement of the first kind only if $\E$ is  commutative  and commutes with $\nsys$.
    
    \item $\ii$ is a repeatable measurement only if $\E$ is  sharp and commutes with $\nsys$.
\end{enumerate}
\end{theorem}
We note that the condition $[\E, \F] = \zero$ in item (i) is independent of conservation, and was already shown in Ref. \cite{Heinosaari2010}. The proof of the above theorem is given in \app{app:faithful-fixed-points} (\thmref{theorem:faithful-fixedpoint-WAY}), and here is a rough sketch for item (i). If $\ff(\ii_\xx^*)$ is an algebra, then $\ff(\ii_\xx^*) \subset \E'$, and so non-disturbance is  possible only if $[\E, \F] = \zero$. But if the measurement obeys a conservation law, and $\F \subset \ff(\ii_\xx^*)$, then additionally it holds that $[\F(y), \nsys] \in \ff(\ii_\xx^*)$, which by the multiplicability theorem implies that $[\F(y), \Delta \nsys] = \zero$. Finally, by  $[\F(y), \nsys] \in \ff(\ii_\xx^*) \subset \E'$  and $[\E, \F] = \zero$, it holds that $ [\F(y), [\E(x), \nsys]] = [\E(x), [\F(y), \nsys]] = \zero$.    

Assuming that $\F$ commutes with $\E$, the additional necessary condition for non-disturbance,  $[\F(y), [\E(x), \nsys]] = \zero$, is guaranteed to hold if either  $\E$ commutes $\nsys$, or if $\F$  commutes with $\nsys$.  Note that since $\ff(\ii_\xx^*) \subset \E'$ when the fixed-point set of the measurement channel is an algebra,  whenever $\E$ does not commute with $\nsys$ then $\F$ must also commute with $\Delta \nsys \ne \zero$, which is not in general guaranteed by commutation of $\F$ with $\nsys$. Of course, unless $\E = \F$, non-disturbance may be possible even if neither $\E$ nor $\F$ commute with $\nsys$.   But by items (ii) and (iii) of the above theorem, when $\E = \F$, i.e., when $\ii$ is a first-kind or repeatable measurement of $\E$, then non-disturbance is  possible only if $\E$ commutes with $\nsys$. Indeed, we see that  when the fixed-point set of the measurement channel is an algebra, then the constraints on   repeatability and first-kindness are much stronger than in the more general case as given by \thmref{theorem:Generalized-WAY} and \thmref{theorem:Strong-WAY}. We may therefore strengthen the necessary conditions for repeatability and first-kindness that are given by the WAY theorem with the following:  in the presence of a conservation law, an  unsharp and possibly non-commutative observable $\E$  not commuting with the conserved quantity  admits a repeatable or first-kind measurement only if the $\E$-channel $\ii_\xx$ perturbs all faithful states.

Let us now consider some interesting consequences of the above theorem. As shown by Proposition 6 in Ref. \cite{Heinosaari2010}, when the system is a qubit, i.e.,   $\dim(\hs) = 2$,  then $\ff(\ii_\xx^*)$ is  an algebra for any instrument $\ii$. See also \corref{corollary:first-kind-WAY-qubit} in \app{app:non-disturbance-fixed-point}.  It follows that for qubits,  the implications of \thmref{theorem:algebra-fixed-point-WAY} will hold in general.   Now let us assume that $\E$ is a binary observable with the effects
\begin{align*}
    \E(\pm) = \frac{1}{2} \left(\onesys \pm \lambda \sigma_1\right),
\end{align*}
where $\lambda \in (0, 1]$ and $\sigma_1, \sigma_2, \sigma_3$ are the Pauli operators. This observable is sharp when $\lambda =1$, and is unsharp when $\lambda <1$.  Since binary observables are commutative, then item (ii) of \thmref{theorem:algebra-fixed-point-WAY} will permit a first-kind measurement of $\E$ so long as the conserved quantity commutes with $\sigma_1$. On the other hand,   by item (iii) repeatability will be allowed only if  $\lambda = 1$ also holds. This is not so surprising since  repeatability is  permitted only when all effects have eigenvalue 1, with such condition being satisfied for qubit observables only when the observable is sharp. Now let us assume that the conserved quantity is $\nsys = \sigma_3$, so that it does not commute with $\E$, which implies that repeatability and first-kindness will be ruled out. But can a measurement of $\E$ not disturb some other observable? Note that $[\E(\pm), \sigma_3] = \mp \lambda \imag \sigma_2$. By item (i), a non-trivial observable $\F$ will be non-disturbed only if it commutes with $\sigma_1$ \emph{and} with $\sigma_2$, which is clearly impossible. Indeed, for qubits in the presence of conservation,   non-disturbing measurements are permitted only if  $[\E,\F] = [\E,\nsys] = [\F, \nsys] = \zero$.

Now let us consider the case where an observable $\E$ is measured by the  L\"uders instrument $\ii^L$ defined in \eq{eq:luders}.  If $\E$ is commutative, then it holds that $\ff({\ii^L_\xx}^*) = \E'$ is a von Neumann algebra (as the commutant of a self-adjoint subset of $\hs$), even in infinite dimensions  \cite{Weihua2010, Prunaru2011}. But recall that the L\"uders instrument is a first-kind measurement if  it is compatible with a commutative observable. By item (ii) of \thmref{theorem:algebra-fixed-point-WAY} it follows that in the presence of a conservation law, a commutative $\E$ admits a L\"uders instrument only if $\E$ commutes with the system part of the conserved quantity. Now let us consider an observable $\E$ that may  be non-commutative, but commutes with $\nsys$. Note that in this case, unless $\dim(\hs) < \infty$, then $\ff({\ii^L_\xx}^*)$ is not necessarily an algebra,  since in infinite dimensions there exist non-commutative observables $\E$ such that $\ff({\ii^L_\xx}^*) \not\subset \E'$  \cite{Arias2002,Weihua2009}. But since $\E'\subset \ff({\ii^L_\xx}^*)$ always holds, then   $[\E, \nsys]=\zero \implies N\sub{\s}^k \in \ff({\ii^L_\xx}^*)$ for all $k\in \nat$. That is, the L\"uders measurement of $\E$ commuting with $\nsys$ will fully conserve $\nsys$. Recall that in such a case,   non-disturbance will not be ruled out for  any observable $\F$ that commutes with $\E$ (see \corref{corrollary:Luders-commuting-no-constraint}).  Indeed,  it will hold that $\Delta \nsys = \zero$, and so any $\F$ will trivially commute with $\Delta \nsys$.  But does an observable $\E$ commuting with $\nsys$ always admit a L\"uders measurement in the presence of a conservation law? We shall now show that in the presence of a full conservation law, and where the apparatus part of the conserved quantity is highly non-degenerate,  such measurements will require a large coherence in the apparatus preparation. This surprising observation can be seen as a ``converse'' WAY theorem. 

\begin{prop}\label{prop:converse-WAY-main}
Let $\mm:=(\ha, \xi, \ee, \Z)$ be a measurement scheme for an instrument $\ii$ acting in $\hs$. Assume that $\ee$ fully conserves an additive quantity $N = \nsys \otimes \oneapp + \onesys \otimes \napp$, where $\nsys \in \los(\hs)$ and $\napp \in \los(\ha)$, and that $\ii_\xx$ fully conserves $\nsys$. Define the subspace of $\ha$ that is involved during the measurement process as
\begin{align*}
\ha(\text{meas}) := \bigcup_{\rho \in \s(\hs)} \supp( \Lambda(\rho)) \cup \supp(\xi)  \subseteq \ha,    
\end{align*}
where $\Lambda$ is the conjugate channel to $\ii_\xx$ defined in \eq{eq:conjugate-channel}. Then either  $\ha(\text{meas})$ is contained within a single degenerate eigenspace of $\napp$, or  $\var{\napp, \xi}:= \tr[\nappk{2}\xi] - \tr[\napp \xi]^2$ must be large. Additionally, if $\ii$ is an extremal instrument, and if $\ha(\text{meas})$ is not contained within a single degenerate eigenspace of $\napp$, then the quantum Fisher information $\qq(\napp, \xi)$ as defined in \eq{eq:QFI-defn} must be large. 
\end{prop}

Note that even if $\ii$ is not a L\"uders instrument, by item (iii) of  \lemref{lemma:fixed-points-instrument} it holds that  if $\ii$ is compatible with observable $\E$ then $\ii_\xx$ fully conserves $\nsys$ only if $[\E, \nsys] = \zero$. Moreover, let us note that  the L\"uders $\E$-instrument is extremal whenever the effects of $\E$ are linearly independent \cite{DAriano2011}. The proof of the above proposition is provided in \app{app:converse-WAY}, and here we present a rough sketch. In the case that $\ee$ fully conserves $N$ and $\ii_\xx$ fully conserves $\nsys$  then both the expected value and the variance of $\napp$ must not change as a result of the measurement interaction. That is,  $\tr[\napp\Lambda(\rho)] = \tr[\napp \xi]$ and $\var{\napp, \Lambda(\rho)} = \var{\napp, \xi}$ for all $\rho \in \s(\hs)$.  It follows that if $\xi$ is an eigenstate of $\napp$, i.e., if $\xi$  has support only in a single degenerate eigenspace of $\napp$, then $\Lambda(\rho)$ must live in the same eigenspace for all $\rho$. That is, $\napp$ must be ``effectively''  fully degenerate, in the sense that $\ha(\text{meas})$ must be contained within a single degenerate eigenspace of $\napp$.  This generalises an observation made in Ref. \cite{Mohammady2020}, which held only in the case of L\"uders measurements of sharp observables, implemented by   normal measurement schemes satisfying the Yanase condition.  But in many physically relevant situations  $\napp$ will not be (effectively) fully degenerate---for example, the apparatus may be a system with a conserved quantity $\napp$ that is completely non-degenerate. In such cases,   when the interaction between system and apparatus obeys a full conservation law,  an instrument that fully conserves $\nsys$ can  be implemented only if the apparatus preparation is not  an eigenstate of $\napp$, which implies that $\var{\napp, \xi}$ must be large. Finally, if the instrument is extremal, and $\napp$ is not (effectively) fully degenerate, then for every pure state decomposition $\xi=\sum_i q_i \pr{\phi_i}$, the uncertainty of $\napp$ in $\phi_i$ must be large,   which implies that the apparatus preparation must have a large coherence as quantified by the quantum Fisher information.

\subsection{Non-disturbance and operator spaces}

Due to the Schauder–Tychonoff fixed point theorem \cite{Fixed-points-appl}, all $\E$-channels  $\ii_\xx$ have at least one fixed state. However, it may be  that none of these are faithful. In such a case, the fixed-point set of the dual channel $\ii_\xx^*$ is not necessarily a von Neumann algebra, but rather forms an operator space \cite{Intro-Operator-Space}. This setting has been much less investigated, and  in  \app{app:operator-space-theory} we provide some novel analysis of the structure of such fixed-point sets. While the discussion thus far has  been applicable for infinite-dimensional systems---except in some examples---in this section we shall always assume that $\dim(\hs) < \infty$. 

We define the minimal support projection $P$ on the fixed-point set $\ff(\ii_\xx)$ as  
\begin{align}\label{eq:P-defn-instrument}
   P = \min \{Q : Q \text{ is a projection, } \rho = Q \rho Q  \, \forall \, \rho \in \ff(\ii_\xx)\}. 
\end{align}
In other words, for all projections $Q$ and fixed states $\rho \in \ff(\ii_\xx)$ such that $\rho = Q \rho Q$, it holds that $Q \geqslant P$. Note that  $P=\onesys$ if and only if $\ff(\ii_\xx)$ contains a faithful state, in which case $\ff(\ii_\xx^*)$ is an algebra, so that we recover the results of \thmref{theorem:algebra-fixed-point-WAY}. We now provide a generalisation of this result which accounts for situations where  $P$ may be smaller than the identity, i.e., where  the $\E$-channel $\ii_\xx$ may perturb all faithful states. Here, we define $P \E P := \{P \E(x) P : x\in \xx\}$ and $P \F P := \{P\F(y) P : y \in \yy\}$ as restrictions of observables $\E$ and $\F$ in $P\hs$.

\begin{theorem}\label{theorem:non-faithful-fixedpoint-WAY-main}
Let $\E :=\{\E(x) : x\in \xx\}$ and $\F:=\{\F(y) : y \in \yy\}$ be observables acting in $\hs$. Let $\mm:=(\ha, \xi, \ee, \Z)$ be a measurement scheme for an  $\E$-instrument $\ii$, and assume that $\ee$ conserves an additive quantity $N = \nsys \otimes \oneapp + \onesys \otimes \napp$ on average, where $\nsys \in \los(\hs)$ and $\napp \in \los(\ha)$. If $P$ is  the minimal support projection on $\ff(\ii_\xx)$, then the following hold: 
\begin{enumerate}[(i)]
    \item $\F \subset \ff(\ii_\xx^*)$  only if $P \F P$ commutes with   $P \E P$, with  $P \Delta  \nsys P := P \ii_\xx^*(\nsys ) P - P \nsys P$, and with \\ $\{[P \E(x) P, P \nsys P] : x\in \xx\}$.
    \item $\ii$  is a measurement of the first kind only if  $P \E P$ is  commutative and commutes with $P \nsys P$. 
    
    \item  $\ii$ is repeatable only if $P \E P$ is sharp and commutes with $P \nsys P$.
\end{enumerate}
\end{theorem}

The proof is provided in \app{app:non-disturbance-fixed-point} (\thmref{theorem:non-faithful-fixedpoint-WAY}) and it follows from similar arguments as those used in \thmref{theorem:algebra-fixed-point-WAY}. That is, by noting that there exists a faithful fixed state in the subspace $P\hs$, we observe that  the projection of the fixed-point set  $\ff(\ii_\xx^*)$  onto the subspace $P \hs$  is a   von Neumann algebra.

As a simple example, let us consider the case where $\E$ is measured by a \emph{nuclear} instrument. The operations of a nuclear instrument $\ii$ are written as  $\ii_x(\cdot) = \tr[\E(x) \cdot]\sigma_x$, where $\{\sigma_x\}$ is a family of states.  It is simple to verify that in such a case,  $P$ is the minimal projection on $\cup_x\supp(\sigma_x) $. Additionally, if $\E$ is a norm-1 observable, and for every $x$ the support of $\sigma_x$ is contained within the eigenvalue-1 eigenspace of $\E(x)$, then such an instrument will be a repeatable measurement of $\E$.  Every observable admits a nuclear instrument and, as shown in Corollary 1 of Ref. \cite{Heinosaari2010}, every instrument compatible with a  rank-1 observable is nuclear.    Now assume that $\ii$ does not disturb some observable $\F$. Since the dual $\E$-channel may be written as $\ii_\xx^*(\cdot) = \sum_x \tr[\cdot \sigma_x] \E(x)$, we observe that $\F \subset \ff(\ii_\xx^*)$ implies that $\F(y) = \sum_{x} p(y|x) \E(x)$. That is, non-disturbance is  possible only if $\F$ is a classical post-processing of $\E$. But note that unless $\E$ is commutative, this does not generally imply that $\F$ must commute with $\E$.  However,  item (i) of \thmref{theorem:non-faithful-fixedpoint-WAY-main} states that $P\F P$ must commute with $P \E P$  and with $P \Delta  \nsys P$. Given that $P \Delta  \nsys P = P \ii_\xx^*( \nsys ) P - P \nsys P = \sum_x \tr[\nsys \sigma_x] P \E(x) P -  P \nsys P$, then  $[P\F P,P \E P]=\zero$ implies that the commutation of $P\F  P$  with  $P \Delta  \nsys P$ simplifies to $[P \F P, P \nsys P ]=\zero$. Indeed, in the case where  $\cup_x\supp(\sigma_x) =\hs$ so that $\P = \onesys$, it follows that only observables $\F$ that commute with both the measured observable $\E$, and with the system part of the conserved quantity $\nsys$, will be non-disturbed. 

While the implications of the above theorem depend on the minimal support projection $P$ on the fixed states  of the measurement channel $\ii_\xx$, and hence on the specific measurement implementation, we may use the structure of the fixed-point set to obtain necessary conditions for first-kindness that depend only on the measured observable. In \app{app:non-disturbance-distinguishability}, we provide some necessary conditions for non-disturbance  that are independent of conservation laws,  showing that non-disturbance is intimately related to distinguishability. In particular, we show that if $\ii$ is a first-kind measurement of $\E$, then this observable must be a classical post-processing of a norm-1 observable $\G \subset \ff(\ii_\xx^*) $, and that there exists a family of states $\{\rho_z\}$ that are perfectly distinguishable by a measurement of $\G$ such that $\{\ii_\xx(\rho_z)\}$ remain perfectly distinguishable. Next,  we use these results to obtain    a quantitative form of the WAY theorem for first-kindness, presented below.

\begin{theorem}\label{theorem:first-kind-WAY-bound}

Consider a measurement scheme $\mm := (\ha, \xi, \ee, \Z)$ for  a nontrivial observable $\E$ with the instrument $\ii$ acting in $\hs$. Assume that $\ii$ is a measurement of the first kind,  and  that $\ee$  conserves an additive quantity $N = N\sub{\s} \otimes \oneapp + \onesys\otimes N\sub{\aa}$ on average, where $\nsys \in \los(\hs)$ and $\napp \in \los(\ha)$.   
For each outcome $x$ associated with a non-trivial effect $\E(x)$, let  $\mathcal{K}_{\max}(x)$ and $\mathcal{K}_{\min}(x)$ be  subspaces of $\hs$ defined by 
\begin{align*}
&\mathcal{K}_{\max}(x):=\{\psi \in \hs : \ \E(x)\psi =\| \E(x)\| \psi\}, &\mathcal{K}_{\min}(x):=\{\phi \in \hs:\ (\onesys - \E(x))\phi  =\| \onesys  - \E(x)\| \phi\}.    
\end{align*}
$\mathcal{K}_{\max}(x)$ and $\mathcal{K}_{\min}(x)$ are orthogonal, and  for all unit vectors $\psi \in \mathcal{K}_{\max}(x)$ and $\phi \in \mathcal{K}_{\min}(x)$, it holds that 
\begin{align}\label{eq:distinguishability-first-kind-WAY-bound}
|\langle \psi |\nsys \phi\rangle | 
\leqslant \|\nsys\| \left( \| \E(x)\|^{\frac{1}{2}} (1 - \|\onesys - \E(x) \|)^{\frac{1}{2}} + (1 - \| \E(x) \|)^{\frac{1}{2}} \|\onesys - \E(x) \|^{\frac{1}{2}}\right). 
\end{align}
\end{theorem}
For a proof, see \app{app:first-kind-WAY} (\thmref{theorem:first-kind-WAY-distinguishability}).   We note that if $\E$ commutes with $\nsys$, then \thmref{theorem:first-kind-WAY-bound} imposes no restrictions on first-kindness. This is because so long as $\E(x)$ is non-trivial and commutes with $\nsys$, then the lower bound of \eq{eq:distinguishability-first-kind-WAY-bound}  vanishes. On the other hand, if for any non-trivial $\E(x)$ not commuting with $\nsys$ it holds that $\<\psi|  \nsys \phi\>\ne 0$ for some  $\psi \in \mathcal{K}_{\max}(x)$ and $\phi \in \mathcal{K}_{\min}(x)$, then $\E$  admits a first-kind measurement only if either (i)  $\|\E(x)\| < 1$ or (ii) $\|\onesys - \E(x)\|<1$. This is so because if both (i) and (ii) are violated, i.e., if $\|\E(x)\| = \|\onesys - \E(x)\| = 1$, then the upper bound of \eq{eq:distinguishability-first-kind-WAY-bound} vanishes. It follows that, in such a case,  $\E(x)$ cannot have both   $1$ and $0$ as an eigenvalue: such an observable must be highly unsharp. Indeed, we may also  recover the first-kindness statement of  \thmref{theorem:Strong-WAY} directly from \thmref{theorem:first-kind-WAY-bound}, which is presented  as \corref{corollary:repeatable-WAY-distinguishability} in \app{app:first-kind-WAY}. This raises an interesting question: will \eq{eq:distinguishability-first-kind-WAY-bound} also hold if we abandon the requirement of first-kindness, and instead assume that the measurement  satisfies the Yanase condition? This question is beyond the scope of the present paper, but the answer  may highlight to what extent the necessary conditions for measurability and non-disturbance will continue to satisfy the ``symmetry'' witnessed so far in WAY-type theorems.  

To demonstrate that \thmref{theorem:first-kind-WAY-bound} provides much stronger constraints than \thmref{theorem:Strong-WAY}, let us consider again the simple model of a binary observable $\E_\lambda := \{\E_\lambda(+), \E_\lambda(-)\}$ acting in $\hs \simeq \co^3$ introduced surrounding \eq{eq:toy-model-strong-WAY}. Recall that  \thmref{theorem:Strong-WAY} did not rule out first-kind measurements of $\E_\lambda$ for any $1/2 < \lambda < 1$.  But now note that   $\mathcal{K}_{\max}(\pm) = \spann\{\ket{\pm}\}$,  $\mathcal{K}_{\min}(\pm) = \spann\{\ket{\mp}\}$, $\| \E_\lambda(\pm)\| = \|\onesys -  \E_\lambda(\pm)\| = \lambda$, $\| \nsys \| = 1$,  and  $|\<\pm|\nsys|\mp\>| = 1$. By \thmref{theorem:first-kind-WAY-bound}, it follows that such an observable admits a first-kind measurement only if 
\begin{align*}
    1 \leqslant 2 \sqrt{\lambda (1-\lambda)},
\end{align*}
which cannot be satisfied for any  $1/2< \lambda < 1$; indeed, the above inequality is  satisfied only if $\lambda = 1/2$, in which case $\E_\lambda(\pm) = \onesys/2$ are trivial effects. 

%%%%%%%%%%%%
\section{Conclusions}

We have provided a number of general and operational bounds which capture measurement error and disturbance, with emphasis on the setting in which there is a conservation law---both ``full'' conservation, and the weaker notion of ``average'' conservation. We obtained  new, quantitative versions of the WAY theorem, which generalise previous work in several respects, going beyond normal measurement schemes, and not assuming that the observable to be measured is sharp. The work presented surrounding the WAY theorem was also studied in the novel setting of sequential measurements for general pairs of observables, and the quantitative bounds were further refined by the analysis of the fixed point structure of the measurement channel in settings which have received scant attention.

We saw that the large apparatus coherence played a key role for measurability and non-disturbance in the presence of a full conservation law, pointing to the requirement of ``large" apparatus. This points further to possible deep connections between the WAY theorem and the rapidly developing theory of quantum reference frames, analysed so far only when the conserved quantity has a conjugate phase \cite{Loveridge2020a,Loveridge2017a}. While necessary, however,  the large apparatus coherence was  shown to not be sufficient for good measurements; we saw that conservation laws impose strict constraints on the error or disturbance for unsharp observables that admit definite values.   

Our work suffers from the drawback that many physically arising conserved quantities are unbounded. Very recently, the measurability part of the WAY theorem for sharp target observables was proven in the setting of  unbounded  conserved quantities, where the conservation law is stated as the invariance of the unitary group generated by the conserved quantity under the action of the  measurement interaction \cite{Kuramochi2022}. The measurability question for unsharp target observables, as well as the question of disturbance,   should also be systematically studied when the conserved quantity is unbounded. This is a technically challenging endeavour and we save it for future work.  

\acknowledgments
M.H.M. acknowledges funding from the European Union’s Horizon 2020 research and innovation programme under the Marie Skłodowska-Curie grant agreement No. 801505, as well as  from  the Slovak Academy of Sciences   under MoRePro project OPEQ (19MRP0027).
T.M. acknowledges financial support from JSPS KAKENHI Grant No.JP20K03732.

%\newpage
%
\appendix

\section{Properties of  operations}\label{app:inequalities-operations}

Operations allow for the construction of an ``operator-valued inner product", which will be frequently used in this paper. For an operation $\Phi^*: \lo(\kk) \to \lo(\h)$, we define the sesquilinear mapping $\langle\langle \cdot |\cdot \rangle \rangle 
: \lo(\kk) \times \lo(\kk) \to \lo(\h)$ by 
\begin{align}\label{eq:sesquilinear-map-defn}
\langle \langle A|B \rangle \rangle :=  \Phi^*(A^*B) - \Phi^*(A^*) \Phi^*(B), 
\end{align}
to hold for all $A, B \in \lo(\kk)$. The following lemma shows that such a map mimics several important properties of an inner product.
\begin{lemma}\label{lemma:cauchy-schwarz}
For all $A,B,C \in \lo(\kk)$, the sesquilinear mapping defined in \eq{eq:sesquilinear-map-defn} satisfies: (i) $\<\< A | B + \lambda C\>\> = \<\<A|B\>\> + \lambda \<\<A| C\>\>$ for all $\lambda \in \co$; (ii) $\<\< A| B \>\> = \<\< B| A \>\>^*$; (iii) $\<\< A| A \>\> \geqslant \zero$; and (iv) the Cauchy-Schwarz  inequality
\begin{align*}
\langle \langle A|B\rangle \rangle \langle \langle B|A\rangle \rangle 
\leqslant \Vert \langle \langle B | B \rangle \rangle \Vert \langle \langle A|A \rangle \rangle .   
\end{align*}
\end{lemma}
\begin{proof}
(i) trivially follows from linearity of operations, while (ii) follows from the fact that an operation preserves the involution, i.e., $\Phi^*(A)^* = \Phi^*(A^*)$. (iii) follows from  Kadison's inequality, or the two-positivity of CP maps \cite{Kadison1952,Choi1974}. To show this, note that  by Stinespring's dilation theorem \cite{Stinespring1955} we may write $\Phi^*(A) = V^*(A \otimes \one\sub{\kk'}) V$, where $V : \h \to \kk\otimes \kk'$ is a linear operator. Since $\Phi^*$ is sub-unital, it must hold that $\Phi^*(\one\sub{\kk}) =  V^*(\one\sub{\kk}\otimes \one\sub{\kk'} )V  \equiv V^* V \leqslant \one\sub{\h}$, with equality if $\Phi^*$ is a channel, in which case $V$ is an isometry. By  the C* identity we therefore have  $\|V V^*\| = \|V^* V\| \leqslant 1$, which implies that  $\zero \leqslant V V^* \leqslant \one\sub{\kk}\otimes\one\sub{\kk'}$. By \eq{eq:sesquilinear-map-defn} we may therefore write  
\begin{align}\label{eq:stinespring-sesquilinear}
 \<\< A| B \>\> &=     V^*(A^* \otimes \one\sub{\kk'})\pi^* \pi (B \otimes \one\sub{\kk'})V,
\end{align}
where $\pi = \pi^* := \sqrt{\one\sub{\kk}\otimes\one\sub{\kk'} - V V^*}$. That $\<\<A|A\>\> \geqslant \zero$ trivially follows.

Finally, we prove the Cauchy-Schwarz inequality which, for the case of channels, was  proven by Janssens in   Lemma 1 of Ref.  \cite{Janssens2017}. The proof for the case of general operations is identical; by  \eq{eq:stinespring-sesquilinear} we may write 
\begin{align*}
\langle \langle A|B\rangle \rangle \langle \langle B|A\rangle \rangle &=    V^*(A^* \otimes \one\sub{\kk'})\pi^* \pi (B \otimes \one\sub{\kk'})V  V^*(B^* \otimes \one\sub{\kk'})\pi^* \pi (A \otimes \one\sub{\kk'})V \nonumber \\
& \leqslant \| \pi (B \otimes \one\sub{\kk'})V  V^*(B^* \otimes \one\sub{\kk'})\pi^* \| V^*(A^* \otimes \one\sub{\kk'})\pi^*  \pi (A \otimes \one\sub{\kk'})V \nonumber \\
& =  \|V^*(B^* \otimes \one\sub{\kk'})\pi^* \pi (B \otimes \one\sub{\kk'})V \| V^*(A^* \otimes \one\sub{\kk'})\pi^*  \pi (A \otimes \one\sub{\kk'})V \nonumber \\
& = \|\<\< B| B\>\> \| \<\<A|A \>\>.
\end{align*}
In the second line we have used the fact that for any self-adjoint operator $A \in \los(\h)$, it holds that $B^* A B \leqslant \| A\| B^* B$ for all $B \in \lo(\h)$, while  in the third line we have used the C* identity $\|A A^* \| = \| A^* A\|$ for all $A \in \lo(\h)$. 
\end{proof}
Note that the sesquilinear mapping  in \eq{eq:sesquilinear-map-defn} does not satisfy  the positive definiteness property  in general, that is, $\langle \langle A| A\rangle \rangle =
\zero $ does not  imply $A=\zero$. This  plays an important role in the multiplicability theorem \cite{Choi1974}, which can be seen as a consequence of \lemref{lemma:cauchy-schwarz}:
\begin{corollary}\label{corollary:multiplicability}
Let $\Phi^*: \lo(\kk) \to \lo(\h)$ be an operation, and consider an operator $B \in \lo(\kk)$. The following hold:
\begin{enumerate}[(i)]
    \item If $\Phi^*(B^* B) = \Phi^*(B^*) \Phi^*(B)$, then $\Phi^*(A B) = \Phi^*(A)\Phi^*(B)$ for all $A \in \lo(\kk)$.
    
    \item If $\Phi^*(B B^*) = \Phi^*(B) \Phi^*(B^*)$, then $\Phi^*(B A) = \Phi^*(B)\Phi^*(A)$ for all $A \in \lo(\kk)$.
\end{enumerate}

\end{corollary}
\begin{proof}
Let us first prove (i). If $ \<\< B| B\>\> = \Phi^*(B^* B) - \Phi^*(B^*) \Phi^*(B)=\zero$, then   $\|\<\< B|B \>\>\|=0$. Therefore,  by \lemref{lemma:cauchy-schwarz} we have for all $A \in \lo(\kk)$ the following:
\begin{align*}
   \zero \leqslant   \<\< A^*| B\>\> \<\<A^* | B \>\>^* = \<\< A^*| B\>\> \<\<B | A^* \>\> \leqslant \zero.
\end{align*}
This implies that  $\<\< A^*| B\>\> = \Phi^*(A B) - \Phi^*(A)\Phi^*(B) = \zero$.  Similarly for (ii), $ \<\< B^*| B^*\>\> = \Phi^*(B B^*) - \Phi^*(B) \Phi^*(B^*)=\zero$ implies that for all $A \in \lo(\kk)$ we have
\begin{align*}
   \zero \leqslant   \<\< B^*| A\>\>^* \<\< B^*| A\>\> = \<\< A| B^*\>\> \<\<B^* | A \>\> \leqslant \zero,
\end{align*}
which implies that  $\<\< B^*| A\>\> = \Phi^*(B A) - \Phi^*(B)\Phi^*(A) = \zero$.

\end{proof}

\lemref{lemma:cauchy-schwarz} also has the following useful consequence:
\begin{corollary}\label{corollary:Norm-commutator-inequality}
Let $\Phi^*: \lo(\kk) \to \lo(\h)$ be an operation. Given the sesquilinear mapping defined in \eq{eq:sesquilinear-map-defn}, for all $A,B \in \lo(\kk)$ it holds that
\begin{align}\label{eq:norm-channel-commutator}
    \|[\Phi^*(A), \Phi^*(B)] - \Phi^*([A,B]) \| \leqslant &\| \<\<A|A\>\>  \|^{\frac{1}{2}}  \| \<\<B^*|B^*\>\> \|^{\frac{1}{2}} + \|\<\<A^*|A^*\>\>   \|^{\frac{1}{2}}  \| \<\<B|B\>\> \|^{\frac{1}{2}}.
\end{align}
\end{corollary}
\begin{proof}
Let us first write 
\begin{align*}
 [\Phi^*(A), \Phi^*(B)] - \Phi^*([A,B]) =    \<\<B^*|A\>\> - \<\<A^*|B\>\>, 
\end{align*}
which gives 
\begin{align}\label{eq:norm-channel-commutator-1}
 \| [\Phi^*(A), \Phi^*(B)] -  \Phi^*([A,B])\| \leqslant   \|\<\<B^*|A\>\>\| + \|\<\<A^*|B\>\>\|. 
\end{align}
By  \lemref{lemma:cauchy-schwarz} and the  C* identity $\| A^* \| = \| A \| = \| A^* A\|^{\frac{1}{2}} =  \|  A A^* \|^{\frac{1}{2}}$ for all $A \in \lo(\h)$, 
 we therefore have
\begin{align*}
\|\<\<B^*|A\>\>\| & = \|\<\<B^*|A\>\>\<\<A|B^*\>\> \|^{\frac{1}{2}} \leqslant \| \<\<A|A\>\>  \|^{\frac{1}{2}} \| \<\<B^*|B^*\>\> \|^{\frac{1}{2}},\nonumber \\
\|\<\<A^*|B\>\>\| &= \|\<\<B|A^*\>\>\<\<A^*|B\>\> \|^{\frac{1}{2}} \leqslant \|\<\<A^*|A^*\>\>   \|^{\frac{1}{2}} \| \<\<B|B\>\> \|^{\frac{1}{2}}.
\end{align*}
Inserting the above inequalities in \eq{eq:norm-channel-commutator-1} gives the bound in \eq{eq:norm-channel-commutator}.
\end{proof}

Finally, we present the following useful properties of  operations:
\begin{lemma}\label{lemma:unsharp-disturbance-bound}
Let $\Phi^* : \lo(\kk) \to \lo(\h)$ be an operation. For any effects  $A\in \lop(\kk)$ and $B \in \lop(\h)$,  it holds that 
\begin{align*}
    \|\Phi^*(A^2) - \Phi^*(A)^2 \| \leqslant 2 \|\Phi^*(A) - B \| + \|B - B^2 \|. 
\end{align*}
\end{lemma} 
\begin{proof}
This inequality (for channels) was given as Eq.(4) in Ref. \cite{Miyadera2008}; the proof below follows  Theorem 2 of Ref. \cite{Miyadera2015e}. 
Let us first define $C:= \Phi^*(A) - B$ for notational simplicity. Now,  given that  $\zero \leqslant A \leqslant \one\sub{\kk}$ implies   $A^2 \leqslant A$, we may write  
\begin{align*}
    \Phi^*(A^2) - \Phi^*(A)^2 &\leqslant \Phi^*(A) - \Phi^*(A)^2 \nonumber \\
    & = [C, B] + C\big(\one\sub{\h} -   \Phi^*(A) - B \big) + B - B^2,
\end{align*}
and so we have the bound 
\begin{align*}
    \|\Phi^*(A^2) - \Phi^*(A)^2 \| & \leqslant \| [C , B]\| + \|C\big(\one\sub{\h} -  \Phi^*(A) - B \big)\| + \|B - B^2\| \nonumber \\
    &\leqslant \| [C , B]\| + \|C\|  \|\one\sub{\h} -  \Phi^*(A) - B\| + \|B - B^2\|\nonumber \\
    & \leqslant \| [C , B]\| + \| C\|  +  \|B - B^2\| \nonumber \\
    & \leqslant 2 \|C\| +  \|B - B^2\|.
\end{align*}
In the  third line we use the fact that $A$ and $B$ are effects which, given that  $\Phi^*$ is an operation, gives $\zero \leqslant  \Phi^*(A) + B   \leqslant 2\one\sub{\h} $. This in turn implies that $\|\one\sub{\h} -  \Phi^*(A) - B\|\leqslant 1$. The inequality in the final line follows from Robertson's uncertainty relation, by which we have 
\begin{align*}
\| [C , B]\| = \sup_{\|\phi\|=1} |\<\phi | \imag[C, B] \phi \>| & \leqslant 2  \sqrt{\<\phi|C^2 \phi\> - \<\phi| C\phi\>^2} \sqrt{\<\phi| B^2 \phi\> - \<\phi| B \phi\>^2} \nonumber \\
 & \leqslant 2  \| C\| \sqrt{\<\phi| B^2 \phi\> - \<\phi| B \phi\>^2}\nonumber \\
 & \leqslant \| C\|.
\end{align*}
The final line follows from the fact that $\zero \leqslant B \leqslant \one\sub{\h}$ implies  $\sqrt{\<\phi| B^2 \phi\> - \<\phi| B \phi\>^2} \leqslant 1/2$.

\end{proof}
		
\begin{lemma}\label{lemma:operation-annihilation}
Let $\Phi^* : \lo(\kk) \to \lo(\h)$ be an operation. Assume that $\Phi^*(A)= \zero$ for some $A \in \lop(\kk)$. It holds that 
\begin{align*}
\Phi^*(A B) = \Phi^*(B A) = \zero
\end{align*}
for all $B \in \lo(\kk)$. 
\end{lemma}

\begin{proof}
First, let us note that for any $B \in \lop(\kk)$, we have 
\begin{align*}
\zero \leqslant \Phi^*(A B A ) \leqslant \| B\| \Phi^*(\sqrt{A}A\sqrt{A}) \leqslant   \| B\| \| A\| \Phi^*(A ) = \zero,
\end{align*}
and so $\Phi^*(A B A )=\zero$. By the two-positivity of CP maps, it follows that for any $B \in \lo(\kk)$ we have 
\begin{align*}
\zero = \Phi^*(A B^* B A ) \geqslant \Phi^*(A B^*) \Phi^*(B A) \geqslant \zero.
\end{align*}
The claim immediately follows.
\end{proof}

\section{Faithful fixed states and von Neuman algebras}\label{app:fixed-point-von-Neumann}
Recall that for channels $\Phi : \trc(\h) \to \trc(\h)$, and their duals $\Phi^* : \lo(\h) \to \lo(\h)$,  the fixed-point sets are defined as 
\begin{align*}
&\ff(\Phi) := \{T \in \trc(\h) : \Phi(T) = T\},  &\ff(\Phi^*) := \{A \in \lo(\h) : \Phi^*(A) = A\}.   
\end{align*}
\begin{lemma}[Lindblad]\label{lemma:Lindblad}  
Assume that $ \ff(\Phi)$ contains a faithful state. Then $\ff(\Phi^*)$ is a von Neumann algebra.
\end{lemma}
\begin{proof}

Suppose $ B \in \ff(\Phi^*)$, and define the operator $\Phi^*(B^*B)- \Phi^*(B^*) \Phi^*(B)=\Phi^*(B^*B)-B^*B$, which is positive due to the two-positivity of CP maps. Let $\ff(\Phi)$ contain a faithful state $\omega$. Then we have 
\begin{eqnarray*}
\tr[\omega(\Phi^*(B^*B)-B^*B)]=\tr[\omega(B^*B-B^*B)]=0.
\end{eqnarray*}
But since $\omega$ is faithful and $\Phi^*(B^*B)-B^*B$ is positive, $\tr[\omega(\Phi^*(B^*B)-B^*B)] = 0$  implies that $\Phi^*(B^*B) = B^*B$.  \corref{corollary:multiplicability} therefore implies that for all $A\in\lo(\h)$, 
\begin{eqnarray*}
\Phi^*( A B) = \Phi^*(A) B. 
\end{eqnarray*}
Therefore, if $A\in \ff(\Phi^*)$, then
$\Phi^*( A B) = A B$, and so  $\ff(\Phi^*)$ is closed under  multiplication and is therefore a $*$-algebra.
Finally, if $\ff(\Phi^*)$ is an algebra, then $\ff(\Phi^*) = \{K_i, K_i^*\}' := \{A \in \lo(\h) : [K_i, A] = [K_i^*, A] = \zero \, \forall i\}$, with $\{K_i\}$  any Kraus representation of $\Phi$ \cite{Kraus1983}, making $\ff(\Phi^*)$ a von Neumann algebra (as the commutant of a self-adjoint subset of $\lo(\h)$) \cite{Bratteli1998}.
\end{proof}

If $\ff(\Phi^*)$ is a von Neumann algebra, it holds that for any  self-adjoint operator $A \in \ff(\Phi^*)$,  the spectral measure of $A$ is also contained in $\ff(\Phi^*)$. In the case that $A$ has a discrete spectrum, i.e.,  $A = \sum_n \lambda_n P_n$, this implies that   $\{P_n\} \subset \ff(\Phi^*)$.

\section{Fixed points of instrument channels}\label{app:instrument-fixed-points}

Here we prove a useful result regarding the fixed-point structure of the $\E$-channel $\ii^*_\xx$, describing a non-selective measurement of an observable $\E$,  which we shall use in several places in this paper.
\begin{lemma}\label{lemma:fixed-points-instrument}
Let  $\ii$ be an instrument compatible with an observable $\E$ acting in $\hs$. The following hold:
\begin{enumerate} [(i)]
    \item If $\E$ is sharp, then $\ff(\ii^*_\xx) \subset \ii_\xx^*(\lo(\hs)) \subset \E'$. 
    \item If $\ff(\ii^*_\xx)$ is a von Neumann algebra, then  $\ff(\ii^*_\xx) \subset \E'$.
    \item If $\ii_\xx$ fully conserves a self-adjoint operator $A \in \los(\hs)$, then $A \in \E'$.
\end{enumerate}

\end{lemma}
\begin{proof}

All $\E$-compatible instruments $\ii$ admit a measurement scheme  $\mm:= (\ha, \xi, \ee, \Z)$. Therefore, by the channel $\Gamma_\xi^\ee$ defined in \eq{eq:Gamma-U}, we may write  $[\E(x), \ii^*_\xx(A)] = [\Gamma_\xi^\ee(\onesys \otimes \Z(x)), \Gamma_\xi^\ee(A \otimes \oneapp)]$. Since $[\onesys \otimes \Z(x), A \otimes \oneapp]=\zero$, and $\Z(x)$ are positive operators,  then by the sesquilinear mapping $\langle \langle A|B \rangle \rangle :=  \Gamma_\xi^\ee(A^*B) - \Gamma_\xi^\ee(A^*) \Gamma_\xi^\ee(B)$ and \corref{corollary:Norm-commutator-inequality} we obtain
\begin{align}\label{eq:lemma-fixed-point-intro-1}
\|[\E(x), \ii^*_\xx(A)] \|  &\leqslant  \| \<\< \onesys \otimes \Z(x)| \onesys \otimes \Z(x)\>\>\|^{\frac{1}{2}} \bigg( \|\<\<A \otimes \oneapp|A \otimes \oneapp\>\>\|^{\frac{1}{2}} + \|\<\<A^* \otimes \oneapp|A^* \otimes \oneapp\>\>\|^{\frac{1}{2}}\bigg).
\end{align}
Since $\Z(x)$ are effects, it follows that 
\begin{align*}
  \<\<\onesys \otimes \Z(x)|\onesys \otimes \Z(x)\>\> = \Gamma_\xi^\ee(\onesys \otimes \Z(x)^2) - \Gamma_\xi^\ee(\onesys \otimes \Z(x))^2 \leqslant \Gamma_\xi^\ee(\onesys \otimes \Z(x)) - \Gamma_\xi^\ee(\onesys \otimes \Z(x))^2=  \E(x) - \E(x)^2.  
\end{align*} 
On the other hand, we have    $\<\<A \otimes \oneapp|A \otimes \oneapp\>\> =  \ii^*_\xx(A^* A ) - \ii^*_\xx( A^*)\ii^*_\xx(A) $ and $\<\<A^* \otimes \oneapp|A^* \otimes \oneapp\>\> = \ii^*_\xx(A A^*) - \ii^*_\xx(A) \ii^*_\xx( A ^*)$.  We thus obtain from \eq{eq:lemma-fixed-point-intro-1} the bound
\begin{align}\label{eq:commutation-effect-fixed-point-sharp-or-algebra}
\|[\E(x), \ii^*_\xx(A)] \|   \leqslant \| \E(x) - \E(x)^2\|^{\frac{1}{2}}\bigg(   \|\ii^*_\xx(A^* A) - \ii^*_\xx(A^*)\ii^*_\xx(A) \|^{\frac{1}{2}} + \|\ii^*_\xx(A A^*) - \ii^*_\xx(A)\ii^*_\xx(A^*) \|^{\frac{1}{2}} \bigg).
\end{align}
Now we may prove (i). If $\E$ is sharp, then the upper bound of \eq{eq:commutation-effect-fixed-point-sharp-or-algebra} vanishes and so for all $A \in \lo(\hs)$, $\ii_\xx^*(A) \in \E'$. As such, $\ii_\xx^*(\lo(\hs))\subset \E'$.  That  $\ff(\ii_\xx^*) \subset \ii_\xx^*(\lo(\hs))$ is trivial.  

Now we prove (ii).  Assume that $A \in \ff(\ii^*_\xx)$, which implies that $A^* \in \ff(\ii^*_\xx)$. But if $\ff(\ii^*_\xx)$ is a von Neumann algebra, this implies that $A^*A, A A^*  \in \ff(\ii^*_\xx)$, and so the upper bound of  \eq{eq:commutation-effect-fixed-point-sharp-or-algebra} vanishes. Consequently, we see that   for all $A \in \lo(\hs)$, $A \in \ff(\ii_\xx^*) \implies A \in \E'$, which implies that $\ff(\ii_\xx^*) \subset \E'$. 

Finally, let us prove (iii). Let $A$ be a self-adjoint operator, and assume that $\ii_\xx$ fully conserves $A$. By \defref{defn:conservation-law}  it holds that $\ii_\xx^*(A^k) = A^k$ for $k=1,2$, and so once again the upper bound of \eq{eq:commutation-effect-fixed-point-sharp-or-algebra} vanishes, implying that $A \in \E'$.

\end{proof}

\section{Disturbance, commutation, and compatibility}\label{app:disturbance-compatibility}

The pair of observables $\E:=\{\E(x): x\in \xx\}$ and $\F:=\{\F(y) : y \in \yy\}$ acting in $\hs$ are \emph{compatible}, or \emph{jointly measurable},   if  they admit a joint observable  $\G := \{\G(x,y) : (x,y) \in \xx\times \yy\}$ so that
\begin{align}\label{eq:joint-observable-compatible}
     &\sum_{y \in \yy} \G(x,y) = \E(x), &\sum_{x\in \xx} \G(x,y) = \F(y) \qquad \forall \, x\in \xx, y \in \yy.
\end{align}
If $\E$ and $\F$ do not admit a joint observable, then they are \emph{incompatible} \cite{Heinosaari2015}. Now let $\ii$ be an $\E$-compatible instrument, and assume that  $\F \subset \ff(\ii^*_\xx)$. In such a case,  we may choose $\G$ as $\G(x,y) = \ii^*_x(\F(y))$, which satisfies \eq{eq:joint-observable-compatible}. It follows that non-disturbance implies compatibility, and so    for two incompatible observables $\E$ and $\F$,  no $\E$-instrument $\ii$ exists that satisfies $\F \subset \ff(\ii^*_\xx)$. Note that while non-disturbance requires compatibility, compatibility does not guarantee non-disturbance. 
For instance, while any observable is compatible with itself, for every informationally complete observable the fixed-point set of its compatible channel is trivial. Indeed,  the size of the fixed-point set of an $\E$-channel is strongly related to the amount of information given by $\E$
as shown in Ref. \cite{Hamamura-Miyadera}. Furthermore, 
as shown in Ref. \cite{Heinosaari2010}, there exist pairs of compatible observables $\E$ and $\F$ where $\E$ admits an instrument that does not disturb $\F$, but all possible $\F$-instruments necessarily disturb $\E$. This further demonstrates that unlike compatibility, non-disturbance is not symmetric.  

As shown in Ref. \cite{Miyadera2008}, the pair of observables $\E$ and $\F$ are compatible only if 
\begin{align}\label{eq:compatibility-necessary-commutativity}
    \|[\E(x), \F(y)] \| \leqslant 2  \|\E(x) - \E(x)^2 \|^{\frac{1}{2}} \| \F(y) - \F(y)^2 \|^{\frac{1}{2}} \qquad \forall \, x\in \xx, y \in \yy.
\end{align}
Commutation is a sufficient condition for compatibility; if $\E$ commutes with $\F$, then  there is  a joint observable $\G$ with effects $\G(x,y) = \E(x) \F(y) \equiv  (\sqrt{\E(x)} \sqrt{\F(y)})^* (\sqrt{\E(x)} \sqrt{\F(y)})$. On the other hand, if either $\E$ or $\F$ is sharp, in which case the upper bound of \eq{eq:compatibility-necessary-commutativity} vanishes, then commutation is a necessary condition for compatibility  \cite{Lahti2003}. For two non-commuting observables to be compatible, therefore, their effects must be sufficiently unsharp. We now provide a bound for  the  disturbance of $\F$ by an $\E$-instrument $\ii$, in terms of the commutation between the effects of $\E$ and $\F$.

\begin{prop}\label{prop:quantitative-bound-disturbance}
Consider the observables $\E$ and $\F$ acting in $\hs$, and let $\|\delta(y)\|$ be  the disturbance of the effects of $\F$ caused by  an $\E$-instrument $\ii$. Then for all $x\in \xx$ and $y \in \yy$ it holds that
\begin{align}\label{eq:disturbance-inequality}
\|  [\E(x), \F(y)] \| &\leqslant   \|\delta(y)\| +  2\| \E(x) - \E(x)^2 \|^{\frac{1}{2}} \| \ii^*_\xx(\F(y)^2) - \ii^*_\xx(\F(y))^2 \|^{\frac{1}{2}}.
\end{align}
 If $\F$ is non-disturbed by $\ii$, that is, if $\delta = 0$, then  for all $x\in \xx$ and $y \in \yy$ it holds that
\begin{align}\label{eq:non-disturbance-inequality}
   \|  [\E(x), \F(y)] \| &\leqslant 2\| \E(x) - \E(x)^2 \|^{\frac{1}{2}} \|\ii^*_\xx(\F(y)^2) - \F(y)^2 \|^{\frac{1}{2}}.
\end{align}
\end{prop}
\begin{proof}
 By \eq{eq:disturbance-quantification}, we may  write 
\begin{align}\label{eq:non-disturbance-equality-1}
    [\E(x), \F(y)] = [\delta(y), \E(x)] + [\E(x),\ii^*_\xx(\F(y))].
\end{align}
Every $\E$-instrument $\ii$ admits a measurement scheme $\mm:=(\ha, \xi, \ee, \Z)$.   Using the channel $\Gamma_\xi^\ee$ defined in \eq{eq:Gamma-U}, we may therefore write   $ [\E(x),\ii^*_\xx(\F(y))] = [\Gamma_\xi^\ee(\onesys \otimes \Z(x)), \Gamma_\xi^\ee(\F(y) \otimes \oneapp)]$. 
Given that $[\onesys \otimes \Z(x), \F(y) \otimes \oneapp]=\zero$, then by the sesquilinear mapping $\<\< A| B\>\> := \Gamma_\xi^\ee( A^* B) - \Gamma_\xi^\ee(A^*) \Gamma_\xi^\ee(B)$ and \corref{corollary:Norm-commutator-inequality}, we obtain from \eq{eq:non-disturbance-equality-1} the bound
\begin{align}\label{eq:non-disturbance-equality-2}
   \| [\E(x), \F(y)]\| &\leqslant \|[\delta(y), \E(x)]\| + 2 \|\<\< \onesys \otimes \Z(x) | \onesys \otimes \Z(x) \>\> \|^{\frac{1}{2}} \|\< \< \F(y) \otimes \oneapp | \F(y) \otimes \oneapp \>\> \|^{\frac{1}{2}}.
\end{align}
Since $\E(x)$ is an effect, then as shown in \lemref{lemma:unsharp-disturbance-bound} we have $ \|[\delta(y), \E(x)]\| \leqslant \| \delta(y) \|$. As shown in \lemref{lemma:fixed-points-instrument}, we have    $ \<\< \onesys \otimes \Z(x) | \onesys \otimes \Z(x) \>\> \leqslant   \E(x) - \E(x)^2$ and  $\< \< \F(y) \otimes \oneapp | \F(y) \otimes \oneapp \>\> =  \ii^*_\xx(\F(y)^2) - \ii^*_\xx(\F(y))^2$. We therefore obtain from \eq{eq:non-disturbance-equality-2} the bound given in \eq{eq:disturbance-inequality}. If $\F$ is non-disturbed by $\ii$, then $\|\delta(y)\| = 0$ and $\ii_\xx^*(\F(y))^2 = \F(y)^2$ for all $y$.  We thus arrive at   \eq{eq:non-disturbance-inequality}. 
\end{proof}

We see that when $\E$ commutes with $\F$ the lower bound of \eq{eq:non-disturbance-inequality} vanishes, in which case  \propref{prop:quantitative-bound-disturbance} does not prohibit non-disturbance. Indeed,  in the case of commuting observables there always exists a non-disturbing instrument; since $\E' \subset \ff({\ii^L_\xx}^*)$ always holds, where $\ii^L$ is the L\"uders $\E$-instrument defined in \eq{eq:luders}, then a L\"uders measurement of $\E$ is guaranteed not to disturb any $\F$ commuting with $\E$ \cite{Busch1998}.  On the other hand,  if $\E$ does not commute with $\F$, then \propref{prop:quantitative-bound-disturbance} allows us to obtain a lower bound for the disturbance that  results given any $\E$-compatible instrument, determined only by the unsharpness and non-commutation of $\E$ and $\F$:
\begin{corollary}\label{corollary:minimum-error-bound-sharpness}
Consider the setup of \propref{prop:quantitative-bound-disturbance}. For all $x\in \xx$ and $y\in \yy$, it also holds that
\begin{align}\label{eq:disturbance-inequality-observables}
\|  [\E(x), \F(y)] \| & \leqslant   \|\delta(y)\| +  2\| \E(x) - \E(x)^2 \|^{\frac{1}{2}} \bigg(2 \|\delta(y)\| + \|\F(y) - \F(y)^2 \|\bigg)^{\frac{1}{2}} .
\end{align}
\end{corollary}
\begin{proof}
Since $\F(y)$  are effects and $\ii^*_\xx$ is a channel, then by \lemref{lemma:unsharp-disturbance-bound} we have $\|\ii^*_\xx(\F(y)^2) - \ii^*_\xx(\F(y))^2 \| \leqslant 2 \|\delta(y)\| + \|\F(y) - \F(y)^2 \|$. As such,  \eq{eq:disturbance-inequality-observables} is obtained directly from \eq{eq:disturbance-inequality}.
\end{proof}
 Note that while \corref{corollary:minimum-error-bound-sharpness} provides a lower bound for the disturbance, which is strictly positive whenever either $\E$ or $\F$ is sharp and these observables do not commute,  such a lower bound will differ depending on whether $\E$ or $\F$ is sharp; if $\E$ is sharp, we have $\delta \geqslant \max_{x,y} \|[\E(x), \F(y)] \|$, whereas if $\F$ is sharp but $\E$ is unsharp, the lower bound for the disturbance may be smaller. Let us illustrate this with the following example. Consider a system  $\hs \simeq \co^2$, with the orthonormal basis $\{|0\>, |1\>\}$, and define $|\pm\> := \frac{1}{\sqrt{2}}(|0\> \pm |1\>)$. Now consider a pair of binary observables   $\A=\{\A(a): a=0,1\}$ and $\B_{\lambda}= \{\B_{\lambda}(b) : b = \pm\}$ acting in $\hs$,  defined by $\A(a)=|a\rangle \langle a|$ and $\B_{\lambda} (b)= \lambda |b\rangle \langle b | + (1-\lambda) \frac{\one}{2}$ for some $0\leqslant \lambda \leqslant 1$. It is simple to verify that $\|[\A(a), \B_{\lambda}(b)]\| = \frac{\lambda}{2}$ for any $a=0,1$ and $b=\pm$. Now we may evaluate the disturbance of one of these observables caused by a L\"uders measurement of the other.  The disturbance of $\B_\lambda(b)$ by a  L\"uders measurement of $\A$ reads  $\| \delta(b)\| = \frac{\lambda}{2}$ for each $b$. Since $\A$ is sharp, then by setting $\E = \A$ and $\F = \B_\lambda$, we see that the inequality in \eq{eq:disturbance-inequality-observables} is tight.  On the other hand, the disturbance of $\A(a)$ by a L\"uders measurement of $\B_\lambda$ reads  $\| \delta(a)\| = \frac{1-\sqrt{1-\lambda^2}}{2}$ for each $a$, which is  smaller than $\frac{\lambda}{2}$ for $0<\lambda <1$.

Let us now consider the case of non-disturbance more carefully. First, let us note  that when we set $\|\delta(y)\|=0$, \eq{eq:disturbance-inequality-observables}  reduces to the compatibility bound of \eq{eq:compatibility-necessary-commutativity}, and states that for non-disturbance to be possible when $\E$ and $\F$ do not commute, then both observables must be sufficiently unsharp so as to be compatible. To be sure, compatibility is a necessary condition for non-disturbance, and the fact that \eq{eq:disturbance-inequality-observables} does not contradict the compatibility bound is not surprising. On the other hand, in the case of non-disturbance this bound is also not very informative---it is possible for two observables to be compatible, while a measurement of one still disturbs the other. To gain a better understanding of non-disturbance, let us consider instead  \eq{eq:non-disturbance-inequality}, the upper bound of which is smaller than the upper bound in \eq{eq:disturbance-inequality-observables} when we set $\|\delta(y)\|=0$, and vanishes if both $\F \subset \ff(\ii^*_\xx)$ and $\F^2 := \{\F(y)^2 : y \in \yy\} \subset \ff(\ii^*_\xx)$ hold. We immediately see that while unsharpness of both $\E$ and $\F$ is necessary for non-disturbance when $\E$ and $\F$ do not commute, it is not sufficient; as shown in Ref. \cite{Heinosaari2010} there are at least two classes of unsharp observables $\F$ where given \emph{any} instrument $\ii$, i.e.,  including instruments that measure an unsharp observable $\E$ that does not commute with $\F$ but is still compatible with $\F$, it holds that $\F \subset \ff(\ii^*_\xx)$ guarantees    $\F^2 \subset \ff(\ii^*_\xx)$:  if $\F$ is a rank-1 observable, or   if $\F$ is an ``informationally equivalent  coarse-graining'' of a sharp observable. 

Let us consider the first option. If $\F$ is a rank-1 observable, then all the effects of $\F$ may be written as $\F(y) = \lambda_y P_y$, where $P_y$ is a rank-1 projection operator and $\lambda_y \in (0,1]$. As shown in \cite{Pellonpaa2014a}, all observables $\E$ that are compatible with a rank-1 observable $\F$ are the post-processings  of $\F$, that is, the effects of $\E$ may be written as $\E(x) = \sum_y p(x|y) \F(y)$, where $\{p(x|y)\}$ is a family of non-negative numbers satisfying $\sum_x p(x|y) =1$ for all $y$. It follows that so long as $\F$ is a non-commutative rank-1 observable, then there exists an unsharp observable $\E$ that is compatible with $\F$ but does not commute with $\F$.  But note that $\ii^*_\xx(\F(y)) = \F(y)$ if and only if $\ii^*_\xx(P_y) = P_y$. As such,   $\ii^*_\xx(\F(y)^2) = \lambda_y^2\ii^*_\xx(P_y) = \lambda_y^2 P_y = \F(y)^2$. It follows that $\F$ will be non-disturbed by  an $\E$-compatible instrument $\ii$  only if $\E$ commutes with $\F$. 

Let us now consider the second option. We say that $\F$ is an informationally equivalent coarse-graining of a sharp observable $\G:= \{\G(z)  : z \in \zz\}$ if there exists an invertible 
 stochastic matrix $M$   such that 
\begin{align*}
    &\F(y) = \sum_z M_{y, z} \G(z), & \G(z) = \sum_y M_{z, y}^{-1} \F(y).
\end{align*}
$\F$ and $\G$ are informationally equivalent because   a  measurement of $\F$ produces  different probability distributions for two states $\rho_1$ and $\rho_2$ if and only if these states produce different probability distributions  given a measurement of $\G$.
Since $\G$ is sharp, then  $\F(y)^2 = \sum_z M_{y,z}^2 \G(z)$. Now assume that $\F \subset \ff(\ii^*_\xx)$. It is simple to verify that this implies  $\G \subset \ff(\ii^*_\xx)$. Therefore, we have  $\ii^*_\xx(\F(y)^2) = \sum_z M_{y,z}^2 \ii^*_\xx(\G(z)) = \sum_z M_{y,z}^2 \G(z) = \F(y)^2$. Once again, $\F$ will be non-disturbed by an $\E$-compatible instrument $\ii$ only if $\E$ commutes with $\F$.  

Both of the above examples offer a very simple interpretation in terms of compatibility. If $\F$ is a rank-1 observable, then non-disturbance of $\F$ implies non-disturbance of sharp rank-1 effects $P_y$. Since non-disturbance requires compatibility, this implies that $\E$ must commute with all $P_y$, and hence with $\F$. 
On the other hand, if $\F$ is a classical coarse-graining of a sharp observable $\G$, then non-disturbance of $\F$ implies non-disturbance of $\G$, and by compatibility $\E$ must commute with $\G$.   Since the effects of  $\F$ are constructed as a mixture of the (projective) effects of $\G$, this concludes that $\E$ must commute with $\F$.

\section{Properties of repeatable instruments}\label{app:repeatable}

In this section, we prove a series of useful results regarding the structure of repeatable instruments, and the measurement schemes that implement them.

\begin{prop}\label{prop:repeatable-instrument-identity}
Let $\mm:=(\ha, \xi, \ee, \Z)$ be a measurement scheme for an $\E$-compatible instrument $\ii$ acting in $\hs$. If  $\ii$ is repeatable, then the following hold: 
\begin{enumerate}[(i)]
    %     \item For all $x\in \xx$ and $A\in \lo(\hs)$, it holds that $\ii_x^*(A) = \ii_x^*(\E(x) A) = \ii_x^*(A \E(x)) = \ii_x^*(\E(x) A \E(x))$.
    %  \item For all $x\in \xx$ and $A\in \lo(\hs)$, it holds that $\ii_\xx^*(\E(x) A) = \ii_\xx^*(A \E(x)) = \ii_\xx^*(\E(x) A \E(x)) = \ii_x^*(A)$.
     \item For all $x\in \xx$ and $n \in \nat$, it holds that  $\E(x) = \Gamma_\xi^\ee(\E(x)^n \otimes \oneapp) =\Gamma_\xi^\ee(\onesys \otimes \Z(x)^n)$.
    \item For all $x\in \xx$, it holds that $\E(x)$ and  $ \Z(x)$ have $1$ as an eigenvalue, and so  there exist projection operators $\P(x) \in \lop(\hs)$ and $\Q(x) \in \lop(\ha)$ which project onto the eigenvalue-1 eigenspaces of $\E(x)$ and $\Z(x)$, respectively. 
    
    \item For all $x, y \in \xx$, it holds that $\P(x) \E(y) = \P(x) \P(y) = \delta_{x,y} \P(x)$ and $\Q(x) \Z(y) = \Q(x) \Q(y) = \delta_{x,y} \Q(x)$.

\item For all $x \in \xx$ and $A \in \lo(\hs\otimes \ha)$, it holds that $\Gamma_\xi^\ee(\P^c(x) \otimes \oneapp A) = \Gamma_\xi^\ee(A \P^c(x) \otimes \oneapp)= \zero$ and $ \Gamma_\xi^\ee(\onesys \otimes \Q^c(x)A) = \Gamma_\xi^\ee(A\onesys \otimes \Q^c(x)) = \zero$, where  $\P^c(x) := \E(x) - \P(x)$ and $\Q^c(x) := \Z(x) - \Q(x)$.

 \item For all $A \in \lo(\hs\otimes \ha)$, it holds that $\Gamma_\xi^\ee(A) =  \Gamma_\xi^\ee(\R  A \R )$,  where $\R  := \sum_{x\in \xx} \P(x)\otimes \Q(x)$.

 \item For all $x\in \xx$, $A \in \lo(\hs)$, and $B \in \lo(\ha)$, it holds that $\ii_\xx^*(\E(x) A)  = \ii_\xx^*(A \E(x)) = \ii_\xx^*(\P(x) A \P(x))$ and  $\Lambda^*(\Z(x) B) = \Lambda^*(B \Z(x) ) =  \Lambda^*(\Q(x) B \Q(x))$.
    
 \item For all $x \in \xx$ and $A \in \lo(\hs)$, it holds that $\ii_x^*(A) =  \ii_x^*(\P(x) A \P(x))$.

\end{enumerate}
\end{prop}
\begin{proof}
\begin{enumerate}[(i):]

% \item For each $x$, define $\E(x)^\perp :=  \onesys - \E(x)$. Note that by \eq{eq:measurement-repeatable}, repeatability implies that $\ii_x^*(\E(x)^\perp) = \zero$.  By \lemref{lemma:operation-annihilation} it holds that $\ii_x^*(\E(x)^\perp A) = \ii_x^*(A \E(x)^\perp) = \zero$ for all $A \in \lo(\hs)$. The claim immediately follows by noting that we may write $A = (\E(x) + \E(x)^\perp) A = A (\E(x) + \E(x)^\perp) = (\E(x) + \E(x)^\perp) A (\E(x) + \E(x)^\perp)$. 

% \item For each $x$, we may write $\ii_\xx^*(\cdot) = \ii_x^*(\cdot) + \ii_{\overline{x}}^*(\cdot)$, where $\overline{x}  := \xx \backslash \{x\}$ is the complement of $x \equiv \{x\}$ in $\xx$, and it holds that $\ii_{\overline{x}}^*(\onesys) = \E(\overline{x}) \equiv  \E(x)^\perp$. Moreover, note that by the same arguments as (i),  for all $A\in \lo(\hs)$ it holds that $\ii_{\overline{x}}^*(\E(x) A) = \ii_{\overline{x}}^*(\E(\overline{x})^\perp A) = \zero$. We thus obtain the following:
% \begin{align*}
% \ii_\xx^*(\E(x) A) &= \ii_x^*(\E(x) A) + \ii_{\overline{x}}^*(\E(x) A)=\ii_x^*(\E(x) A) = \ii_x^*(A).
% \end{align*}
% Similarly, it holds that $\ii_\xx^*( A \E(x)) = \ii_\xx^*(\E(x) A \E(x)) = \ii_x^*(A)$. The claim immediately follows. 

\item The repeatability condition implies that for all $x\in \xx$, it holds that $\E(x) = \ii_x^*(\E(x)) = \Gamma_\xi^\ee(\E(x)\otimes \Z(x))$. It follows that for any state $\rho \in \s(\hs)$, we have 
\begin{align*}
\tr[\rho \E(x)] &=\tr[\rho \Gamma_\xi^\ee(\E(x)\otimes \Z(x))] \\
&\leqslant  \tr[\rho \Gamma_\xi^\ee(\E(x)^2\otimes \oneapp)]^{\frac{1}{2}} \tr[\rho\Gamma_\xi^\ee(\onesys\otimes \Z(x)^2)]^{\frac{1}{2}} \\
&\leqslant  \tr[\rho\Gamma_\xi^\ee(\E(x)\otimes \oneapp)]^{\frac{1}{2}} \tr[\rho\Gamma_\xi^\ee(\onesys \otimes \Z(x))]^{\frac{1}{2}} \\
&=\tr[\rho \E(x)]. 
\end{align*}
Here, the second line follows from the Cauchy-Schwarz inequality,  the 
third line follows from the fact that $\E(x)$ and $\Z(x)$ are effects and so $\E(x)^2 \leqslant \E(x)$ and $\Z(x)^2 \leqslant  \Z(x)$, and the final line follows from the fact that repeatability implies first-kindness and that $\mm$ is a measurement scheme for $\E$. As the second inequality must be an equality, we thus have  $\E(x) = \Gamma_\xi^\ee(\E(x)^n \otimes \oneapp) =\Gamma_\xi^\ee(\onesys \otimes \Z(x)^n)$ for $n=1,2$. To show that the relations hold for all $n\in \nat$, it suffices to show that for all $\rho$, the Cauchy-Schwarz inequality and the above arguments implies
\begin{align*}
0 \leqslant \tr[\rho\Gamma_\xi^\ee((\E(x)^n - \E(x)^{n+1})\otimes \oneapp )] 
&\leqslant  \tr[\rho\Gamma_\xi^\ee(\E(x)^{2(n-1)}\otimes \oneapp)]^{\frac{1}{2}}\tr[\rho\Gamma_\xi^\ee((\E(x) -\E(x)^2)^2 \otimes \oneapp)]^{\frac{1}{2}} \\
&\leqslant  \tr[\rho\Gamma_\xi^\ee (\E(x)^{2(n-1)}\otimes \oneapp)]^{\frac{1}{2}} \tr[\rho \Gamma_\xi^\ee ((\E(x)-\E(x)^2)\otimes \oneapp)]^{\frac{1}{2}} = 0, 
\end{align*}
and so it holds that $\Gamma_\xi^\ee((\E(x)^n - \E(x)^{n+1})\otimes \oneapp )=\zero$. Similar steps show that  $\Gamma_\xi^\ee(\onesys \otimes (\Z(x)^n - \Z(x)^{n+1})) = \zero$. The claims are thus obtained by induction.

\item Note that for any operation $\Phi^* : \lo(\kk) \to \lo(\h)$, it holds that $\|\Phi^*(A)\| \leqslant \| A\|$ for all $A \in \lo(\kk)$. As such, by (i) we have   
$\|\E(x) \| = \|\Gamma_\xi^\ee(\E(x)^2\otimes \oneapp)\| \leqslant \| \E(x)^2\| = \| \E(x) \|^2$. But since $\E(x)$ is an effect it also holds that  $\| \E(x)\| \geqslant  \| \E(x)\| ^2$. It follows that  $\|\E(x)\|$ is either zero or one. As we assume that $\E(x)$ is not vanishing, then $\| \E(x)\| =1$ follows. Similarly, we have $1 = \| \E(x) \| = \|\Gamma_\xi^\ee(\onesys \otimes \Z(x))\| \leqslant \| \Z(x)\|$, and since $\Z(x)$ is an effect, then it must hold that $\|\Z(x)\|=1$. 

Now we shall show that $\E(x)$ has $1$ as an eigenvalue, i.e., there exists a unit-vector $\psi \in \hs$ such that $\E(x) \psi = \psi$. If this is not so, then we would have  $\lim_{n \to \infty} \E(x)^n = \zero $, which would contradict (i). Therefore, there exists a projection operator $\P(x)$ that projects onto  the eigenvalue-1 eigenspace of  $\E(x)$. Similar arguments hold for $\Z(x)$ and $\Q(x)$.

\item 
For each $x$, define $\P^c(x) := \E(x) - \P(x)$. Since $\E(x)$ is an effect and $\P(x)$ projects onto the eigenvalue-1 eigenspace of $\E(x)$,  it trivially holds that $\psi \in \supp(\P(x)) \implies \psi \in \ker(\P^c(x))$. Now, given that  $\psi \in \supp(\P(x))$ implies that $\P(x) \psi = \psi$, and denoting  the null vector in $\hs$ as $\emptyset$, we have
\begin{align*}
   \emptyset& =  (\onesys - \P(x)) \psi \nonumber \\
   &= (\onesys - \E(x)) \psi + \P^c(x) \psi \nonumber \\
    & =(\onesys - \E(x)) \psi  \nonumber \\
    & = \sum_{y \ne x} \E(y) \psi .
\end{align*}
 By positivity of $\E(y)$, the above equation implies that 
\begin{align*}
    \sum_{y \ne x} \<\psi| \E(y) \psi\> = \sum_{y \ne x} \<\sqrt{\E(y)} \psi| \sqrt{\E(y)} \psi\> = 0,
\end{align*}
which can  be satisfied only if $\sqrt{\E(y)} \psi = \emptyset \implies \E(y) \psi = \emptyset$ for all $y \ne x$.  We thus have $\psi \in \supp(\P(x)) \implies \psi \in \ker(\E(y)) \, \forall \, y \ne x$, and so the support of $\P(x)$ must be orthogonal to the support of $\E(y)$ for all $y \ne x$. That $\P(x)$ and $\P(y)$ for $x\ne y$ have orthogonal supports follows trivially. Similar arguments hold for $\Q(x)$, $\Z(y)$,  and $\Q(y)$.

\item  It trivially  holds that   $\| \P^c(x)\| < 1$ and $\| \Q^c(x) \| <1 $. We thus have $\lim_{n \to \infty } \P^c(x)^n = \zero$ and  $\lim_{n \to \infty } \Q^c(x)^n = \zero$. As stated in (iii), the supports of $\P(x)$ and $\P^c(x)$ are orthogonal, and so it holds that $\P(x) \P^c(x) = \P^c(x) \P(x) = \zero$. As such,  for all $n\in \nat$ we have $\E(x)^n = \P(x) + \P^c(x)^n$. Similarly, $\Z(x)^n = \Q(x) + \Q^c(x)^n$.  By (i), it holds that $\Gamma_\xi^\ee(\E(x)^n \otimes \oneapp) = \Gamma_\xi^\ee(\P(x) \otimes \oneapp) + \Gamma_\xi^\ee(\P^c(x)^n \otimes \oneapp) = \E(x)$ for all $n \in \nat$, and so it must hold that $\Gamma_\xi^\ee(\P^c(x) \otimes \oneapp) = \lim_{n \to \infty } \Gamma_\xi^\ee(\P^c(x)^n \otimes \oneapp) =\zero$. Similarly, $\E(x) = \Gamma_\xi^\ee(\onesys \otimes \Q(x)) + \Gamma_\xi^\ee(\onesys \otimes \Q^c(x)^n)$ for all $n \in \nat$ implies that $\Gamma_\xi^\ee(\onesys \otimes \Q^c(x)) = \zero$. The claim follows from  \lemref{lemma:operation-annihilation}.

\item First, let us note that by repeatability, it holds that $\Gamma_\xi^\ee(\E(x) \otimes \Z(y)) = \ii_y^*(\E(x)) = \zero$ for all $x\ne y$. By \lemref{lemma:operation-annihilation}, it follows that $\Gamma_\xi^\ee(\E(x) \otimes \Z(y) A) = \Gamma_\xi^\ee(A\E(x) \otimes \Z(y)) = \zero$ holds for all $A$ and $x \ne y$. We may therefore write
\begin{align*}
\Gamma_\xi^\ee(A) &= \sum_{x,x',y,y'} \Gamma_\xi^\ee(\E(x)\otimes \Z(x') A \E(y) \otimes \Z(y')) \\
& = \sum_{x,y} \Gamma_\xi^\ee(\E(x)\otimes \Z(x) A \E(y)\otimes \Z(y)) \\
&= \sum_{x,y} \Gamma_\xi^\ee(\P(x)\otimes \Q(x) A \P(y)\otimes \Q(y)) = \Gamma_\xi^\ee(\R  A \R ).    
\end{align*}
In the final line, we have used the fact that $\E(x) = \P(x) + \P^c(x)$ and $\Z(x) = \Q(x) + \Q^c(x)$, together with (iv).

\item We may write
\begin{align*}
\ii_\xx^*(\E(x) A) & =  \Gamma_\xi^\ee(\R  ( \E(x) A \otimes \oneapp ) \R ) \\
&= \Gamma_\xi^\ee(\P(x) A \P(x) \otimes \Q(x)) \\
& = \Gamma_\xi^\ee(\P(x) A \P(x) \otimes \onesys) - \Gamma_\xi^\ee(\P(x) A \P(x) \otimes \Q(x)^\perp) \\
& = \Gamma_\xi^\ee(\P(x) A \P(x) \otimes \onesys) = \ii_\xx^*(\P(x) A \P(x)).
\end{align*}
The first line follows from (v), and the third line follows from the definition $\Q(x)^\perp := \onesys - \Q(x)$. The final line is obtained by (v) and noting that $(\P(x) A \P(x) \otimes \Q(x)^\perp) \R  = \P(x) A \P(x) \otimes \Q(x)^\perp \Q(x) = \zero$. The relation $\ii_\xx^*(A \E(x)) = \ii_\xx^*(\P(x) A \P(x))$ holding for all $A$ trivially follows from above and by observing that $\ii_\xx^*(\E(x) A^*)^* = \ii_\xx^*(A\E(x))$ and $\ii_\xx^*(\P(x) A^* \P(x))^* = \ii_\xx^*(\P(x) A \P(x))$.

Similarly, we may write
\begin{align*}
\Lambda^*(\Z(x) B) & =  \Gamma_\xi^\ee(\R  (\onesys \otimes \Z(x) B) \R ) \\
&= \Gamma_\xi^\ee(\P(x)  \otimes \Q(x) B \Q(x)) \\
& =  \Gamma_\xi^\ee(\onesys  \otimes \Q(x) B \Q(x)) - \Gamma_\xi^\ee(\P(x)^\perp  \otimes \Q(x) B \Q(x)) \\
& = \Gamma_\xi^\ee(\onesys  \otimes \Q(x) B \Q(x)) = \Lambda^*(\Q(x) B \Q(x)).
\end{align*}

\item We may write 
\begin{align*}
\ii_x^*(A) &=  \Gamma_\xi^\ee(\R (A \otimes  \Z(x) ) \R ) \\
& = \Gamma_\xi^\ee(\P(x) A \P(x) \otimes \Q(x)) \\
& = \Gamma_\xi^\ee(\P(x) A \P(x) \otimes \Z(x)) - \Gamma_\xi^\ee(\P(x) A \P(x) \otimes \Q^c(x)) \\
& = \Gamma_\xi^\ee(\P(x) A \P(x) \otimes \Z(x)) = \ii_x^*(\P(x) A \P(x)).
\end{align*}
In the first line we have used (v), in the third line we use $\Q(x) = \Z(x) - \Q^c(x)$, and in the final line we use (iv).

\end{enumerate}
\end{proof}

Let us highlight one interesting property of repeatable instruments: if $\ii$ is repeatable, then for all input states $\rho$, the output states will be perfectly distinguishable. For the input state $\rho$, we define the normalised post-measurement states as $\rho_x:= \ii_x(\rho)/\tr[\ii_x(\rho)]$ for any $x$ satisfying $\tr[\ii_x(\rho)]>0$. By item (vii) of the above proposition,  the Schr\"odinger picture operations of a repeatable instrument satisfy $\ii_x(T) = \P(x) \ii_x(T) \P(x)$ for all $x$ and $T \in \trc(\hs)$, and so $\rho_x$ will  have support only in the eigenvalue-1 eigenspace of $\E(x)$. But by item (iii), such eigenvalue-1 eigenspaces are orthogonal, and so it holds that $\rho_x \rho_y = \rho_y \rho_x = \zero$ for all $x\ne y$.

\section{Conservation laws}\label{app:conservation-law}

Recall that by  \defref{defn:average-conservation}  a channel $\Phi$ conserves $N$ on average if $\Phi^*(N)=N$, while by \defref{defn:conservation-law} $\Phi$ fully conserves $N$ if $\Phi^*(N^k)=N^k$ for all $k\in \nat$. We shall now show that  full conservation is in fact equivalent to just the first two moments being conserved, and that is is also equivalent to ``invariance'' of the unitary group generated by $N$ under the action of $\Phi^*$, i.e., that $\Phi^*(e^{\imag t N}) = e^{\imag t N}$ for all $t \in \re$.

\begin{prop}\label{prop:conservation-multiplication}
Let $\Phi : \trc(\h) \to \trc(\h)$ be a channel, and let $N \in \los(\h)$ be a self-adjoint operator. The following statements are equivalent:
\begin{enumerate}[(i)]
    \item $\Phi$ fully conserves $N$.
    \item $\Phi^*(N^k) = N^k$ for $k=1,2$.
    \item $\Phi^*(e^{\imag t N}) = e^{\imag t N}$ for all $t\in \re$.
\end{enumerate}
\end{prop}  
\begin{proof}

(i) $\implies$ (ii): This is trivial. 

(ii) $\implies$ (i):  Assume that $\Phi^*(N^k) = N^k$ for $k=1,2$.   For any $k \geqslant 2$,  \corref{corollary:multiplicability} implies that  $\Phi^*(N^{k+1}) = \Phi^*(N^{k} N)  = \Phi^*(N^k) N$.   The claim follows  by induction. 

(ii) $\implies$ (iii): Define $f(t):= \Phi^*(e^{\imag t N})$. Since $N$ is bounded,  $e^{\imag t N}$ is bounded and strongly continuous, and $\Phi^*$ is a channel, then $f(t)$ is infinitely differentiable. Now assume that $\Phi^*(N^k) = N^k$ for $k=1,2$. It follows that
\begin{align*}
 \frac{d}{dt}f(t) = \imag  \Phi^*( N e^{\imag t N}) = \imag N\Phi^*(e^{itN}) = \imag N f(t),  
\end{align*}
where the second equality follows from   \corref{corollary:multiplicability}. Indeed, by induction we obtain $(d^k/ dt^k) f(t) = \imag ^k N^k f(t)$ for all $k \in \nat$. Since $f(0) = \one$, by Taylor expansion around $t=0$ we observe that
\begin{align*}
\Phi^*(e^{\imag t N}) = \sum_{k=0}^\infty \frac{t^k}{k !} \left(\frac{d^k}{dt^k}f(t)\right) \bigg|_{t=0} =  \sum_{k=0}^\infty \frac{\imag^k t^k}{k !}  N^k =  e^{\imag t N}.
\end{align*}

(iii) $\implies$ (i): Assume that $\Phi^*(e^{\imag t N}) = e^{\imag t N}$ for all $t\in \re$. It follows that 
\begin{align*}
\frac{d^k}{dt^k}\Phi^*(e^{\imag t N}) =  \frac{d^k}{dt^k} e^{\imag t N} \implies  \Phi^*(N^k e^{\imag t N}) =  N^k e^{\imag t N} 
\end{align*}
holds for all $k\in \nat$ and all $t$. Since $e^{\imag t N} =\one$ when $t=0$, it follows that   $\Phi^*(N^k) = N^k$ for all $k$. 
\end{proof}  

A property that channels may enjoy is   ``covariance'' under the action of a unitary group, i.e., that $\Phi^*(e^{\imag t N} A e^{-\imag t N}) = e^{\imag t N} \Phi^*(A) e^{-\imag t N}$ holds for all $t\in \re$ and $A\in \lo(\h)$. While a channel may be coavariant while not invariant---for example, $\Phi^*(\cdot) = \tr[\cdot \omega] \one$ such that $[\omega, N]=\zero$ is covariant but not invariant---we now show that invariance implies covariance.

\begin{corollary}
Let $\Phi : \trc(\h) \to \trc(\h)$ be a channel, and let $N \in \los(\h)$ be a self-adjoint operator. If $\Phi$ fully conserves $N$, then for all $t\in \re$ and $A\in \lo(\h)$ it holds that 
\begin{align*}
    \Phi^*(e^{\imag t N} A e^{-\imag t N}) = e^{\imag t N} \Phi^*(A) e^{-\imag t N}
\end{align*}
\end{corollary}
\begin{proof}
Let us define $V(t) := e^{\imag t N}$ for notational simplicity. By \propref{prop:conservation-multiplication}, full conservation of $N$ by $\Phi$ implies that $\Phi^*(V(t)) = V(t)$ for all $t$. In particular, noting that $V(t)^* = V(-t)$,  this implies that $\Phi^*(V(t) ^* V(t)) - \Phi^*(V(t)^*) \Phi^*(V(t)) = \Phi^*(\one) - \one = \zero$. As such, by \corref{corollary:multiplicability} it follows that 
\begin{align*}
    \Phi^*(V(t) A V(-t)) = V(t) \Phi^*( A V(-t)) = V(t) \Phi^*( A ) V(-t).
\end{align*}
\end{proof}

We note that the condition $\Phi^*(N^k) = N^k$ for $k=1,2$ was taken as a potential definition of conservation \emph{simpliciter} in Ref. \cite{Luo2007}. However, the authors here  conjectured that, in finite dimensions,  the condition $\Phi^*(N^2) = N^2$ may be dropped, and that (in our formulation) both average and full conservation are  equivalent. We shall now address this issue: by a simple counter-example, we shall show that average and full conservation are in fact  not equivalent for general channels, even in finite dimensions. 

Let us consider a system   $\h \simeq \co^3$ with an orthonormal basis $\{|-1\>, |0\>,
|1 \>\}$. Now consider $N=\sum_n n |n\rangle \langle n| \equiv  |1\rangle \langle 1| - |-1\rangle \langle -1|$, and  a channel $\Phi^*$ defined by 
\begin{align*}
\Phi^*(A) = \<1|A|1\> |1\> \<1| +\< -1|A|-1\> |-1\>\<-1| + \< +|A|+\>|0\>\< 0|,
\end{align*}
to hold for all $A \in \lo(\h)$, where we define $|+\> := \frac{1}{\sqrt{2}}(|1\> + |-1\>)$. It is simple to verify that   $\Phi^*(N)=N$, that is, $\Phi$  conserves $N$ on average. However,  $\Phi^*(N^2) = \one \ne N^2$, and so  $\Phi$ does not fully conserve $N$. Since full conservation is equivalent to invariance, then it follows that $\Phi$ is also not invariant. Indeed, we can easily verify that $\Phi$ is not covariant either;  for example, if we choose $A = |+\>\<+|$, then it holds that
\begin{align*}
 \Phi^*(e^{\imag t N} |+\>\<+| e^{-\imag t N}) &= \frac{1}{2} (|1\> \<1| + |-1\>\<-1|) + \cos(t)^2 |0\>\<0|, \\
 e^{\imag t N}\Phi^*( |+\>\<+| )e^{-\imag t N} &= \frac{1}{2} (|1\> \<1| + |-1\>\<-1|) + |0\>\<0|,
\end{align*}
which  coincide only when $t$  is  an integer multiple of $\pi$.

While the above discussion shows that  average conservation is in general a weaker condition than full conservation,  we shall now show that in the special case of unitary channels, average  and full conservation are equivalent: 

\begin{lemma}\label{lemma:conservation-unitary}
Let $\Phi(\cdot) := U(\cdot) U^*$ be a unitary channel, with $U \in \lo(\h)$ a unitary operator, and let $N\in \los(\h)$ be a self-adjoint operator. The following statements are equivalent:
\begin{enumerate}[(i)]
    \item $[U, N]=\zero$.
    \item $\Phi$ conserves $N$ on average.
    \item $\Phi$ fully conserves $N$.
\end{enumerate}
\end{lemma}
\begin{proof}
(i) $\iff$ (ii), (i) $\implies$ (iii),  and (iii) $\implies$ (ii) are trivial. To  show (iii) $\implies$ (i), let us first write 
\begin{align*}
    [U, N]^*[U,N] = \Phi^*(N^2) + N^2 - \Phi^*(N)N - N \Phi^*(N).
\end{align*}
If $\Phi$ fully conserves $N$, then the right hand side vanishes. But since the left hand side is a positive operator, then it holds that $[U, N]=\zero$. Finally, we shall show that (ii) $\implies$ (iii). Since $\Phi$ is unitary, then  $\Phi^*(A^*) \Phi^*(B) = U^* A^* U U^* B U = U^* A^* B U = \Phi^*(A^* B)$ holds for all $A, B\in \lo(\h)$. If $\Phi^*(N) = N$, it follows that  $\Phi^*(N^2) = N^2$. The claim follows from \propref{prop:conservation-multiplication}. 
\end{proof}

\section{Bounds for measurement error under conservation laws}\label{app:error-bounds}

Here we provide quantitative trade-off relations for measurement error under additive conservation laws, both average and full.  
\newpage
\begin{theorem}\label{theorem:measurability} 
Let $\mm:=(\ha, \xi, \ee, \Z)$ be a measurement scheme for an observable acting in $\hs$, and assume that $\ee$  conserves an additive quantity $N = \nsys \otimes \oneapp + \onesys \otimes \napp$ on average, where $\nsys \in \los(\hs)$ and $\napp \in \los(\ha)$. Let $\|\epsilon(x)\|$ be the error in measuring the effects of the target observable $\E$, as defined in \eq{eq:error-equality}. Then for all $x\in \xx$ it holds that
\begin{align}\label{eq:measurability-error-bound}
    \|[\E(x), \nsys] - \Lambda^*([\Z(x), \napp]) \| \leqslant 2  \|\nsys\|\|\epsilon(x)\| + 2 \|\Gamma_\xi^\ee(N^2) - \Gamma_\xi^\ee(N)^2 \|^{\frac{1}{2}} \bigg(2 \|\epsilon(x) \| + \| \E(x) - \E(x)^2 \|\bigg)^{\frac{1}{2}},
\end{align}
where $\Gamma_\xi^\ee$ is the channel defined in \eq{eq:Gamma-U}, and $\Lambda$ is the conjugate channel to $\ii_\xx$ defined in \eq{eq:conjugate-channel}.
\end{theorem}

\begin{proof}
By \eq{eq:error-equality}, we have
\begin{align*}
    [\E(x), \nsys] - \Lambda^*([\Z(x), \napp]) = [\nsys, \epsilon(x)] + [\Lambda^*(\Z(x)), \nsys] - \Lambda^*([\Z(x), \napp]).
\end{align*}
 Since  $N$ is additive, average conservation of $N$ by $\ee$ implies that $\Gamma_\xi^\ee(N) = \Gamma_\xi(N) = \nsys + \tr[\napp \xi] \onesys$. It follows that  for all $A \in \lo(\hs)$,  $[A, N\sub{\s} ] = [A, \Gamma_\xi^\ee(N)]$, and so we may write $[\Lambda^*(\Z(x)), \nsys] = [\Gamma_\xi^\ee(\onesys \otimes \Z(x)), \Gamma_\xi^\ee(N)]$. Moreover, by additivity of $N$ we have   $\Gamma_\xi^\ee([\onesys \otimes \Z(x), N]) = \Gamma_\xi^\ee(\onesys \otimes[ \Z(x), \napp])=\Lambda^*([\Z(x), \napp])$.  We may therefore write
\begin{align}\label{eq:error-quantity-equality-1}
    [\E(x), \nsys] - \Lambda^*([\Z(x), \napp]) = [\nsys, \epsilon(x)] + [\Gamma_\xi^\ee(\onesys \otimes \Z(x)), \Gamma_\xi^\ee(N)] - \Gamma_\xi^\ee([\onesys \otimes \Z(x), N]).
\end{align}
By the sesquilinear mapping $\<\<A|B\>\> = \Gamma_\xi^\ee(A^* B) - \Gamma_\xi^\ee(A^*)\Gamma_\xi^\ee(B)$ and \corref{corollary:Norm-commutator-inequality} we obtain from  \eq{eq:error-quantity-equality-1} the bound 
\begin{align*}
 \| [\E(x), \nsys] - \Lambda^*([\Z(x), \napp])\| &\leqslant \| [\nsys, \epsilon(x)]\|  +  2\|\<\< N | N \>\>\|^{\frac{1}{2}} \|\<\< \onesys \otimes \Z(x) | \onesys \otimes \Z(x) \>\>\|^{\frac{1}{2}}.
\end{align*}
Note that $\|[\nsys,  \epsilon(x) ] \|\leqslant 2 \|N\sub{\s} \|\|\epsilon(x) \| $, and  that  $\<\< N | N \>\> = \Gamma_\xi^\ee(N^2) - \Gamma_\xi^\ee(N)^2$ by definition. Given that   $\onesys\otimes \Z(x)$ and $\E(x)$ are effects and $\Gamma_\xi^\ee$ is a channel, \lemref{lemma:unsharp-disturbance-bound} gives $\|\<\< \onesys \otimes \Z(x) | \onesys \otimes \Z(x) \>\>\| \leqslant 2 \|\epsilon(x) \| + \| \E(x) - \E(x)^2 \|$. We thus obtain the bound given in \eq{eq:measurability-error-bound}.  
\end{proof}

\begin{prop}\label{prop:measurability-Fisher}
Let $\mm:=(\ha, \xi, \ee, \Z)$ be a measurement scheme for an observable acting in $\hs$, and assume that $\ee$ fully conserves an additive quantity $N = \nsys \otimes \oneapp + \onesys \otimes \napp$, where $\nsys \in \los(\hs)$ and $\napp \in \los(\ha)$. Let $\|\epsilon(x)\|$ be the error in measuring the effects of the target observable $\E$, as defined in \eq{eq:error-equality}. Then for all $x\in \xx$ it holds that
\begin{align}\label{eq:measurability-error-bound-Fisher}
    \|[\E(x), \nsys] - \Lambda^*([\Z(x), \napp]) \| \leqslant 2  \|\nsys\|\|\epsilon(x)\| + \frac{1}{2} \qq(\napp, \xi)^{\frac{1}{2}},
\end{align}
where $\qq(\napp, \xi)$ denotes the quantum Fisher information of $\napp$ in the state $\xi$. Additionally, if $\E$ is an extremal observable and $\mm$ is a measurement scheme for $\E$, then for all $x\in \xx$ it holds that
\begin{align}\label{eq:measurability-bound-Fisher-extremal}
    \|[\E(x), \nsys] - \Lambda^*([\Z(x), \napp]) \| \leqslant   \qq(\napp, \xi)^{\frac{1}{2}} \|\E(x) - \E(x)^2 \|^{\frac{1}{2}}.
\end{align}
\end{prop}
\begin{proof}
 Let $\{q_i, \phi_i\}$ be an arbitrary ensemble of unit vectors  that satisfies $\xi = \sum_i q_i \pr{\phi_i}$. We may thus write $\Gamma_\xi^\ee(\cdot) = \sum_i q_i \Gamma_{\phi_i}^\ee(\cdot)$, where $\Gamma_{\phi_i}^\ee(\cdot) \equiv \Gamma_{\pr{\phi_i}}^\ee(\cdot)$.   Given the additivity of $N$ and the  conservation law, we may rewrite \eq{eq:error-quantity-equality-1} as
 \begin{align*}
    [\E(x), \nsys] - \Lambda^*([\Z(x), \napp]) = [\nsys, \epsilon(x)] + \sum_i q_i \bigg([\Gamma_{\phi_i}^\ee(\onesys \otimes \Z(x)), \Gamma_{\phi_i}^\ee(N)] - \Gamma_{\phi_i}^\ee(\onesys \otimes[ \Z(x), \napp])\bigg).
\end{align*}
By the sesquilinear mappings $\<\<A|B\>\>_i := \Gamma_{\phi_i}^\ee(A^* B) - \Gamma_{\phi_i}^\ee(A^*) \Gamma_{\phi_i}^\ee(B)$,  \corref{corollary:Norm-commutator-inequality} and  \lemref{lemma:conservation-variance-condition},  we obtain the bounds
\begin{align*}
    \|[\E(x), \nsys] - \Lambda^*([\Z(x), \napp]) \|     & \leqslant 2\|\nsys\| \|\epsilon(x)\| + 2 \sum_i q_i \var{\napp, \phi_i} ^{\frac{1}{2}} \|\Gamma_{\phi_i}^\ee(\onesys \otimes \Z(x)^2) - \Gamma_{\phi_i}^\ee(\onesys \otimes \Z(x))^2\|^{\frac{1}{2}}.
\end{align*}
Since both $\Z(x)$ and $\Gamma_{\phi_i}^\ee(\onesys \otimes \Z(x))$ are effects, we obtain  
\begin{align*}
\|\Gamma_{\phi_i}^\ee(\onesys \otimes \Z(x)^2) - \Gamma_{\phi_i}^\ee(\onesys \otimes \Z(x))^2\|^{\frac{1}{2}} \leqslant \|\Gamma_{\phi_i}^\ee(\onesys \otimes \Z(x)) - \Gamma_{\phi_i}^\ee(\onesys \otimes \Z(x))^2\|^{\frac{1}{2}} \leqslant 1/2 .   
\end{align*}
We thus arrive at the bound 
\begin{align*}
    \|[\E(x), \nsys] - \Lambda^*([\Z(x), \napp]) \|     & \leqslant 2\|\nsys\| \|\epsilon(x)\| +  \sum_i q_i \var{\napp, \phi_i}^{\frac{1}{2}} \\
    &\leqslant 2\|\nsys\| \|\epsilon(x)\| + \left( \sum_i q_i \var{\napp, \phi_i} \right)^{\frac{1}{2}},
\end{align*}
where the second line follows from the concavity of the square root. By choosing the ensemble $\{q_i, \phi_i\}$ that gives the quantum Fisher information as in \eq{eq:QFI-defn},  we arrive at the bound in \eq{eq:measurability-error-bound-Fisher}.

Now assume that $\E$ is an extremal observable \cite{DAriano2005}. This implies that for any pair of observables $\E^{(1)}$ and $\E^{(2)}$, and any $\lambda \in (0,1)$, the effects of $\E$ can be decomposed as $\E(x) = \lambda \, \E^{(1)}(x) + (1-\lambda) \, \E^{(2)}(x)$ only if $\E = \E^{(1)} = \E^{(2)}$. It follows that if $\mm$ is a measurement scheme for $\E$, that is, if $\epsilon = 0$,  then $\Gamma_{\phi_i}^\ee(\onesys \otimes \Z(x)) = \E(x)$ for all $i$. Consequently, we obtain the bounds $\|\Gamma_{\phi_i}^\ee(\onesys \otimes \Z(x)^2) - \Gamma_{\phi_i}^\ee(\onesys \otimes \Z(x))^2\|^{\frac{1}{2}}  \leqslant \| \E(x) - \E(x)^2\|^{\frac{1}{2}}$ for all $i$, which gives 
\begin{align*}
    \|[\E(x), \nsys] - \Lambda^*([\Z(x), \napp]) \|     &  \leqslant 2\left( \sum_i q_i \var{\napp, \phi_i} \right)^{\frac{1}{2}}\| \E(x) - \E(x)^2\|^{\frac{1}{2}}.
\end{align*}
 Once again choosing the ensemble that gives the quantum Fisher information, we arrive at \eq{eq:measurability-bound-Fisher-extremal}. 

\end{proof}

\section{Bounds for measurement disturbance under conservation laws}\label{app:disturbance-bounds}

Here we provide quantitative trade-off relations for measurement disturbance under additive conservation laws, both average and full. Note that here, the observable that may or may not be disturbed is not necessarily the same observable that is measured by the instrument $\ii$. 

\begin{theorem} \label{theorem:quantitative-bound-disturbance-WAY}
 Let $\mm:=(\ha, \xi, \ee, \Z)$ be a measurement scheme for an instrument $\ii$ acting in $\hs$, and assume that $\ee$  conserves an additive quantity $N = N\sub{\s} \otimes \oneapp + \onesys\otimes N\sub{\aa}$ on average, where $\nsys \in \los(\hs)$ and $\napp \in \los(\ha)$. Let $\|\delta(y)\|$ be the disturbance of  the effects of an observable $\F:= \{\F(y) : y \in \yy\}$ caused by $\ii$, as defined in \eq{eq:disturbance-quantification}. Then for all  $y \in \yy$ it holds that
\begin{align}\label{eq:disturbance-inequality-WAY}
 \| [\F(y), \nsys] - \ii^*_\xx([\F(y), \nsys]) \| & \leqslant 2\|\nsys\| \|\delta(y)\| + 2 \|\Gamma_\xi^\ee(N^2) - \Gamma_\xi^\ee(N)^2 \|^{\frac{1}{2}} \|\ii^*_\xx(\F(y) ^2) - \ii^*_\xx(\F(y))^2  \|^{\frac{1}{2}},
\end{align}
where $\Gamma_\xi^\ee$ is the channel defined in \eq{eq:Gamma-U}.  If $\F$ is non-disturbed by $\ii$, that is, if $\delta = 0$, then for all $y \in \yy$ it holds that
\begin{align}\label{eq:non-disturbance-inequality-WAY}
\| [\F(y), \nsys] - \ii^*_\xx([\F(y), \nsys]) \| & \leqslant  2 \|\Gamma_\xi^\ee(N^2) - \Gamma_\xi^\ee(N)^2 \|^{\frac{1}{2}} \|\ii^*_\xx(\F(y) ^2) - \F(y)^2  \|^{\frac{1}{2}} .
\end{align}
\end{theorem}
\begin{proof}

 By \eq{eq:disturbance-quantification}, we may write
\begin{equation*}
   [\F(y) ,\nsys ] - \ii^*_\xx([\F(y), \nsys]) = [N\sub{\s} ,\delta(y) ] + [\ii^*_\xx(\F(y) ),N\sub{\s} ] - \ii^*_\xx([\F(y), \nsys]).
\end{equation*}
 Since  $N$ is additive, average conservation of $N$ by $\ee$ implies that $\Gamma_\xi^\ee(N) = \Gamma_\xi(N) = \nsys + \tr[\napp \xi] \onesys$. It follows that  for all $A \in \lo(\hs)$,  $[A, N\sub{\s} ] = [A, \Gamma_\xi^\ee(N)]$. We may therefore  write $[\ii^*_\xx(\F(y) ),N\sub{\s} ] = [\Gamma_\xi^\ee(\F(y) \otimes \oneapp), \Gamma_\xi^\ee(N)]$. Additionally, by additivity of $N$ we may write  $\Gamma_\xi^\ee([\F(y) \otimes \oneapp, N]) = \Gamma_\xi^\ee([\F(y) , N\sub{\s} ]\otimes \oneapp) = \ii^*_\xx([\F(y), \nsys])$. We thus arrive at 
 \begin{align}\label{eq:WAY-disturbance-equality-1}
 [\F(y) ,\nsys ] - \ii^*_\xx([\F(y), \nsys]) = [N\sub{\s} ,\delta(y) ] + [\Gamma_\xi^\ee(\F(y) \otimes \oneapp), \Gamma_\xi^\ee(N) ] - \Gamma_\xi^\ee([\F(y) \otimes \oneapp, N]).    
 \end{align}
By the sesquilinear mapping $\<\< A | B\>\> := \Gamma_\xi^\ee(A^* B) - \Gamma_\xi^\ee(A^*) \Gamma_\xi^\ee(B)$ and \corref{corollary:Norm-commutator-inequality} we obtain from  \eq{eq:WAY-disturbance-equality-1} the bound 
\begin{align}\label{eq:equality-disturbance-commutation}
   \|  [\F(y) ,N\sub{\s} ] - \ii^*_\xx([\F(y) , N\sub{\s} ])\|   & \leqslant   \|[N\sub{\s} ,\delta(y) ]\|   + 2\| \<\<N | N\>\> \|^{\frac{1}{2}} \|\<\<  \F(y) \otimes \oneapp |  \F(y) \otimes \oneapp\>\> \|^{\frac{1}{2}}.
\end{align}
By definition, we have $\<\< \F(y) \otimes \oneapp | \F(y) \otimes \oneapp \>\> =  \ii^*_\xx(\F(y) ^2) - \ii^*_\xx(\F(y) )^2$ and $\<\<N | N\>\> = \Gamma_\xi^\ee(N^2) - \Gamma_\xi^\ee(N)^2$. Noting that $\|[\nsys,  \delta(y) ] \|\leqslant 2 \|N\sub{\s} \|\|\delta(y) \| $,   we thus obtain from  \eq{eq:equality-disturbance-commutation} the bound given in \eq{eq:disturbance-inequality-WAY}. If $\F$ is non-disturbed by $\ii$, then   $\|\delta(y)\| = 0$ and $\ii^*_\xx(\F(y) )^2 = \F(y)^2$ for all $y$. We thus arrive at    \eq{eq:non-disturbance-inequality-WAY}. 
\end{proof}
\begin{corollary}\label{corollary:WAY-disturbance-bound-unsharpness}
Consider the set-up of \thmref{theorem:quantitative-bound-disturbance-WAY}. For all $y \in \yy$ it also holds that 
\begin{align*}
 \| [\F(y), \nsys] - \ii^*_\xx([\F(y), \nsys]) \| & \leqslant 2\|\nsys\| \|\delta(y)\| + 2 \|\Gamma_\xi^\ee(N^2) - \Gamma_\xi^\ee(N)^2 \|^{\frac{1}{2}} \bigg( 2 \|\delta(y)\| + \|\F(y) - \F(y)^2 \|\bigg)^{\frac{1}{2}}  .
\end{align*}
\end{corollary}
\begin{proof}
Since $\F(y)$ are effects and $\ii^*_\xx$ is a channel, then by \lemref{lemma:unsharp-disturbance-bound} we have $\|\ii^*_\xx(\F(y) ^2) - \ii^*_\xx(\F(y) )^2  \|\leqslant 2 \|\delta(y)\| + \|\F(y) - \F(y)^2 \|$. The claim immediately follows from \eq{eq:disturbance-inequality-WAY}.
\end{proof}

\begin{corollary}\label{corrollary:Luders-commuting-no-constraint}
If $\E$ commutes with $\nsys$, and is measured by the L\"uders instrument $\ii^L$, then  \thmref{theorem:quantitative-bound-disturbance-WAY} does not rule out non-disturbance of any observable $\F$ that commutes with $\E$.
\end{corollary}
\begin{proof}
Since $\E'\subset \ff({\ii^L_\xx}^*)$ always holds, then $[\E,\F]=\zero$ implies that $\F \subset \ff({\ii^L_\xx}^*)$. But  $[\E, \nsys] = \zero$ implies that   ${\ii^L_\xx}^*(\nsysk{k}) = \nsysk{k}$ for $k=1,2$. By the multiplicability theorem (\corref{corollary:multiplicability}), it holds that ${\ii^L_\xx}^*(\nsys A) = \nsys {\ii^L_\xx}^*(A)$ and ${\ii^L_\xx}^*(A \nsys) = {\ii^L_\xx}^*(A) \nsys$ for all $A\in \lo(\hs)$, and so  we obtain ${\ii^L_\xx}^*([\F(y), \nsys]) = [\F(y), \nsys]$. In such a case,    the lower bound of \eq{eq:non-disturbance-inequality-WAY} vanishes, and so no constraints are imposed on non-disturbance for $\F$. 
\end{proof}
 
\begin{prop}\label{prop:quantitative-bound-disturbance-WAY-Fisher}
 Let $\mm:=(\ha, \xi, \ee, \Z)$ be a measurement scheme for an instrument $\ii$ acting in $\hs$, and assume that $\ee$ fully conserves an additive quantity $N = N\sub{\s} \otimes \oneapp + \onesys\otimes N\sub{\aa}$, where $\nsys \in \los(\hs)$ and $\napp \in \los(\ha)$. Let $\|\delta(y)\|$ be the disturbance of  the effects of an observable $\F= \{\F(y) : y \in \yy\}$ caused by $\ii$, as defined in \eq{eq:disturbance-quantification}. Then for all  $y \in \yy$ it holds that
\begin{align}\label{eq:WAY-bound-Fisher}
 \| [\F(y), \nsys] - \ii^*_\xx([\F(y), \nsys]) \| & \leqslant 2\|\nsys\| \|\delta(y)\| + \frac{1}{2}\qq(\napp, \xi)^{\frac{1}{2}},
\end{align}
where $\qq(\napp, \xi)$ denotes the quantum Fisher information of $\napp$ in the state $\xi$. Additionally, if $\ii$ is an extremal instrument,  then  for all $y\in \yy$ it holds that
\begin{align}\label{eq:WAY-bound-Fisher-extremal}
 \| [\F(y), \nsys] - \ii^*_\xx([\F(y), \nsys]) \| & \leqslant 2\|\nsys\| \|\delta(y)\| +  \qq(\napp, \xi)^{\frac{1}{2}} \|\ii_\xx^*(\F(y)^2) - \ii_\xx^*(\F(y))^2 \|^{\frac{1}{2}}.
\end{align}
\end{prop}
\begin{proof}
 Let $\{q_i, \phi_i\}$ be an arbitrary ensemble of unit vectors that satisfies $\xi = \sum_i q_i \pr{\phi_i}$.  We may thus write $\Gamma_\xi^\ee(\cdot) = \sum_i q_i \Gamma_{\phi_i}^\ee(\cdot)$, where $\Gamma_{\phi_i}^\ee(\cdot) \equiv \Gamma_{\pr{\phi_i}}^\ee(\cdot)$. By the conservation law and additivity of $N$, we may therefore rewrite \eq{eq:WAY-disturbance-equality-1} as
 \begin{equation*}\label{eq:WAY-disturbance-equality-pure-dec}
   [\F(y) ,\nsys ] - \ii^*_\xx([\F(y), \nsys]) = [N\sub{\s} ,\delta(y) ] + \sum_i q_i \bigg([\Gamma_{\phi_i}^\ee(\F(y)\otimes \oneapp ), \Gamma_{\phi_i}^\ee(N) ] - \Gamma_{\phi_i}^\ee([\F(y)\otimes \oneapp, N])\bigg),
\end{equation*}
 which, by the sesquilinear mappings $\<\<A|B\>\>_i := \Gamma_{\phi_i}^\ee(A^*B) - \Gamma_{\phi_i}^\ee(A^*) \Gamma_{\phi_i}^\ee(B)$, \corref{corollary:Norm-commutator-inequality}, and \lemref{lemma:conservation-variance-condition} gives the bound
\begin{align*}
    \|[\F(y), \nsys] - \ii^*_\xx([\F(y), \nsys]) \| \leqslant 2\|\nsys\| \|\delta(y)\| + 2 \sum_i q_i \var{\napp, \phi_i}^{\frac{1}{2}} \| \Gamma_{\phi_i}^\ee(\F(y)^2 \otimes \oneapp) - \Gamma_{\phi_i}^\ee(\F(y) \otimes \oneapp)^2 \|^{\frac{1}{2}}.
\end{align*}
 Since both $\F(y)$ and $\Gamma_{\phi_i}^\ee(\F(y) \otimes \oneapp)$ are effects, we   have
\begin{align*}
 \|\Gamma_{\phi_i}^\ee(\F(y)^2 \otimes \oneapp) - \Gamma_{\phi_i}^\ee(\F(y) \otimes \oneapp)^2 \|^{\frac{1}{2}} \leqslant \|\Gamma_{\phi_i}^\ee(\F(y) \otimes \oneapp) - \Gamma_{\phi_i}^\ee(\F(y) \otimes \oneapp)^2 \|^{\frac{1}{2}}   \leqslant \frac{1}{2}.
\end{align*}
We thus arrive at the bound 
\begin{align*}
    \|[\F(y), \nsys] - \ii^*_\xx([\F(y), \nsys]) \| &\leqslant 2\|\nsys\| \|\delta(y)\| +  \sum_i q_i \var{\napp, \phi_i}^{\frac{1}{2}} \nonumber \\
    & \leqslant 2\|\nsys\| \|\delta(y)\| +  \left(\sum_i q_i \var{\napp, \phi_i}\right)^{\frac{1}{2}},
\end{align*}
where the second line follows from the concavity of the square root. By choosing the ensemble $\{q_i, \phi_i\}$ that gives the quantum Fisher information as in \eq{eq:QFI-defn},  we  arrive at  \eq{eq:WAY-bound-Fisher}.

Now assume that $\ii$ is an extremal instrument \cite{Pellonpaa2013}. This implies that for any pair of instruments $\ii^{(1)}$ and $\ii^{(2)}$, and any $\lambda \in (0,1)$, the operations of $\ii$ can be decomposed as $\ii_x(\cdot) = \lambda \, \ii_x^{(1)}(\cdot) + (1-\lambda) \, \ii_x^{(2)}(\cdot)$  only if $\ii = \ii^{(1)} = \ii^{(2)}$. It holds that $\Gamma_{\phi_i}^\ee(\cdot \otimes \oneapp) = \ii_\xx^*(\cdot)$ for all $i$, and so we obtain
\begin{align*}
    \|[\F(y), \nsys] - \ii^*_\xx([\F(y), \nsys]) \|  \leqslant  2\|\nsys\| \|\delta(y)\| + 2 \left(\sum_i q_i \var{\napp, \phi_i}\right)^{\frac{1}{2}} \| \ii_\xx^{*}(\F(y)^2) - \ii_\xx^{*}(\F(y))^2 \|^{\frac{1}{2}}.
\end{align*}
Once again, by choosing the ensemble that gives the quantum Fisher information, we  arrive at \eq{eq:WAY-bound-Fisher-extremal}.
\end{proof}

\section{Proof for Generalised WAY theorem 2}\label{app:strong-WAY-proof}

Here, we shall provide a detailed proof for \thmref{theorem:Strong-WAY} presented in the main text.

\begin{theorem}[Generalised WAY theorem 2]\label{theorem:General-WAY-2}
Let $\mm := (\ha, \xi, \ee, \Z)$ be a measurement scheme for an $\E$-instrument $\ii$ acting in $\hs$, and  assume that $\ee$  conserves an additive quantity  $N = N\sub{\s} \otimes \oneapp + \onesys\otimes N\sub{\aa}$ on average, where $\nsys \in \los(\hs)$ and $\napp \in \los(\ha)$.  If either $\ii$ is a measurement of the first kind, or the Yanase condition $[\Z, \napp]=\zero$ is satisfied, then for any effect $\E(x)$ that has both eigenvalue 1 and 0, it holds that
\begin{align}\label{eq:Strong-WAY}
  \P [\E(x), \nsys] \P =   [\E(x), \P \nsys \P] = \zero,
\end{align}
where  $\P := \P_0(x) + \P_1(x)$, with $\P_0(x)$ and $\P_1(x)$ orthogonal projections onto the eigenvalue-0 and eigenvalue-1 eigenspaces of $\E(x)$, respectively. 
\end{theorem}
\begin{proof}
Let us first note that $\E(x)\P = \P \E(x) = \P_1(x)$, and $\E(x)^\perp \P = \P \E(x)^\perp = \P_0(x)$, where $\E(x)^\perp := \onesys - \E(x)$. Indeed, $\P_0(x)$ may equivalently be considered as the projection onto the eigenvalue-1 eigenspace of $\E(x)^\perp$. Now define the operation $\Gamma_{\xi, \P}^\ee : \lo(\hs \otimes \ha) \to \lo(\hs)$ as $\Gamma_{\xi, \P}^\ee(\cdot) := \P \Gamma_\xi^\ee(\cdot) \P$, where $\Gamma_\xi^\ee$ is the unital CP map defined in \eq{eq:Gamma-U}. If $\mm$ is a measurement scheme for $\E$, then  $\Gamma_{\xi, \P}^\ee(\onesys \otimes \Z(x)) = \P_1(x)$.  Given  $1 \geqslant \|\Z(x) \| \geqslant \|\Gamma_{\xi,\P}^\ee(\onesys \otimes \Z(x)) \| = \| \P_1(x) \| = 1$, it follows that $\|\Z(x)\|=1$. Therefore, by the two-positivity of CP maps, and the relation $A^* B A \leqslant \|B\| A^*A$ for self-adjoint $B$, we observe that 
\begin{align*}
    \P_1(x) = \|\Z(x)\| \Gamma_{\xi, \P}^\ee(\onesys \otimes \Z(x)) \geqslant \Gamma_{\xi, \P}^\ee(\onesys \otimes \Z(x)^2) \geqslant \Gamma_{\xi, \P}^\ee(\onesys \otimes \Z(x))^2 = \P_1(x),
\end{align*}
and so 
\begin{align}\label{eq:strong-WAY-meas-1}
\Gamma_{\xi, \P}^\ee(\onesys \otimes \Z(x)^2) = \Gamma_{\xi, \P}^\ee(\onesys \otimes \Z(x))^2  = \P_1(x).
\end{align}
Now note that by additivity of $N$, and the conservation law,  it holds that $\Gamma_{\xi, \P}^\ee(N) = \P \nsys \P + \tr[\napp \xi] \P $. If the Yanase condition holds,  we may write
\begin{align*}
    \P [\E(x), \nsys] \P &= [\P \E(x) \P, \P \nsys \P] \\
    &= [\Gamma_{\xi, \P}^\ee(\onesys \otimes \Z(x)), \Gamma_{\xi, \P}^\ee(N)] \\
    & = \Gamma_{\xi, \P}^\ee([\onesys \otimes \Z(x), N]) \\
    &= \Gamma_{\xi, \P}^\ee(\onesys \otimes[\Z(x), \napp]) = \zero.
\end{align*}
The third line follows from \eq{eq:strong-WAY-meas-1} and the multiplicability theorem (\corref{corollary:multiplicability}), and the final line follows from the Yanase condition. As such,  we arrive at \eq{eq:Strong-WAY}.

Now let us abandon the Yanase condition, but instead assume that  $\ii$ is a first-kind measurement for $\E$. This implies that $\Gamma_{\xi, \P}^\ee(\E(x) \otimes \oneapp) = \P_1(x)$. Since $\|\E(x)\|=1$, then  by the two-positivity of CP maps, and the relation $A^* B A \leqslant \|B\| A^*A$ for self-adjoint $B$, we obtain
\begin{align*}
\P_1(x) = \| \E(x) \| \Gamma_{\xi, \P}^\ee(\E(x) \otimes \oneapp)   \geqslant  \Gamma_{\xi, \P}^\ee(\E(x)^2 \otimes \oneapp) \geqslant \Gamma_{\xi, \P}^\ee(\E(x) \otimes \oneapp)^2 = \P_1(x),
\end{align*}
and so 
\begin{align}\label{eq:strong-WAY-firstkind-1}
\Gamma_{\xi, \P}^\ee(\E(x)^2 \otimes \oneapp) = \Gamma_{\xi, \P}^\ee(\E(x) \otimes \oneapp)^2 =  \P_1(x).
\end{align}
 By the same arguments as in item (i)  of \propref{prop:repeatable-instrument-identity}, one can show from \eq{eq:strong-WAY-meas-1} and  \eq{eq:strong-WAY-firstkind-1} that $\Gamma_{\xi, \P}^\ee(\onesys \otimes \Z(x)^n) =  \Gamma_{\xi, \P}^\ee(\E(x)^n \otimes \oneapp) = \P_1(x)$ for all $n\in \nat$. Consequently, by the same arguments as in item  (iv) of \propref{prop:repeatable-instrument-identity}, it follows that  $ \Gamma_{\xi, \P}^\ee(\onesys \otimes (\Z(x) - \Q(x))) = \zero$, where $\Q(x)$ is the projection onto the eigenvalue-1 eigenspace of $\Z(x)$, and $\Gamma_{\xi, \P}^\ee(( \E(x) - \P_1(x)) \otimes \oneapp) = \zero$. Moreover, by \eq{eq:strong-WAY-meas-1} ,  \eq{eq:strong-WAY-firstkind-1},  the multiplicability theorem (\corref{corollary:multiplicability}), and defining  $\Z(x)^\perp := \oneapp - \Z(x)$,  it follows that 
\begin{align*}
\Gamma_{\xi, \P}^\ee(\E(x) \otimes \Z(x)^\perp) = \Gamma_{\xi, \P}^\ee(\E(x) \otimes \oneapp) \Gamma_{\xi, \P}^\ee(\onesys \otimes \Z(x)^\perp) = \P_1(x) \P \E(x)^\perp \P = \P_1(x) \P_0(x) = \zero, \\
\Gamma_{\xi, \P}^\ee(\E(x)^\perp \otimes \Z(x)) = \Gamma_{\xi, \P}^\ee(\E(x)^\perp \otimes \oneapp) \Gamma_{\xi, \P}^\ee(\onesys \otimes \Z(x)) =  \P \E(x)^\perp \P \P_1(x) = \P_0(x) \P_1(x) = \zero,
\end{align*}
and so by the same arguments as in items (v) and (vi)  of  \propref{prop:repeatable-instrument-identity} it follows that 
\begin{align}\label{eq:first-kind-repeatable-mimic-identity}
    \Gamma_{\xi, \P}^\ee(\E(x)A \otimes \oneapp) = \Gamma_{\xi, \P}^\ee(A \E(x) \otimes \oneapp) = \Gamma_{\xi, \P}^\ee(\P_1(x) A \P_1(x) \otimes\oneapp)
\end{align}
for all $A \in \lo(\hs)$. By additivity of $N$, and the conservation law,  we may therefore write
\begin{align*}
    \P [\E(x), \nsys] \P &= [\P \E(x) \P, \P \nsys \P] \\
    &= [\Gamma_{\xi, \P}^\ee(\E(x) \otimes \oneapp), \Gamma_{\xi, \P}^\ee(N)] \\
    & = \Gamma_{\xi, \P}^\ee([\E(x) \otimes \oneapp, N]) \\
    &= \Gamma_{\xi, \P}^\ee([\E(x), \nsys] \otimes \oneapp) = \zero.
\end{align*}
The third line follows from   \eq{eq:strong-WAY-firstkind-1} and \corref{corollary:multiplicability}, while the final line follows from \eq{eq:first-kind-repeatable-mimic-identity}.  Once again we arrive at \eq{eq:Strong-WAY}.

\end{proof}

\section{The Weak Yanase condition from conservation laws}\label{app:Weak-Yanase}

Thus far, we have only considered the case where the measurement interaction $\ee$ between system and apparatus conserves an additive quantity $N$. However, pointer objectification will also result in state changes, and it may be the case that the expected value of $N$ will change as a result. Now let us provide a generalised prescription of measurement schemes that captures also the state changes  due to pointer objectification. Recall that  $\mm:=(\ha, \xi, \ee, \Z)$ is a measurement scheme for an observable $\E$ acting in $\hs$ if $\E(x) = \Gamma_\xi\circ \ee^*(\onesys \otimes \Z(x))$.   Now consider the tuple $\tilde \mm := (\ha, \xi, \jj)$,  where $\jj := \{\jj_x : x\in \xx\}$ is an instrument acting in $\hs\otimes \ha$. $\tilde \mm$ is also a measurement scheme for $\E$ if  $\E(x) = \Gamma_\xi\circ \jj_x^*(\onesys \otimes \oneapp)$. It is straightforward to show that this is satisfied if $\jj$ is compatible with the ``Heisenberg-evolved'' pointer observable 
\begin{align}\label{eq:heisenberg-pointer}
 \Z^\tau(x)  := \ee^*(\onesys\otimes \Z(x)),
\end{align}
that is, if $\jj_x^*(\onesys \otimes \oneapp) =  \Z^\tau(x)$.  We  say that $\tilde \mm$ obeys a  full (average) conservation law if the channel  $\jj_\xx$ fully (on average)  conserves a quantity $N$. The operations $\jj_x$ can be constructed as a sequential application of the channel $\ee$ followed by the operations of some $\Z$-compatible instrument acting in $\ha$, the latter of which provides a physical characterisation of the pointer objectification process. In such a case, a sufficient condition for conservation of $N$ by $\jj_\xx$ is  the conservation of $N$ by both $\ee$ and the $\Z$-channel. But it may be the case that $\ee$ fully conserves $N$ while the $\Z$-channel  conserves $N$ only on average, and vice versa. In such cases, the  channel $\jj_\xx$ will   conserve $N$ only on average. 

By \lemref{lemma:fixed-points-instrument}, it holds that if $\jj_\xx$ conserves $N$ on average, and if  either $\Z^\tau$ is sharp or if $\jj_\xx$ also fully conserves $N$, then  
\begin{align}\label{eq:weak-Yanase}
    [\Z^\tau, N]=\zero.
\end{align}
 This commutation relation is known as the weak Yanase condition \cite{Tukiainen2017}.  We note that if $\ee$ conserves $N$ on average and if either $\Z^\tau$ is sharp or if $\ee$ also fully conserves $N$, then the Yanase condition implies the weak Yanase condition. First, let us assume that $\Z^\tau$ is sharp.  Since $\Z(x)$ is an effect then by two-positivity of CP maps we have $\Z^\tau(x)= \ee^*(\onesys \otimes \Z(x)) \geqslant \ee^*(\onesys \otimes \Z(x)^2) \geqslant \ee^*(\onesys \otimes \Z(x))^2 = \Z^\tau(x)$, and so  we have $\ee^*(\onesys \otimes \Z(x)^2) = \ee^*(\onesys \otimes \Z(x))^2$. On the other hand, if $\ee$ fully conserves $N$ then $\ee^*(N^2) = \ee^*(N)^2 = N^2$. In either case, by \corref{corollary:multiplicability} we have 
 \begin{align*}
[\Z^\tau(x), N] =  [\ee^*(\onesys \otimes \Z(x)), \ee^*(N)] = \ee^*([\onesys \otimes \Z(x), N]) = \ee^*(\onesys \otimes[ \Z(x), \napp]),
 \end{align*}
  and so if $[ \Z(x), \napp]=\zero$, then $[\Z^\tau(x), N]=\zero$. Moreover, if $\ee(\cdot) = U(\cdot)U^*$ is a unitary channel, and $\ee$ conserves $N$,   then $ [\Z^\tau(x), N] =  U^*(\onesys \otimes[\Z(x), \napp ])U$. In such a case the weak Yanase condition is equivalent to the Yanase condition;  multiplying both sides of the equality $U^*(\onesys \otimes[\Z(x), \napp ])U = \zero$  by $U$ from the left and by $U^*$ from the right shows that $[\Z^\tau(x), N]= \zero \iff [\Z(x), \napp ]=\zero$. However, in general it may be the case that the weak Yanase condition is satisfied but the Yanase condition is violated.   
 
 The following proposition shows that if the weak Yanase condition is satisfied, then the measurability part of the WAY theorem will hold. Moreover,    we see that there are cases where a large coherence of the conserved quantity in the apparatus is necessary for good measurements even  without a full conservation law---for example, if either the interaction channel $\ee$ or the $\Z$-channel   conserves $N$ only on average, but  $\Z^\tau$ is sharp, in which case the weak Yanase condition is guaranteed to hold.

\newpage
 
\begin{prop}\label{prop:weak-Yanase-WAY}
Let $\mm:= (\ha, \xi, \ee, \Z)$ be a measurement scheme for an observable acting in $\hs$, and  let $\| \epsilon(x)\|$ be the error in measuring the effects of the target observable $\E$. Assume that $\mm$ satisfies the weak Yanase condition $[\Z^\tau, N]=\zero$, where $\Z^\tau$ is the Heisenberg-evolved pointer observable defined in \eq{eq:heisenberg-pointer} and $N = \nsys \otimes \oneapp + \onesys \otimes \napp$, where $\nsys \in \los(\hs)$ and $\napp \in \los(\ha)$.  Then for all $x\in \xx$ it holds that
\begin{align}\label{eq:weak-Yanase-WAY-variance}
 \| [\E(x), \nsys] \| \leqslant 2  \| \nsys \|\| \epsilon(x)\| + 2 \var{\napp, \xi}^{\frac{1}{2}} \bigg(2 \| \epsilon(x)\| + \|\E(x) - \E(x)^2 \|\bigg)^{\frac{1}{2}},  
\end{align}
and 
\begin{align}\label{eq:weak-Yanase-WAY-QFI}
 \| [\E(x), \nsys] \| \leqslant 2  \| \nsys \|\| \epsilon(x)\| + \frac{1}{2} \qq(\napp, \xi)^{\frac{1}{2}},
\end{align}
where $\var{\napp, \xi}:= \tr[N\sub{\aa}^2 \xi] - \tr[\napp \xi]^2$  and $\qq(\napp, \xi)$ is the the quantum Fisher information of $\napp$ in $\xi$ as defined in \eq{eq:QFI-defn}. Additionally, if $\mm$ is a measurement scheme for $\E$, then for  any effect  $\E(x)$ that has both eigenvalue 1 and 0, it holds that
\begin{align}\label{eq:weak-Yanase-Strong-WAY}
  \P [\E(x), \nsys] \P = [\E(x), \P \nsys \P] = \zero,
\end{align}
where $\P = \P_0(x) + \P_1(x)$, with $\P_0(x)$ and $\P_1(x)$ orthogonal projections onto the eigenvalue-0 and eigenvalue-1 eigenspaces of $\E(x)$, respectively.
\end{prop}
\begin{proof}
 By \eq{eq:error-equality}, we may write  $ \epsilon(x) := \Gamma_\xi^\ee(\onesys \otimes \Z(x)) - \E(x) \equiv \Gamma_\xi(\Z^\tau(x)) - \E(x)$. By  additivity of $N$ we have $\Gamma_\xi(N) = \nsys + \tr[\napp \xi] \onesys$, and so we may write  
\begin{align*}
    [\E(x), \nsys]  = [\nsys, \epsilon(x)] + [\Gamma_\xi(\Z^\tau(x)), \Gamma_\xi(N)].
\end{align*}
Since $\Gamma_\xi$ is a channel, and the weak Yanase condition $[\Z^\tau, N]=\zero$ holds, then by the sesquilinear map $\<\<A|B\>\> := \Gamma_\xi(A^*B) - \Gamma_\xi(A^*) \Gamma_\xi(B)$ and \corref{corollary:Norm-commutator-inequality} we obtain 
\begin{align*}
 \|[\E(x), \nsys]\|  \leqslant  2  \|\nsys\| \| \epsilon(x)\|  +  2\| \Gamma_\xi(N^2) - \Gamma_\xi(N)^2 \|^{\frac{1}{2}} \| \Gamma_\xi(\Z^\tau(x)^2) - \Gamma_\xi(\Z^\tau(x))^2\|^{\frac{1}{2}}.   
\end{align*}
As shown in \lemref{lemma:conservation-variance-condition}, additivity of $N$ implies that $\|\Gamma_\xi(N^2) - \Gamma_\xi(N)^2\| = \var{\napp, \xi}$. On the other hand, by \lemref{lemma:unsharp-disturbance-bound} we obtain $\| \Gamma_\xi(\Z^\tau(x)^2) - \Gamma_\xi(\Z^\tau(x))^2\| \leqslant 2 \| \epsilon(x) \| + \|\E(x) - \E(x)^2 \|$. As such, we obtain the bound in \eq{eq:weak-Yanase-WAY-variance}.  \eq{eq:weak-Yanase-WAY-QFI} and \eq{eq:weak-Yanase-Strong-WAY} are trivially obtained by adapting the arguments in \propref{prop:measurability-Fisher} and \thmref{theorem:General-WAY-2} to the above, i.e., by replacing $\Z$ with $\Z^\tau$, and replacing the operations  $\Gamma_{\phi_i}^\ee$ and $\Gamma_{\xi, \P}^\ee$ with the operations  $\Gamma_{\phi_i}$ and $\Gamma_{\xi, \P}$, respectively. 
\end{proof}

\section{Proof of the ``converse'' WAY theorem}\label{app:converse-WAY}
Here, we provide a proof for \propref{prop:converse-WAY-main} presented in the main text. 

\begin{prop}\label{prop:converse-WAY}
Let $\mm:=(\ha, \xi, \ee, \Z)$ be a measurement scheme for an instrument $\ii$ acting in $\hs$. Assume that $\ee$ fully conserves an additive quantity $N = \nsys \otimes \oneapp + \onesys \otimes \napp$, where $\nsys \in \los(\hs)$ and $\napp \in \los(\ha)$, and that $\ii_\xx$ fully conserves $\nsys$. Define the eigenspace of $\ha$ that is involved during the measurement process as
\begin{align*}
\ha(\text{meas}) := \bigcup_{\rho \in \s(\hs)} \supp( \Lambda(\rho)) \cup \supp(\xi)  \subseteq \ha,   
\end{align*}
where $\Lambda$ is the conjugate channel to $\ii_\xx$ defined in \eq{eq:conjugate-channel}. Then either  $\ha(\text{meas})$ is contained within a single degenerate eigenspace of $\napp$, or  $\var{\napp, \xi}:= \tr[\nappk{2}\xi] - \tr[\napp \xi]^2$ must be large. Additionally, if $\ii$ is an extremal instrument, and if $\ha(\text{meas})$ is not contained within a single degenerate eigenspace of $\napp$, then the quantum Fisher information $\qq(\napp, \xi)$ as defined in \eq{eq:QFI-defn} must be large. 
\end{prop}
\begin{proof}
 Let us first observe that by \eq{eq:conservation-compensation}, if $\ee$ conserves $N$ on average, then $\nsys \in \ff(\ii_\xx^*)$ implies that irrespective of the apparatus preparation $\xi$, it holds that  $\tr[\napp \Lambda(\rho)] = \tr[\napp \xi]$ for all $\rho \in \s(\hs)$ or, equivalently, that $\Lambda^*(\napp) = \tr[\napp \xi] \onesys$. That is, the expected value of the apparatus part of the conserved quantity does not change as a result of the measurement interaction. While average conservation does not imply that the variance must also stay the same, this implication can be shown to follow in the case of full conservation. By  \defref{defn:conservation-law},  full conservation of $N$ by $\ee$ implies that $\Gamma_\xi(N^2) =  \Gamma_\xi^\ee(N^2)$, where $\Gamma_\xi^\ee$ is the channel defined in \eq{eq:Gamma-U}. Given that  $\Gamma_\xi^\ee(\cdot \otimes \oneapp) = \ii_\xx^* (\cdot)$ and $\Gamma_\xi^\ee(\onesys \otimes \cdot) = \Lambda^*(\cdot)$, we thus obtain 
\begin{align*}
\nsysk{2} + 2\tr[\napp \xi]\nsys + \tr[\nappk{2} \xi]\onesys & =  \Gamma_\xi^\ee(\nsysk{2} \otimes \oneapp) + 2\Gamma_\xi^\ee((\nsys \otimes \oneapp)(\onesys\otimes \napp)) + \Gamma_\xi^\ee(\onesys \otimes \nappk{2}) \\
& = \Gamma_\xi^\ee(\nsysk{2} \otimes \oneapp) + 2\Gamma_\xi^\ee(\nsys \otimes \oneapp)\Gamma_\xi^\ee(\onesys\otimes \napp) + \Gamma_\xi^\ee(\onesys \otimes \nappk{2}) \\
& = \ii_\xx^* (\nsysk{2}) +2 \ii_\xx^* (\nsys) \Lambda^*(\napp) + \Lambda^*(\nappk{2}) \\
& = \nsysk{2} + 2 \tr[\napp \xi] \nsys + \Lambda^*(\nappk{2}).
\end{align*}
In the second line, we have used the fact that full conservation of $\nsys$ by $\ii_\xx$ implies that $\Gamma_\xi^\ee(\nsysk{2} \otimes \oneapp) = \Gamma_\xi^\ee(\nsys \otimes \oneapp)^2 = \nsysk{2}$, and the multiplicability theorem (\corref{corollary:multiplicability}). It follows that $\Lambda^*(\nappk{k}) = \tr[\nappk{k} \xi]\onesys$ for $k=1,2$. As such, for any input state $\rho$ of the system to be measured, it holds that
\begin{align*}
\var{\napp, \Lambda(\rho)} = \tr[\nappk{2}\Lambda(\rho)] - \tr[\napp \Lambda(\rho)]^2 = \tr[\nappk{2} \xi] - \tr[\napp \xi]^2 = \var{\napp, \xi} .   
\end{align*}
Now assume that $\xi$ is an eigenstate of $\napp$, i.e., that there exists $c \in \re$ such that $\napp \xi = c \, \xi$. Since such a condition is equivalent to a vanishing variance, then we see that    $\var{\napp, \Lambda(\rho)} = \var{\napp,\xi} = 0$, and so  $\Lambda(\rho)$ must also be eigenstates of $\napp$ with the same eigenvalue $c$. In fact, in such a case it holds that $\Lambda^*(\nappk{k}) = \Lambda^*(\napp)^k = \tr[\napp \xi]^k \onesys$ for $k=1,2$, which implies that $\tr[\nappk{k} \Lambda(\rho)] = \tr[\nappk{k} \xi] = c^k$ for all $k \in \nat$ and $\rho \in \s(\hs)$.  In other words, if $\xi$ is an eigenstate of $\napp$, then  $\napp$  must be ``effectively'' fully degenerate, i.e., $\ha(\text{meas})$  must be contained within a single degenerate eigenspace of $\napp$. Therefore, if $\ha(\text{meas})$ contains more than one degenerate eigenspace of $\napp$,     the apparatus must be prepared in a state with a large uncertainty in $\napp$. 

Now assume that $\ii$ is an extremal instrument, i.e., that for any $\lambda \in (0,1)$, the operations of $\ii$ admit a decomposition $\ii_x(\cdot) = \lambda \, \ii_x^{(1)}(\cdot) + (1-\lambda) \, \ii_x^{(2)}(\cdot)$ only if $\ii = \ii^{(1)} = \ii^{(2)}$. In such a case, it follows that if $\mm:= (\ha, \xi, \ee, \Z)$ is a measurement scheme for $\ii$, then  for any pure state decomposition $\xi = \sum_i q_i \pr{\phi_i}$, it holds that   $(\ha, \phi_i, \ee, \Z)$ must also be measurement schemes for $\ii$. By the above arguments, it follows that unless $\napp$ is effectively degenerate, then each  $\phi_i$ must have a large uncertainty in $\napp$, i.e., the quantum Fisher information $\qq(\napp, \xi)$ must be large. 
\end{proof}

\section{Faithful fixed states and measurement disturbance}\label{app:faithful-fixed-points}

Recall from \propref{prop:quantitative-bound-disturbance} and  \thmref{theorem:quantitative-bound-disturbance-WAY} that under a conservation law, an $\E$-instrument $\ii$ will no disturb an observable $\F$ only if 
\begin{align*}
   \|  [\E(x), \F(y)] \| &\leqslant 2\| \E(x) - \E(x)^2 \|^{\frac{1}{2}} \|\ii^*_\xx(\F(y)^2) - \F(y)^2 \|^{\frac{1}{2}} \qquad \forall \, x\in \xx, \, y \in \yy,
\end{align*}
and 
\begin{align*}
\| [\F(y), \nsys] - \ii^*_\xx([\F(y), \nsys]) \| & \leqslant  2 \|\Gamma_\xi^\ee(N^2) - \Gamma_\xi^\ee(N)^2 \|^{\frac{1}{2}} \|\ii^*_\xx(\F(y) ^2) - \F(y)^2  \|^{\frac{1}{2}} \qquad \forall \,  y \in \yy.
\end{align*}
If $\F \subset \ff(\ii_\xx^*) \implies \F^2 \subset \ff(\ii_\xx^*)$, then the upper bound of the above equations vanish, and so $\F$ will be non-disturbed only if $[\F, \E] = \zero$ and $ [\F(y), \nsys] \in \ff(\ii_\xx^*)$. By the multiplicability theorem (\corref{corollary:multiplicability}), the second condition can be shown to simplify to $[\F(y), \nsys] = \ii_\xx^*([\F(y), \nsys]) = [\F(y), \ii_\xx^*(\nsys)]$, which reads $[\F(y), \ii_\xx^*(\nsys) - \nsys] = \zero$. In \app{app:disturbance-compatibility} we saw that the implication $\F \subset \ff(\ii_\xx^*) \implies \F^2 \subset \ff(\ii_\xx^*)$ holds if either $\F$ is sharp, rank-1, or a coarse-graining of a sharp observable.  We now show that if $\ff({\ii^*_\xx})$ is a von Neumann algebra, so that $\F \subset \ff(\ii_\xx^*) \implies \F^2 \subset \ff(\ii_\xx^*)$, similar and stronger constraints will hold for all observables.  To this end,  let us first prove a useful lemma.

\begin{lemma}\label{lemma:faithful-fixedpoint-commutant}
Let   $\mm:= (\ha, \xi, \ee, \Z)$ be a measurement scheme for an $\E$-instrument $\ii$ acting in $\hs$. Assume that  $\ee$  conserves an additive quantity $N = N\sub{\s} \otimes \oneapp + \onesys\otimes N\sub{\aa}$ on average, where $\nsys \in \los(\hs)$ and $\napp \in \los(\ha)$,  and that $\ff(\ii_\xx^*)$ is a von Neumann algebra.   Then  for all $A \in \lo(\hs)$, the following implication holds:  $A\in \ff(\ii_\xx^*) \implies [A, N\sub{\s} ] \in \ff(\ii_\xx^*) \implies [A, N\sub{\s} ] \in  \E'$.
\end{lemma}
\begin{proof}

Recall that for all $A \in \lo(\hs)$, we have $\ii^*_\xx(A) = \Gamma_\xi^\ee(A \otimes \oneapp)$, where $\Gamma_\xi^\ee$ is  the channel  defined in \eq{eq:Gamma-U}.    Average conservation of $N$ by $\ee$   implies that  
\begin{align*}
\nsys + \tr[\napp \xi] \onesys =  \Gamma_\xi^\ee( \nsys\otimes \oneapp) + \Gamma_\xi^\ee(\onesys\otimes \napp ). 
\end{align*}
Since $\ff(\ii^*_\xx)$ is a von Neumann algebra, it follows that for all  $A \in \lo(\hs)$,  $A \in \ff(\ii^*_\xx) \implies A^*A, A A^* \in \ff(\ii^*_\xx)$ which, by \corref{corollary:multiplicability},  implies that $A\Gamma_\xi^\ee(B) = \Gamma_\xi^\ee((A  \otimes \oneapp)B)$ and $ \Gamma_\xi^\ee(B)A = \Gamma_\xi^\ee(B(A  \otimes \oneapp))$  for all $B\in \lo(\hs\otimes \ha)$. Therefore, for all $A \in \ff(\ii^*_\xx)$ we have
\begin{align*}
[A,  \nsys] &= [A, \Gamma_\xi^\ee(\nsys \otimes \oneapp)] + [A, \Gamma_\xi^\ee(\onesys \otimes \napp)] \nonumber \\
& = \Gamma_\xi^\ee([A\otimes \oneapp, \nsys \otimes \oneapp ]) + \Gamma_\xi^\ee([A \otimes \oneapp, \onesys \otimes \napp]) \nonumber \\
& = \Gamma_\xi^\ee([A, \nsys ]\otimes \oneapp) = \ii^*_\xx([A, \nsys]).
\end{align*}
Consequently, we see that $A  \in \ff(\ii^*_\xx) \implies [A, \nsys] \in \ff(\ii^*_\xx)$. But as shown in \lemref{lemma:fixed-points-instrument}, if $\ff(\ii_\xx^*)$ is a von Neumann algebra then $\ff(\ii_\xx^*) \subset \E'$. It follows that $A  \in \ff(\ii^*_\xx) \implies [A,
 \nsys]  \in \E'$.

\end{proof}

We are now ready to prove the following:

\begin{theorem}\label{theorem:faithful-fixedpoint-WAY}
Let $\E :=\{\E(x) : x\in \xx\}$ and $\F:=\{\F(y) : y \in \yy\}$ be observables acting in $\hs$. Let $\mm:=(\ha, \xi, \ee, \Z)$ be a measurement scheme for an  $\E$-instrument $\ii$, and assume that $\ee$ conserves an additive quantity $N = \nsys \otimes \oneapp + \onesys \otimes \napp$ on average, where $\nsys \in \los(\hs)$ and $\napp \in \los(\ha)$. If $\ff(\ii_\xx^*)$ is a von Neumann algebra, then the following hold:
\begin{enumerate}[(i)]
    \item $\F \subset \ff(\ii_\xx^*)$ only if  $\F$ commutes with  $\E$,   with $\Delta \nsys := \ii_\xx^*(\nsys) - \nsys$, and with $\{[\E(x), \nsys] : x\in \xx\}$.
    \item $\ii$  is a measurement of the first kind only if  $\E$  is a commutative observable that commutes with $N\sub{\s}$.
    \item $\ii$ is  repeatable  only if $\E$ is sharp and commutes with $\nsys$.
\end{enumerate}
\end{theorem}
\begin{proof}
\begin{enumerate}[(i):]
    \item By \lemref{lemma:faithful-fixedpoint-commutant},  $\ff(\ii_\xx^*) \subset \E'$. Therefore, $\F \subset \ff(\ii_\xx^*) \implies [\E, \F] = \zero$. Moreover, the conservation law implies that $[\F, \nsys] \subset \ff(\ii^*_\xx)$ must hold.  Since $\ff(\ii_\xx^*)$ is an algebra, then $\F(y) \in \ff(\ii_\xx^*) \implies \F(y)^2 \in \ff(\ii_\xx^*)$. By \corref{corollary:multiplicability}, it follows that $\ii_\xx^*(\F(y) \nsys) = \F(y) \ii_\xx^*(\nsys)$ and $\ii_\xx^*(\nsys\F(y)) = \ii_\xx^*(\nsys) \F(y)$, and so $\ii_\xx^*([\F(y), \nsys]) = [\F(y), \ii_\xx^*(\nsys)]$. As such, $\ii_\xx^*([\F(y), \nsys]) = [\F(y), \nsys] \implies [\F(y), \ii_\xx^*(\nsys) - \nsys ] = \zero$. Moreover,  by \lemref{lemma:faithful-fixedpoint-commutant} $\F \subset \ff(\ii_\xx^*)$ implies  $[\F(y), \nsys] \in \E'$. Since $[A, [B, C]] = [B, [A,C]] + [C, [B, A]]$ holds for all $A, B, C \in \lo(\hs)$, while $[\E, \F] = \zero$, then this implies that  $[\F(y), [\E(x), \nsys]] = \zero$.
    
    \item  $\ii$ is a measurement of the first kind only if $\E \subset \ff(\ii^*_\xx)$. Commutativity of  $\E$ follows from (i).   Now, let us define $\nsys(t):=e^{\imag t \E(x)}
    \nsys e^{-\imag t \E(x)}$.   We may write 
    \begin{align*}
    \frac{d}{dt}\nsys(t)     =\imag [\E(x), \nsys(t)],
    \end{align*}
    from which we obtain 
    \begin{align*}
    \nsys(t)&=\nsys     + \imag \int^t_0 dt_1     [\E(x), \nsys(t_1)]     \\
    &=     \nsys + \imag \int^t_0 d t_1 [\E(x), \nsys]     - \int^t_0 dt_1 \int^{t_1}_0     d t_2 \, e^{\imag t_2 \E(x)}[\E(x), [\E(x), \nsys]]     e^{- \imag t_2 \E(x)}      \\
    &= \nsys + \imag t[\E(x), \nsys].
    \end{align*}
    In the final line we have used (i), which implies that $[\E(x),[\E(x), \nsys]]=\zero$ holds for all $x$. We thus obtain the inequality 
    \begin{eqnarray*}
    2\| \nsys \| \geqslant     \| \nsys(t) - \nsys \|= |t| \| [\E(x), \nsys]\|
    \end{eqnarray*}
    for all $t$. Given that $\nsys$ is a bounded operator, this is clearly satisfied only if     $[\E(x), \nsys]=\zero$. 
    
    % Let $R$ be a spectral projection of $\E(x) $. Since $\ff(\ii^*_\xx)$ is a von Neumann algebra,  then $\E(x) \in \ff(\ii^*_\xx) \implies R \in \ff(\ii^*_\xx)$ for all $n$, and by \lemref{lemma:faithful-fixedpoint-commutant}  $R \in \ff(\ii^*_\xx) \implies [R, \nsys] \in  \E'$. But $[A, \E(x)]=\zero$ only if $[A, R] = \zero$. As such, we must have $[[R, \nsys], R^\perp] = \zero$, where $R^\perp := \onesys - R$. But this implies that $[R\nsys - \nsys R,R^\perp] = R \nsys R^\perp + R^\perp \nsys R = \zero$. Multiplying from the left by $R$, we thus have $R \nsys R^\perp = R \nsys (\onesys - R) = \zero$,  and so $R \nsys  = R \nsys R$. Since the right hand side is self-adjoint, it follows that  $R \nsys = \nsys R$.  Since this relation holds for all  spectral projections of all effects of $\E$, it follows that $\E$ must commute with $\nsys$.
    
    \item  Since repeatability implies first-kindness, then by (ii) $\E$ must commute with $\nsys$.   Sharpness of $\E$ follows from item (i) of \propref{prop:repeatable-instrument-identity} which gives $\ii_\xx^*(\E(x)^2) = \E(x)$ for a repeatable $\ii$, and the fact that if $\ff(\ii^*_\xx)$ is a von Neumann algebra then $\E \subset \ff(\ii^*_\xx)$ implies that $\E^2 \subset \ff(\ii^*_\xx)$. 
\end{enumerate}

\end{proof}

 We shall now give two examples where $\ff(\ii^*_\xx)$ is a von Neumann algebra, and so the implications of \thmref{theorem:faithful-fixedpoint-WAY} hold.

\begin{lemma}\label{lemma:Luders-algebra}
Consider the L\"uders $\E$-instrument $\ii^L$ acting in $\hs$, defined in \eq{eq:luders}.   If either (i)  $\dim(\hs) <\infty$, or (ii) $\E$  is commutative, then $\ff({\ii^L_\xx}^*)$ is a von Neumann algebra. 
\end{lemma} 
\begin{proof}
Let us first consider (i).  Define the complete mixture $\omega:= \onesys/ \dim(\hs)$, which is faithful. It follows trivially that $\ii^L_\xx(\omega) = \omega$, and so $\ff(\ii^L_\xx)$ contains a faithful state $\omega$. By \lemref{lemma:Lindblad}  $\ff({\ii^L_\xx}^*)$ is a von Neumann algebra.  Now let us consider (ii). Recall  that  $\ff(\ii^*_\xx)$ is a von Neumann algebra if  $\ff(\ii^*_\xx) = \{K_i, K_i^*\}'$, with $\{K_i\}$ any Kraus representation of $\ii_\xx$ \cite{Bratteli1998}. But for a L\"uders instrument, we have $\{K_i, K_i^*\}' = \{\sqrt{\E(x)}\}' = \E'$. While $\E' \subset \ff({\ii^L_\xx}^*)$ always holds, it was observed that in infinite-dimensional systems there exists $\E$ for which $ \ff({\ii^L_\xx}^*) \not\subset \E'$ \cite{Arias2002,Weihua2009}. However, it was shown in \cite{Busch1998} that for binary observables, it always holds that $\ff({\ii^L_\xx}^*) = \E'$. Since binary observables are commutative, this led to the conjecture that the fixed-point set of the L\"uders $\E$-channel is the commutant of $\E$ for all commutative observables \cite{Arias2002}, making $\ff({\ii^L_\xx}^*)$ a von Neumann algebra, which was later proven to be the case   \cite{Weihua2010, Prunaru2011}. 

\end{proof}

Let us highlight an interesting consequence of the above lemma:

\begin{corollary}\label{Luders-WAY}
Let   $\mm := (\ha, \xi, \ee, \Z)$ be a measurement scheme for an $\E$-compatible L\"uders instrument $\ii^L$  acting in $\hs$. Assume that $\E$ is commutative, and that  $\ee$  conserves an additive quantity $N = N\sub{\s} \otimes \oneapp + \onesys\otimes N\sub{\aa}$ on average. It holds that $\E$ commutes with $\nsys$. 
\end{corollary}
\begin{proof}
If $\E \subset \E'$, then $\ff({\ii^L_\xx}^*) =  \E'$ is a von Neumann algebra. Moreover, it holds that  $\E \subset \ff({\ii^L_\xx}^*)$,  so that the L\"uders instrument for a commutative observable is  a measurement of the first kind.  It follows from item (ii) of  \thmref{theorem:faithful-fixedpoint-WAY} that $\E$ must commute with $\nsys$.  
\end{proof}

\begin{lemma}\label{lemma:fuzzy-sharp}
 Assume that  $\dim(\hs) <\infty$, and let $\G:= \{\G(z) : z\in \zz\}$ be a rank-1  observable acting in $\hs$.  If an instrument $\ii$ does not disturb $\G$, then $\ff(\ii^*_\xx)$ is a von Neumann algebra.
\end{lemma}
\begin{proof}
Since $\G(z)$ are rank-$1$ effects, we may write $\G(z) = \lambda_z P_z$, with $\lambda_z \in (0,1]$ and $P_z$ a rank-1 projection. It follows that $\G \subset \ff(\ii^*_\xx) \implies \{P_z\} \subset \ff(\ii^*_\xx)$. But we may write
\begin{align*}
\tr[P_z \ii_\xx(P_z)] =  \tr[\ii_\xx^*(P_z) P_z] = \tr[P_z P_z] = \tr[P_z] = 1    ,
\end{align*} 
and so $\G \subset \ff(\ii^*_\xx) \implies \{P_z\} \subset \ff(\ii_\xx)$. Consequently, we may construct the faithful state $\omega = \sum_z p_z P_z $ with  $p_z >0$ and $\sum_z p_z =1$, so that $\omega \in \ff(\ii_\xx)$. By \lemref{lemma:Lindblad}, $\ff(\ii^*_\xx)$ is a von Neumann algebra. 
\end{proof}

\section{Non-faithful fixed states and measurement disturbance}\label{app:non-faithful-fixed-points}
In this section we analyse the structure of the fixed-point set of arbitrary channels, which need not contain a faithful state. From here, the results of the previous section are generalised. We then provide novel quantitative bounds for first-kind measurements which  complement our generalisation of the WAY theorem given in  \thmref{theorem:Strong-WAY}.

Due to the Schauder–Tychonoff fixed point theorem \cite{Fixed-points-appl}, all channels $\Phi : \trc(\hs) \to \trc(\hs)$ have at least one fixed state. However, it may be  that none of these are faithful. In such a case, the fixed-point set of the dual channel is not necessarily a von Neumann algebra, but rather forms an operator space \cite{Intro-Operator-Space}. This setting has been much less investigated, and its analysis
forms the first part of this section. While the discussion thus far has  been applicable for infinite-dimensional systems---except in some examples---in this section we shall always assume that $d:= \dim(\hs) < \infty$.

\subsection{Fixed-point structure of arbitrary channels}\label{app:operator-space-theory}

Consider a channel $\Phi : \trc(\hs) \to \trc(\hs)$, and its dual in the Heisenberg picture $\Phi^*: \lo(\hs) \to \lo(\hs)$. We may define the channels
\begin{align}\label{eq:Phi-av-defn}
& \Phi_\av(\cdot):= \lim_{N\to \infty} \frac{1}{N}\sum_{n=1}^N 
\Phi^{n}(\cdot), &\Phi^*_\av(\cdot):= \lim_{N\to \infty} \frac{1}{N}\sum_{n=1}^N 
\Phi^{*n}(\cdot),
\end{align} 
where $\Phi^{n}$ denotes $n$ consecutive applications of $\Phi$.  Note that these limits exists since $d <\infty$.   According to the Jordan decomposition theorem, $\Phi^*$ is represented as a summation of projections onto eigenspaces multiplied by the corresponding eigenvalues,  and nilpotent operators whose eigenspaces are invariant subspaces; $\Phi^*_\av$ corresponds to the projection onto the subspace 
with eigenvalue $1$. The fixed-point set  $\ff(\Phi^*)$   forms an operator space, i.e., a norm-closed vector subspace of the codomain of $\ff(\Phi^*)$, and $\Phi^*_\av$ is a CP projection onto $\ff(\Phi^*)$.
\begin{lemma}\label{lemma:Phi-av-prop}
Consider the channels  $\Phi_\av^*$ and $\Phi_\av$ defined in \eq{eq:Phi-av-defn}. These have the following properties:
\begin{enumerate}[(i)]
\item $\Phi^* \circ  \Phi^*_\av=\Phi^*_\av \circ \Phi^* = \Phi^*_\av \circ \Phi^*_\av=\Phi^*_\av$ and $\Phi \circ \Phi_\av = \Phi_\av \circ \Phi = \Phi_\av\circ \Phi_\av = \Phi_\av$.

\item $\Phi^*_\av(\lo(\hs))=\ff(\Phi^*_\av) = \ff(\Phi^*)$ and $\Phi_\av(\trc(\hs)) = \ff(\Phi_\av) = \ff(\Phi)$.
\end{enumerate}
\end{lemma}
\begin{proof}
  (i) is trivial, and so we shall only prove (ii). Let us first consider  the Heisenberg picture channel $\Phi_\av^*$. That  $\ff(\Phi^*) \subset \ff(\Phi^*_\av)$ is trivial. Conversely, for any $A\in \ff(\Phi^*_\av)$,  by (i) we have $\Phi^*(A) = \Phi^*\circ \Phi^*_\av(A)=\Phi^*_\av(A) = A$, and therefore $\ff(\Phi^*_\av) \subset \ff(\Phi^*)$. It follows that $\ff(\Phi^*_\av) = \ff(\Phi^*)$. Similarly,  for all $A\in \lo(\hs)$ it holds that $\Phi^*_\av\circ \Phi^*_\av(A) 
=\Phi^*_\av(A)$, and thus  $\Phi^*_\av(\lo(\hs)) \subset \ff(\Phi^*_\av)$. That $\ff(\Phi^*_\av) \subset \Phi^*_\av(\lo(\hs))$ is trivial, and so we also have $\Phi^*_\av(\lo(\hs)) = \ff(\Phi^*_\av)$. The relations in (ii) for the Schr\"odinger picture channel $\Phi_\av$ follow from similar arguments.
\end{proof}

Now consider the state
\begin{align}\label{eq:rho-0}
   \rho_0:= \Phi_\av\left(\frac{1}{d} \onesys\right). 
\end{align}
By \lemref{lemma:Phi-av-prop}, it holds that $\rho_0 \in \ff(\Phi_\av)= \ff(\Phi)$. We define by $P$ the minimal support projection on $\rho_0$:
\begin{align}\label{eq:P-defn}
   P:= \min \{Q : Q \text{ is a projection, } \rho_0 = Q \rho_0 Q \}. 
\end{align}
In other words, for all projections $Q$ such that $\rho_0 = Q \rho_0 Q$, it holds that $Q \geqslant P$. Note that if $P=\onesys$ then $\ff(\Phi)$ contains a faithful state.  The following lemma provides some useful properties of $P$.
\begin{lemma}\label{lemma:P-properties}
Consider the state $\rho_0$  defined in \eq{eq:rho-0}, with  the minimal support projection  $P$ as defined in \eq{eq:P-defn}, and $P^\perp := \onesys - P$ its orthogonal complement. The following hold:
\begin{enumerate}[(i)]
    \item $\Phi^*_\av(P)=\onesys$ and $\Phi_\av^*(P^\perp) = \zero$.

\item For all $A \in \lo(\hs)$, $\Phi^*_\av(A) = \Phi^*_\av(P A P)$.

\item  $P=\min \{Q:  Q\mbox{ is a projection, } \rho =  Q \rho  Q \,\ \forall \,  \rho \in \ff(\Phi)\}$.

\item  $P=\min \{Q:  Q\mbox{ is a projection, } \Phi^*_\av(Q)=\onesys$\}. 

\item  $\Phi^*(P)\geqslant P$ and $\Phi^*(P^{\perp}) \leqslant P^{\perp}$.  

\end{enumerate}
\end{lemma}
\begin{proof}
\begin{enumerate}[(i):]
    \item Since $\Phi_\av^*$ is a channel, and $\zero < P \leqslant \onesys$, it follows that $\zero \leqslant \Phi_\av^*(P) \leqslant \onesys$. But  by \eq{eq:rho-0} $\tr[ \Phi^*_\av(P)] = d \,  \tr[\rho_0 P] =d$, and so $\Phi_\av^*(P) = \onesys$. It trivially follows that $\Phi_\av^*(P^\perp) = \Phi_\av^*(\onesys) - \Phi_\av^*(P) = \onesys - \onesys = \zero$. 
    
    \item Since $P$ is positive, then by (i) and \lemref{lemma:operation-annihilation} it holds that  $ \Phi^*_\av(P^{\perp}A)  = \Phi^*_\av(AP^{\perp}) = \zero$ for all $A$. The claim follows by noting that we may write $A = (P + P^\perp) A (P + P^\perp)$. 

\item By \lemref{lemma:Phi-av-prop} and (ii), for all $A \in \lo(\hs)$ and $\rho \in \ff(\Phi)$ we have $\tr[A \rho] = \tr[\Phi_\av^*(A) \rho] = \tr[\Phi_\av^*(P A P) \rho] = \tr[A P \rho P]$, and so $\rho \in \ff(\Phi) \implies \rho = P \rho P$. Since $P$ is the minimal support projection on $\rho_0 \in \ff(\Phi)$, the claim follows. 

\item By (i), $\Phi^*_\av(P)=\onesys$ holds. Suppose another projection $Q$ satisfies $\Phi^*_\av(Q)=\onesys$. Then (ii) implies that $ \Phi^*_\av(PQP) = \Phi^*_\av(Q) =\onesys$. As we have $1=\tr[(\onesys/d)\Phi^*_\av(PQP)]
=\tr[\rho_0 PQP]=\tr[\rho_0 Q]$, it follows that $Q\geqslant P$. 

\item Since $P$ is the smallest projection satisfying $\rho_0 = P \rho_0 P$, while  $\tr[\rho_0 \Phi^*(P)] =\tr[\Phi(\rho_0) P]=\tr[\rho_0 P]=1$, it follows that $\Phi^*(P)  \geqslant  P$, and hence  $\Phi^*(P^\perp) = \Phi^*(\onesys) - \Phi^*(P) \leqslant \onesys - P = P^\perp$.

\end{enumerate}

\end{proof}

Let us now define the operations
\begin{align}\label{eq:Phi-P-av}
&\Phi^*_{\av, P}(\cdot) := P\Phi^*_\av(\cdot)P, &\Phi^*_P(\cdot) := P\Phi^*(\cdot)P.     
\end{align}
Note that $\Phi_{\av,P}^* : \lo(\hs) \to \lo( \hs)$  is not necessarily unital, since
$\Phi^*_{\av,P}(\onesys)=P \leqslant  \onesys$. However, the restriction of $\Phi^*_{\av,P}$ to $\lo(P\hs) \to \lo(P \hs)$, 
 which is also denoted by $\Phi^*_{\av,P}$, is unital and hence a channel, since $P$ is the identity in $\lo(P \hs)$ and  $\Phi^*_{\av,P}(P) =  P$. The same holds for $\Phi_{P}^*$.

\begin{lemma}\label{lemma:fav1}
Consider the operations defined in \eq{eq:Phi-P-av}. The following hold:
\begin{enumerate}[(i)]
\item $\Phi^*_P(A) = \Phi^*_P(P A P)$ for all $A \in \lo(\hs)$.

 \item $\Phi^*_{\av, P}(A) = \Phi^*_{\av, P}(PAP)$ for all $A \in \lo(\hs)$.

\item $\Phi^*_{\av, P}$ is a completely positive projection
$\lo(\hs) \to \lo(P\hs)$. 
\item $\ff(\Phi^*_{\av,P})= \Phi^*_{\av,P}(\lo(P\hs))$.  
 \item   $\Phi^*_\av$ is a bijection from 
 $\ff(\Phi^*_{\av,P})$ to $\ff(\Phi^*)$.
\item The inverse of $\Phi^*_\av: 
\ff(\Phi^*_{\av,P}) \to \ff(\Phi^*)$ is $\Ad_P: A \mapsto P A P$.
\end{enumerate}
\end{lemma}
\begin{proof}
\begin{enumerate}[(i):]
\item By item (v) of \lemref{lemma:P-properties}, it holds that $\zero \leqslant \Phi^*_P(P^\perp) = P \Phi^*(P^\perp) P \leqslant P P^\perp P =\zero $, and so $\Phi^*_P(P^\perp) = \zero$. By \lemref{lemma:operation-annihilation},  it holds that  $\Phi_P^*(P^\perp A) = \Phi_P^*(A P^\perp)= \zero$ for all $A \in \lo(\hs)$. Since $A = (P + P^\perp) A (P + P^\perp)$, the claim follows.
 
 \item By item (ii) of \lemref{lemma:P-properties}, $\Phi_{\av}^*(A) = \Phi_{\av}^*(P A P)$. The claim immediately follows.

\item By item (i) of \lemref{lemma:Phi-av-prop} and (ii) above, we have
\begin{align*}
\Phi^*_{\av, P} \circ \Phi^*_{\av, P}(A) = \Phi^*_{\av, P} (P \Phi^*_{\av}(A) P) = \Phi^*_{\av, P} \circ \Phi^*_{\av}(A)
=P\Phi^*_\av\circ\Phi^*_\av(A)P = P\Phi^*_\av(A)P =\Phi^*_{\av, P}(A).    
\end{align*}

\item   By (iii), it follows  that for any $A \in \lo(P \hs)$, $\Phi^*_{\av,P}(A) \in \ff(\Phi^*_{\av,P})$ and so $\Phi^*_{\av,P}(\lo(P\hs)) \subset \ff(\Phi^*_{\av,P})$. The converse is trivial.

\item  For all $A \in \ff(\Phi^*) = \ff(\Phi^*_\av)$, there exists an operator $P A P \in \ff(\Phi_{\av,P}^*)$ such that $\Phi_\av^*(P A P) = \Phi_\av^*(A) = A$. Therefore, $\Phi_\av^*$ is surjective.  Now assume that there exists  $A \in \ff(\Phi_{\av, P}^*)$ such that  $\Phi^*_\av(A)=\zero$. This  implies that $A = \Phi_{\av, P}^*(A) = P\Phi_\av^*(A)P = \zero$. Therefore, $\Phi_\av^*$ is injective. 
\item Follows from above.
\end{enumerate}
\end{proof}

The above results have the following useful consequence:

\begin{prop}\label{prop:fav2}
Consider the operations  $\Phi_{\av,P}^*$ and $\Phi^*_P$  defined in \eq{eq:Phi-P-av}. 
It holds that 
\begin{align*}
\ff(\Phi^*_{\av,P})= \ff(\Phi^*_P) =  P \ff(\Phi_\av^*)P = P\ff(\Phi^*)P     
\end{align*} 
is a von Neumann algebra in $\lo(P \hs)$.
\end{prop}
\begin{proof}
Recall that the operation $\Phi_{\av,P}^*: \lo(P\hs) \to \lo(P\hs)$ is unital, where the unit in $\lo(P\hs)$ is $P$. Moreover, $\rho_0$ as defined in \eq{eq:rho-0} is a faithful fixed point of $\Phi_{\av,P}$ in $\trc(P \hs)$. By \lemref{lemma:Lindblad}, the fixed-point set $\ff(\Phi^*_{\av,P})$ is a   von Neumann algebra in $\lo(P\hs)$. Now  we need only show that $\ff(\Phi^*_{\av,P})= \ff(\Phi^*_P) = P \ff(\Phi_\av^*)P = P\ff(\Phi^*)P$. That $ P \ff(\Phi_\av^*)P = P\ff(\Phi^*)P$ trivially follows from \lemref{lemma:Phi-av-prop}, which gives $\ff(\Phi^*) = \ff(\Phi_\av^*)$.   That  $\ff(\Phi^*_{\av,P})=  P\ff(\Phi^*)P$ follows from item (vi) of \lemref{lemma:fav1}, since the map $\Ad_P : A \mapsto PAP$ is a bijection from $\ff(\Phi^*)$ to $\ff(\Phi^*_{\av,P})$.  To show that  $\ff(\Phi^*_{\av,P})= \ff(\Phi^*_P) $, let us first define the  operation  $\Phi_{P, \av}^*:\lo(\hs) \to \lo(P\hs)$ as
\begin{align*}
    \Phi_{P, \av}^*(\cdot) &:= \lim_{N \to \infty} \frac{1}{N} \sum_{n=1}^N \Phi_P^{*n}(\cdot).
\end{align*}
 But by item (i) of \lemref{lemma:fav1},   $\Phi^*_P(A) = \Phi^*_P(P A P)$, and so  $\Phi_P^{*n}(A) =  P \Phi^{*n}(A) P$ for all $A \in \lo(\hs)$ and $n \in \nat$. It follows that
\begin{align*}
  \Phi_{P, \av}^*(\cdot)   & = \lim_{N \to \infty} \frac{1}{N} \sum_{n=1}^N P\Phi^{*n}(\cdot)P = P \Phi_{\av}^*(\cdot) P =: \Phi^*_{\av, P}(\cdot).
\end{align*}
Since $\Phi_{\av, P}^* = \Phi_{P, \av}^*$,   it trivially follows that  $\Phi_P^* \circ \Phi_{\av, P}^* = \Phi_{\av, P}^* \circ \Phi_P^* = \Phi_{\av, P}^* \circ \Phi_{\av, P}^* = \Phi_{\av, P}^*$. Therefore, $\ff(\Phi^*_{\av,P})  = \ff(\Phi^*_P)$ can be shown  by the same arguments as in  \lemref{lemma:Phi-av-prop}

\end{proof}

%%%%%%%%%%

\subsection{Measurement disturbance revisited}\label{app:non-disturbance-fixed-point}

We are now ready to address the question of measurement disturbance, generalising the observations of \thmref{theorem:faithful-fixedpoint-WAY}. As before, let $\ii:= \{\ii_x: x\in \xx\}$ be an $\E$-compatible instrument, with $\ii_\xx(\cdot) := \sum_x \ii_x(\cdot)$ the corresponding $\E$-channel. By \eq{eq:Phi-av-defn} and \eq{eq:Phi-P-av} we define 
\begin{align*}
    \ii_\av^*(\cdot) &:= \lim_{N \to \infty} \frac{1}{N} \sum_{n=1}^N \ii_\xx^{*n}(\cdot), \qquad 
    \ii_{\av,P}^*(\cdot) := P \ii_\av^*(\cdot) P, \qquad \ii_P^*(\cdot) := P \ii_\xx^*(\cdot) P,
\end{align*}
where as in \eq{eq:P-defn}, $P$ is the minimal support projection of $\rho_0 := \ii_\av(\frac{1}{d} \onesys)$, which corresponds with the minimal projection on the support of $\ff(\ii_\xx)$. By  \propref{prop:fav2}, $\ff(\ii_{\av,P}^*) = \ff(\ii_P^*) =P \ff(\ii^*_\av)P = P \ff(\ii^*_\xx)P$ is a von Neumann algebra in $\lo(P \hs)$. We define by $P \E P := \{P \E(x) P : x\in \xx\}$  the restriction of  $\E$ to an observable in  $P\hs$, which satisfies $\sum_x P \E(x) P = P$, and  $(P \E P) ' := \{A \in \lo(P \hs) : [P \E(x) P, A] = \zero \, \forall x\}$ denotes the commutant of $P \E P$ in $\lo(P \hs)$. $P \F P$ is similarly defined.

Before generalising \thmref{theorem:faithful-fixedpoint-WAY} for the case where $\ff(\ii_\xx)$ may not contain any faithful states, and  thus $\ff(\ii_\xx^*)$ may not necessarily be a von Neumann algebra, let us first prove a generalisation of  \lemref{lemma:faithful-fixedpoint-commutant}. 
\begin{lemma}\label{lemma:non-faithful-fixedpoint-commutant}
Let   $\mm:= (\ha, \xi, \ee, \Z)$ be a measurement scheme for an $\E$-instrument $\ii$ acting in $\hs$, and let $P$ be the minimal support projection on $\ff(\ii_\xx)$. It holds that $\ff(\ii_P^*) \subset (P \E P) '$.  Additionally, if $\ee$  conserves an additive quantity $N = N\sub{\s} \otimes \oneapp + \onesys\otimes N\sub{\aa}$ on average, where $\nsys \in \los(\hs)$ and $\napp \in \los(\ha)$, then  for all $A\in \lo(\P \hs)$ the following implication holds:   $A\in \ff(\ii_P^*) \implies [A,  P N\sub{\s} P ] \in \ff(\ii_P^*) \implies [A,  P N\sub{\s} P ] \in (P\E P)'$.
\end{lemma}
\begin{proof}
 By \eq{eq:Gamma-U}, let us define the   operation $\Gamma_{\xi,P}^\ee : \lo(\hs\otimes \ha) \to \lo(\hs), B \mapsto P \Gamma_\xi^\ee(B) P$. It is easily verified that $\ii_P^*(A) =  \Gamma_{\xi,P}^\ee(A \otimes \oneapp)$ for all $A \in \lo(P\hs)$, and $ P \E(x) P = \Gamma_{\xi,P}^\ee(\onesys \otimes \Z(x))$. Note that by the same arguments as item (i) of  \lemref{lemma:fav1}, it can easily be shown that $\Gamma_{\xi,P}^\ee(B) = \Gamma_{\xi,P}^\ee((P \otimes \oneapp)B (P \otimes \oneapp))$ holds for all $B \in \lo(\hs\otimes \ha)$. It follows that $\Gamma_{\xi, P}^\ee$  is  unital  when restricted to $\lo(P\hs\otimes \ha) \to \lo(P \hs)$, and we may  equivalently write $ P \E(x) P = \Gamma_{\xi,P}^\ee(P \otimes \Z(x))$.
 
  By  \propref{prop:fav2}, $\ff(\ii_P^*) $ is a von Neumann algebra in $\lo(P \hs)$, and so for all $A \in \lo(P \hs)$, if $A \in \ff(\ii_P^*)$, then $A^*A, A A^* \in \ff(\ii_P^*)$.  By \corref{corollary:multiplicability} it holds that for all $A \in \ff(\ii_P^*)$ and $B \in \lo(P \hs \otimes \ha)$ we have  $  A\Gamma_{\xi, P}^\ee (B) = \Gamma_{\xi, P}^\ee ((A \otimes \oneapp)B)$ and $  \Gamma_{\xi, P}^\ee (B)A = \Gamma_{\xi, P}^\ee (B(A \otimes \oneapp))$. It follows that for all $A \in \ff(\ii_P^*)$ and $x\in \xx$ we have
  \begin{align*}[ A, P \E(x) P] = [A, \Gamma_{\xi,P}^\ee(P \otimes \Z(x))]= \Gamma_{\xi,P}^\ee([A \otimes \oneapp, P \otimes \Z(x)]) = \zero,
  \end{align*}
and so  $ \ff(\ii^*_P) \subset (P \E P) '$.

 Now let us assume that $\ee$  conserves an additive quantity $N = N\sub{\s} \otimes \oneapp + \onesys\otimes N\sub{\aa}$ on average. This  implies that $P \Gamma_\xi(N) P = \Gamma_{\xi,P}^\ee(N) = \Gamma_{\xi,P}^\ee((P \otimes \oneapp)N(P \otimes \oneapp))$, and so  
\begin{eqnarray*}
P \nsys P + \tr[\napp \xi] P
= \Gamma_{\xi,P}^\ee(P \nsys P \otimes \oneapp)  + \Gamma_{\xi,P}^\ee(P \otimes \napp). 
\end{eqnarray*}
 Since $ \ff(\ii^*_P)$ is a von Neumann algebra in $\lo(P \hs)$, then  by \corref{corollary:multiplicability} and the arguments above, it follows that for all $A \in \ff(\ii_P^*)$ we have
\begin{align*}
    [A, P \nsys P] &= [ A, \Gamma_{\xi,P}^\ee(P \nsys P \otimes \oneapp)] + [ A,\Gamma_{\xi,P}^\ee(P \otimes \napp)] \nonumber \\
    & = \Gamma_{\xi,P}^\ee([A \otimes \oneapp, P\nsys P \otimes \oneapp]) + \Gamma_{\xi,P}^\ee([A \otimes \oneapp, P \otimes \napp ]) \nonumber \\
    & = \Gamma_{\xi,P}^\ee([A, P \nsys P]\otimes \oneapp) = \ii^*_P([A, P \nsys P]).
\end{align*}
We thus have $A \in \ff(\ii^*_P) \implies [A, P\nsys P] \in \ff(\ii^*_P)$, and since $ \ff(\ii^*_P) \subset (P \E P) '$, it follows that $[A, P\nsys P] \in (P \E P)'$. 
\end{proof}

We are now ready to generalise \thmref{theorem:faithful-fixedpoint-WAY}. 
\begin{theorem}\label{theorem:non-faithful-fixedpoint-WAY}
Let $\E :=\{\E(x) : x\in \xx\}$ and $\F:=\{\F(y) : y \in \yy\}$ be observables acting in $\hs$. Let $\mm:=(\ha, \xi, \ee, \Z)$ be a measurement scheme for an  $\E$-instrument $\ii$, and assume that $\ee$ conserves an additive quantity $N = \nsys \otimes \oneapp + \onesys \otimes \napp$ on average, where $\nsys \in \los(\hs)$ and $\napp \in \los(\ha)$. If $P$ is  the minimal support projection on $\ff(\ii_\xx)$, then the following hold: 
\begin{enumerate}[(i)]
    \item $\F \subset \ff(\ii_\xx^*)$  only if $P \F P$ commutes with   $P \E P$, with  $P \Delta  \nsys P := P \ii_\xx^*( \nsys ) P - P \nsys P$, and with \\
    $\{[P \E(x) P, P \nsys P] : x\in \xx\}$.
    \item $\ii$  is a measurement of the first kind only if  $P\E P$  is  commutative and commutes with $P \nsys P$. 
    
    \item  $\ii$ is repeatable only if $P \E P$ is sharp and commutes with $P \nsys P$.
\end{enumerate}
\end{theorem}
\begin{proof}
\begin{enumerate}[(i):]
    \item  By \propref{prop:fav2}, $P \ff(\ii_\xx^*) P = \ff(\ii^*_P)$, and so $\F \subset \ff(\ii_\xx^*)$ implies that $P \F P \subset \ff(\ii^*_P)$.  By \lemref{lemma:non-faithful-fixedpoint-commutant}, it holds that $P \F P$ must commute with $P \E P$, and that  $[ P \F(y) P, P\nsys P] = \ii_P^*([ P \F(y) P, P\nsys P])$.    Given that $\ff(\ii^*_P)$ is a von Neumann algebra, we have $\ii_P^*((P \F(y) P)^2) = \ii_P^*(P \F(y) P)^2 = (P \F(y) P)^2$.  By \corref{corollary:multiplicability} it follows that  $[ P \F(y) P, P\nsys P] = [ P \F(y) P, \ii_P^*(P\nsys P)]$. By item (i) of \lemref{lemma:fav1}, it holds that $\ii_P^*(P\nsys P) = \ii_P^*(\nsys) = P \ii_\xx^*(\nsys) P$. Therefore, $P \F P$ must commute with $P \Delta  \nsys P$. Moreover, by \lemref{lemma:non-faithful-fixedpoint-commutant} non-disturbance implies that $[ P \F(y) P, P\nsys P] \in (P \E P)'$, and since $P \F P$ commutes with $P \E P$, this implies that $[P \F(y) P, [P \E(x) P, P \nsys P]] = \zero$ must hold. 
    
    \item If $\ii$ is a measurement of the first kind, $\E \subset \ff(\ii^*_\xx)$,  it follows from \propref{prop:fav2} that  $P \E P \subset \ff(\ii^*_P)$.  By \lemref{lemma:non-faithful-fixedpoint-commutant} $P \E P $ must be commutative, and since $\ff(\ii^*_P)$ is a von Neumann algebra, then for all spectral projections $R \in \lo(P \hs)$ of $P \E(x) P$ it must hold that  $[R, P \nsys P]$ commutes with $P \E P$. But this implies that $[[R, P \nsys P], R^\perp]=\zero$, where we define $R^\perp := P - R$. As such, given that $P R = R P = R$,  it holds that  $[R  \nsys P -P \nsys R,R^\perp] = R \nsys R^\perp + R^\perp \nsys R = \zero$. Multiplying from the left by $R$, we thus have $R \nsys R^\perp = R \nsys (P - R) = \zero$,  and so $R  \nsys P  = R \nsys R$. Since the right hand side is self-adjoint, and $R  \nsys P = R P \nsys P$, it follows that  $[R, P \nsys P ]=\zero$.  Since this relation holds for all  spectral projections $R$ of all effects of $P \E P$, it follows that $P \E P$ must commute with $P \nsys P$.

    \item Commutativity of $P \E P$ with $P \nsys P$ follows from (ii) and the fact that repeatability implies first-kindness. Sharpness  of $P \E P$ follows from the fact that the fixed points of a repeatable instrument can  have support only in the eigenvalue-1 eigenspaces of $\E$, as shown in  \propref{prop:repeatable-instrument-identity}.
\end{enumerate}
 
\end{proof}

Note that if $P=\onesys$, implying that $\ff(\ii_\xx)$ contains a faithful state so that $\ff(\ii_\xx^*)$ is a von Neumann algebra, then the above theorem reduces to 
 \thmref{theorem:faithful-fixedpoint-WAY}. Interestingly,  in the case of qubits such an equivalence will always hold, even if $\ff(\ii_\xx)$ does not contain a faithful state. We demonstrate this by an alternative proof for Proposition 6 in Ref. \cite{Heinosaari2010}.

\begin{corollary}\label{corollary:first-kind-WAY-qubit}
If $\dim(\hs) = 2$, then \thmref{theorem:non-faithful-fixedpoint-WAY} reduces to \thmref{theorem:faithful-fixedpoint-WAY}.
\end{corollary}
\begin{proof}
 Let us consider the minimal support projection $P$ on the fixed-point set $\ff(\ii_\xx)$. As all channels must have at least one fixed state, then when $\dim(\hs)=2$ it holds that,  for any instrument $\ii$,  either   $P= \onesys$  or  $P$ is a rank-1 projection.  If $P=\onesys$ then     $\ff(\ii_\xx)$ contains a  faithful state, so that by \lemref{lemma:Lindblad} $\ff(\ii_\xx^*)$ is a von Neumann algebra, and \thmref{theorem:non-faithful-fixedpoint-WAY} reduces to \thmref{theorem:faithful-fixedpoint-WAY}. Now assume that $ P $ is a rank-1 projection, so that for all $A \in \lo(\hs)$ it holds that $P A P = \lambda P$, with  some $\lambda \in \co$. Recall from \lemref{lemma:Phi-av-prop} that $\ff(\ii_\xx^*) = \ff(\ii_\av^*)$. By item (ii) of \lemref{lemma:P-properties},   for all  $A\in \ff(\ii_\xx^*)$ it holds that $A = \ii_\av^*(A)=\ii_\av^*(P A P) = \lambda \ii_\av^*( P) = \lambda \onesys$, and so $\ff(\ii_\xx^*)$ is a trivial von Neumann algebra containing only operators proportional to the identity. In this case, the only non-disturbed observables are trivial, and will clearly commute with all of $\lo(\hs)$. In such a case we may simply replace $P$ with $\onesys$ in items (i)-(iii) of  \thmref{theorem:non-faithful-fixedpoint-WAY}, so that it reduces to \thmref{theorem:faithful-fixedpoint-WAY}.  
  
\end{proof}

\subsection{Non-disturbance and distinguishability}\label{app:non-disturbance-distinguishability}

Here, we present some novel results regarding the structure of non-disturbing measurements that go beyond those in the preceding sections, indicating an intimate relationship between non-disturbance and distinguishability. These results hold for general instruments, are independent of  conservation laws, and do not explicitly depend on the support projection $P$ on the fixed-point set of the specific measurement channel $\ii_\xx$. First, let us show that if an instrument $\ii$ does not disturb a non-trivial observable, then there exists a family of distinguishable states that remain distinguishable after a non-selective measurement by $\ii$.  

\begin{prop}[Non-disturbance implies distinguishability]\label{prop:nondisturbance-distinguishability} 
Consider an  instrument $\ii$ acting in $\hs$, and assume that $\ii$ does not  disturb a non-trivial observable $\F$. Then there exists a norm-1 observable $\G:=\{\G(z) : z\in \zz\}$ acting in $\hs$ that is non-disturbed by $\ii$,  so that for every family of  states $\{\rho_z : z\in \zz\}$ that are perfectly distinguishable by a $\G$ measurement,   $\{\ii_\xx(\rho_z) : z \in \zz\}$ remain perfectly distinguishable by a $\G$ measurement. Moreover, if $\ff(\ii_\xx)$ contains a faithful state, then $\G$ can be  taken as a sharp observable.
\end{prop}
\begin{proof}
Suppose that a non-trivial observable $\F:= \{\F(y) : y \in \yy\}$ is non-disturbed by  $\ii$, i.e, $\F \subset \ff(\ii^*_\xx) = \ff(\ii_\av^*)$.  Given that $\ff(\ii^*_{\av,P}) = P \ff(\ii^*_\av) P $ (see \propref{prop:fav2}),  this implies that $P \F(y) P \in \ff(\ii^*_{\av,P})$. That $\F$ is non-trivial implies that there must be a $y$ for which $P \F(y) P$ is not proportional to $P$. If that were the case every $P \F(y)P$ could be written as $P\F(y)P=c_y P$
with some $c_y\geqslant 0$, which would imply that $\F(y)=\ii^*_\av(\F(y)) = \ii_\av^*(P \F(y) P)=c_y \ii_\av^*(P) = c_y \onesys$, and so $\F$ would be a  trivial observable. 
Therefore $\ff(\ii^*_{ \av, P}) = P\ff(\ii^*_\av)P$ is a nontrivial von Neumann algebra in $\lo(P \hs)$, and there exists a family of projections $\R:= \{\R(z):z \in \zz\}\subset \ff(\ii^*_{\av,P})$ satisfying $\R(z) \R(y) = \delta_{z,y} \R(z)$ and $\sum_z \R(z) = P$. We may consider $\R$ as a sharp observable acting in  $P \hs$. 

 Using $\R$, we may define a (generally unsharp) observable $\G$ acting in $\hs$ by
\begin{align*}
 \G := \{\G(z) = \ii^*_\av(\R(z)): z \in \zz\},   
\end{align*} 
where   $\sum_z \G(z) = \ii_\av^*(P) = \onesys$. Given that $\ii_\xx^* \circ \ii_\av^* = \ii_\av^*$ (see \lemref{lemma:Phi-av-prop}), it follows that   $\G \subset \ff(\ii_\xx^*)$, i.e., $\G$ is non-disturbed by $\ii$. Moreover, since $ P \G(z) P = \ii^*_{\av,P}(\R(z))= \R(z)$ holds, then each (non-zero) effect of $\G$ has at least one eigenvector with eigenvalue $1$, and so $\| \G(z)\|=1$. Moreover, if $\ff(\ii_\xx)$ contains a faithful state, then $P=\onesys$, and so we have $\G(z) = P \G(z) P =  \R(z)$, implying that $\G \equiv \R$ is sharp. Now let us note that  the family of states $\{\rho_z  \}$ are perfectly distinguishable given a measurement of $\G$ if and only if $\rho_z = \P(z) \rho_z \P(z)$, where $\P(z) \geqslant \R(z)$ projects onto the eigenvalue-1 eigenspace of $\G(z)$. In such a case  it trivially holds that  $\tr[\G(y) \rho_z] =  \delta_{z,y}$. But  Since $\G \subset \ff(\ii^*_\xx)$, we also have 
\begin{align*}
\tr[\G(y) \ii_\xx(\rho_z)] = \tr[\ii_\xx^*(\G(y))\rho_z] = \tr[\G(y) \rho_z] = \delta_{z,y},
\end{align*}
and so $\{\ii_\xx(\rho_z)\}$ continues to be perfectly distinguishable by a $\G$ measurement. 
\end{proof}

In the special case where $\ii$ is a measurement of the first kind, we may strengthen the above result as follows: 

\begin{theorem}\label{thm:first-kind-post-processing-norm-1}
Let $\ii$ be an instrument compatible with a non-trivial observable $\E$ acting in $\hs$. If $\ii$ is a measurement of the first kind, then  
$\E$ is described by a classical post-processing of a norm-1 observable $\G :=\{\G(z) : z \in \zz\}$ with properties given in \propref{prop:nondisturbance-distinguishability}, i.e., 
\begin{align}\label{eq:post-processing-first-kind}
\E(x) = \sum_z p(x|z) \G(z),  
\end{align} 
where $\{p(x|z)\}$ is a family of non-negative numbers that satisfy $\sum_x p(x|z)=1$ for each $z$.  
\end{theorem}
\begin{proof}
Assume that the $\E$-instrument $\ii$ is a measurement of the first kind, that is, $\E \subset \ff(\ii^*_\xx)$. It follows that  $P \E P \subset \ff(\ii^*_{\av,P})$. In fact, we can show that    $P \E P \subset \ff(\ii^*_{\av,P})  \cap \ff(\ii^*_{\av,P})'$. First, recall from \propref{prop:nondisturbance-distinguishability} that if $\E$ is non-trivial then there exists a family of projections $\R := \{\R(z) : z\in \zz\} \subset \ff(\ii^*_{\av,P})$, satisfying $\R(z) \R(y) = \delta_{z,y} \R(z)$ and $\sum_z \R(z) = P$. If $P \E P \not\subset \ff(\ii^*_{\av,P})  \cap \ff(\ii^*_{\av,P})'$, then $\R$ can be chosen so that $[P\E(x)P, \R(z)]\neq \zero$ for some  $x$ and $z$. But note that $\R \subset \ff(\ii^*_{\av,P}) = P \ff(\ii^*_\xx) P$ implies that the observable $\{P\ii_x^*(\R(z))P, P \ii_x^*(\onesys - \R(z)) P : x\in \xx\}$ is  a joint measurement for $P \E P$ and the sharp observable $\{\R(z), P - \R(z)\}$. By compatibility, it follows that $[P\E(x)P, \R(z)]= \zero$ must hold for all $x$ and $z$, and so   $P \E P$ must be contained in the Abelian algebra $\ff(\ii^*_{\av,P}) \cap \ff(\ii^*_{\av,P})'$. Consequently,  $\R \subset \ff(\ii^*_{\av,P}) \cap \ff(\ii^*_{\av,P})'$ can be chosen so as to simultaneously diagonalise all  $P \E(x) P$, that is, we may write $P \E(x) P=\sum_z p(x|z) \R(z)$. Recalling that $\E(x) = \ii_\av^*(\E(x)) = \ii_\av^*(P \E(x) P)$,  then defining the observable $\G$ by $\G(z) = \ii^*_\av(\R(z))$ gives us \eq{eq:post-processing-first-kind}. As in  \propref{prop:nondisturbance-distinguishability} it holds that $\G \subset \ff(\ii_\xx^*)$,  $\G$ is a norm-1 observable, if $\ff(\ii_\xx)$ contains a faithful  state then $\G$ is also sharp, and if $\{\rho_z\}$ are perfectly distinguishable by a $\G$ measurement then so are $\{\ii_\xx(\rho_z)\}$.    
\end{proof}

Finally, we present the following implication of the  above theorem: 
\begin{corollary}\label{corollary:first-kind-fidelity-norm}
Let $\ii$ be an  instrument compatible with a non-trivial observable $\E$ acting in $\hs$, and assume that $\ii$ is a measurement of the first kind. For any outcome $x$ associated with a non-trivial effect $\E(x)$, and for any pair of unit vectors $\psi, \phi \in \hs$ satisfying $ \E(x)\psi =\| \E(x)\| \psi$ and $(\onesys - \E(x))\phi  =\| \onesys  - \E(x)\| \phi$, respectively, it holds that $\psi$ and $\phi$ are orthogonal, and that $F(\ii_\xx(\pr{\psi}), \ii_\xx(\pr{\phi}))=0$, where $F(\rho, \sigma):= \tr[\sqrt{\sqrt{\rho} \sigma \sqrt{\rho}}]$ is the fidelity between states $\rho$ and $\sigma$. 
\end{corollary}
\begin{proof}
For each outcome $x$, we may coarse-grain $\E$ into a binary observable $\{\E(x), \E(\overline{x}) := \onesys - \E(x)\}$. By \thmref{thm:first-kind-post-processing-norm-1}, first-kindness implies that $\E(x) = \sum_z p(x|z) \G(z)$ and $\E(\overline{x}) = \sum_z p(\overline{x}|z) \G(z)$, with $p(\overline{x}|z)  := 1 - p(x|z)$, where $\{\G(z) : z \in \zz\}$ is a norm-1 observable with properties given in \propref{prop:nondisturbance-distinguishability}, while $\{p(x|z)\}$ is a family of non-negative numbers satisfying $\sum_{x'} p(x'|z)=1$ for each $z$. Now let us define $p_{\max}(x) := \max_z\{p(x|z)\}$ and $p_{\min}(x) := \min_z\{p(x|z)\}$ as the maximum and minimum values of the set $\{p(x|z)\}$. We may thus define the sets $Z_{\max} := \{z \in \zz: p(x|z) = p_{\max}(x)\}$ and $Z_{\min} := \{z \in \zz: p(x|z) = p_{\min}(x)\}$. Using such sets, we may define $\G(Z_{\max}) := \sum_{z\in Z_{\max}} \G(z)$ and $\G(Z_{\min}) := \sum_{z\in Z_{\min}} \G(z)$. Since $\G$ is norm-1, then we may also define $\P(Z_{\max}) := \sum_{z\in Z_{\max}} \P(z)$ and $\P(Z_{\min}) := \sum_{z\in Z_{\min}} \P(z)$, where $\P(z)$ is the projection onto the eigenvalue-1 eigenspace of $\G(z)$. Since $\E(x)$ is assumed to be non-trivial, then it must hold that $Z_{\max} \cap Z_{\min} = \emptyset$. If this were not so, it would hold that all $p(x|z)$ are the same, in which case $\E(x) \propto \onesys$. Consequently, $\P(Z_{\max}) \P(Z_{\min}) = \zero$.  

Now let us note that 
\begin{align*}
\| \E(x) \| = \sup_{\|\psi\|=1} \<\psi| \E(x) \psi\> =  \sup_{\|\psi\|=1}  \sum_{z} p(x|z) \<\psi| \G(z) \psi\>.    
\end{align*} 
We may now  show that a unit vector $\psi$ satisfies $\E(x) \psi = \| \E(x) \| \psi$ if and only if $\P(Z_{\max}) \psi = \psi$. Let us first prove the only if statement. For any unit vector $\psi$, it holds that $\<\psi| \E(x) \psi\>  \leqslant   p_{\max}(x)$, which follows from the fact that $p(x|z)$ are positive numbers and that $\{\<\psi| \G(z) \psi\>\}$ is a probability distribution, with the upper bound being saturated when $\<\psi| \G(Z_{\max}) \psi\> = 1$.  But this in turn is satisfied only if $\<\psi| \P(Z_{\max}) \psi\> = 1$, in which case $\P(Z_{\max}) \psi = \psi$. As such, it follows that $\|\E(x)\| = p_{\max}(x)$, and  the unit vector $\psi$ satisfies $\E(x) \psi = \| \E(x)\| \psi$ only if $\P(Z_{\max}) \psi = \psi$. The if statement is trivial. 

By similar arguments as above, we may show that $\|\onesys - \E(x) \| = \| \E(\overline{x}) \| = p_{\max}(\overline{x}) =1 - p_{\min}(x)$, and that   the unit vector $\phi$ satisfies $(\onesys - \E(x))\phi  =\| \onesys  - \E(x)\| \phi$ if and only if $\P(Z_{\min})\phi = \phi$.  Since $\E(x)$ is non-trivial, then as argued above  $\psi$ and $\phi$ are orthogonal, and perfectly distinguishable by a $\G$ measurement. By \propref{prop:nondisturbance-distinguishability} it holds that $\ii_\xx(\pr{\psi})$ and $\ii_\xx(\pr{\phi})$ are also perfectly distinguishable by a $\G$ measurement, that is,  $F(\ii_\xx(\pr{\psi}), \ii_\xx(\pr{\phi}))=0$.

\end{proof}

\subsection{Measurements of the first kind, distinguishability, and the Wigner-Araki-Yanase theorem}\label{app:first-kind-WAY}

We shall now use the results in the preceding section   to obtain   quantitative bounds for first-kind measurements in the presence of a conservation law,  that complement our generalisation of the WAY theorem  given in  \thmref{theorem:Strong-WAY}.  To this end, let us first provide a generalisation of Theorem 2 in Ref. \cite{Miyadera2006a}, which we shall use in the sequel:
\begin{lemma}\label{lemma:WAY-distinguishability-inequality}
Let $\mm:= (\ha, \xi, \ee, \Z)$ be a measurement scheme for an $\E$-instrument $\ii$ acting in $\hs$, and assume that $\ee$  conserves an additive quantity $N = \nsys \otimes \oneapp + \onesys \otimes \napp$ on average, where $\nsys \in \los(\hs)$ and $\napp \in \los(\ha)$. For any pair of orthogonal unit vectors $\psi, \phi \in \hs$, the following will hold:
\begin{align}\label{eq:WAY-distinguishability-inequality}
    |\<\psi| \nsys \phi\>| \leqslant \| \napp \| F(\ii_\xx(\pr{\psi}), \ii_\xx(\pr{\phi})) + \|\nsys\| F(\Lambda(\pr{\psi}), \Lambda(\pr{\phi})), 
\end{align}
where $\Lambda$ is the conjugate channel to $\ii_\xx$ defined in \eq{eq:conjugate-channel}, and $F(\rho, \sigma)$ is the fidelity between states $\rho$ and $\sigma$.
\end{lemma}
\begin{proof}
Let us consider the augmented Hilbert space $\ha \otimes \kk$ so that $\xi\in \s(\ha)$ admits the purification $\xi = \tr\sub{\kk}[\pr{\varphi}]$, with the unit   vector $\varphi \in \ha \otimes \kk$. Moreover, if $\kk$ is sufficiently large, then by Stinespring's dilation theorem the channel $\ee^*$ can be expressed as $\ee^*(A) = V^*(A \otimes \one\sub{\kk})V$ for all $A \in \lo(\hs\otimes \ha)$, where $V: \hs\otimes \ha \to \hs\otimes \ha \otimes \kk$ is an isometry. By additivity of $N$, and orthogonality of $\psi, \phi$, we have
\begin{align*}
    \<\psi \otimes \varphi| N \phi \otimes \varphi\> = \<\psi| \nsys \phi\>. 
\end{align*}
On the other hand, average conservation  of $N$ by $\ee$ implies that 
\begin{align*}
    N = \ee^*(N) = V^*(\nsys \otimes \oneapp \otimes \one\sub{\kk} )V + V^*(\onesys \otimes \napp \otimes \one\sub{\kk} )V.
\end{align*}
We therefore have
\begin{align}\label{eq:inner-prod-WAY-dist-1}
 |\<\psi| \nsys \phi\>| \leqslant |\<\psi\otimes \varphi|V^*(\nsys \otimes \oneapp \otimes \one\sub{\kk} )V\phi \otimes \varphi\>| + |\<\psi \otimes \varphi|V^*(\onesys \otimes \napp \otimes \one\sub{\kk} )V \phi \otimes \varphi\>|.   
\end{align}
For any observable $\G:= \{\G(z)\}$ acting in $\ha$,  we may write
\begin{align*}
&|\<\psi\otimes \varphi|V^*(\nsys \otimes \oneapp \otimes \one\sub{\kk} )V\phi \otimes \varphi\>| \nonumber \\
& \qquad = |\sum_z \<\psi\otimes \varphi|V^*(\nsys \otimes \G(z) \otimes \one\sub{\kk} )V\phi \otimes \varphi\>| \nonumber \\
& \qquad  \leqslant \|\nsys \| \sum_z | \<\psi\otimes \varphi|V^*(\onesys \otimes \G(z) \otimes \one\sub{\kk} )V\phi \otimes \varphi\>| \nonumber \\
&\qquad \leqslant \|\nsys\| \sum_z |\<\psi\otimes \varphi|V^*(\onesys \otimes \G(z) \otimes \one\sub{\kk} )V\psi \otimes \varphi\>|^{\frac{1}{2}} |\<\phi\otimes \varphi|V^*(\onesys \otimes \G(z) \otimes \one\sub{\kk} )V\phi \otimes \varphi\>|^{\frac{1}{2}} \nonumber \\
& \qquad = \|\nsys\| \sum_z \tr[\onesys \otimes \G(z)  \ee(\pr{\psi} \otimes \xi)]^{\frac{1}{2}} \tr[\onesys \otimes \G(z)  \ee(\pr{\phi} \otimes \xi)]^{\frac{1}{2}} \nonumber \\
& \qquad =  \|\nsys\| \sum_z \tr[ \G(z)  \Lambda(\pr{\psi})]^{\frac{1}{2}} \tr[\G(z) \Lambda(\pr{\phi})]^{\frac{1}{2}}.
\end{align*}
In the third line we have used the Cauchy-Schwarz inequality,  in the fourth line we used Stinespring's dilation theorem together with the fact that $\varphi$ is a purification of $\xi$, and in the final line we use the definitions of the partial trace and the conjugate channel $\Lambda$. Now, note that  the fidelity satisfies  $F(\rho, \sigma) = \min_{\G} \sum_z \tr[\G(z) \rho]^{\frac{1}{2}}\tr[\G(z) \sigma]^{\frac{1}{2}}$ \cite{Fuchs1995, Barnum1996}. Therefore, choosing $\G$ so as to obtain the fidelity, we have
\begin{align*}
 |\<\psi\otimes \varphi|V^*(\nsys \otimes \oneapp \otimes \one\sub{\kk} )V\phi \otimes \varphi\>| \leqslant \|\nsys\| F(\Lambda(\pr{\psi}), \Lambda(\pr{\phi})).   
\end{align*}
Using similar steps, we may also write 
\begin{align*}
|\<\psi\otimes \varphi|V^*(\onesys \otimes \napp \otimes \one\sub{\kk} )V\phi \otimes \varphi\>| \leqslant \|\napp\| F(\ii_\xx(\pr{\psi}), \ii_\xx(\pr{\phi})).    
\end{align*} 
By \eq{eq:inner-prod-WAY-dist-1}, we thus obtain \eq{eq:WAY-distinguishability-inequality}.
\end{proof}

We are now ready to prove our main result in this section:

\begin{theorem}\label{theorem:first-kind-WAY-distinguishability}

Consider a measurement scheme $\mm := (\ha, \xi, \ee, \Z)$ for  a nontrivial observable $\E$ with the instrument $\ii$ acting in $\hs$. Assume that $\ii$ is a measurement of the first kind,  and  that $\ee$  conserves an additive quantity $N = N\sub{\s} \otimes \oneapp + \onesys\otimes N\sub{\aa}$ on average, where $\nsys \in \los(\hs)$ and $\napp \in \los(\ha)$.  
For each outcome $x$ associated with a non-trivial effect $\E(x)$, let  $\mathcal{K}_{\max}(x)$ and $\mathcal{K}_{\min}(x)$ be  subspaces of $\hs$ defined by 
\begin{align*}
&\mathcal{K}_{\max}(x):=\{\psi \in \hs : \ \E(x)\psi =\| \E(x)\| \psi\}, &\mathcal{K}_{\min}(x):=\{\phi \in \hs:\ (\onesys - \E(x))\phi  =\| \onesys  - \E(x)\| \phi\}.    
\end{align*}
$\mathcal{K}_{\max}(x)$ and $\mathcal{K}_{\min}(x)$ are orthogonal, and  for all unit vectors $\psi \in \mathcal{K}_{\max}(x)$ and $\phi \in \mathcal{K}_{\min}(x)$, it holds that 
\begin{align}\label{eq:distinguishability-WAY-bound}
|\langle \psi |\nsys \phi\rangle | 
\leqslant \|\nsys\| \left( \| \E(x)\|^{\frac{1}{2}} (1 - \|\onesys - \E(x) \|)^{\frac{1}{2}} + (1 - \| \E(x) \|)^{\frac{1}{2}} \|\onesys - \E(x) \|^{\frac{1}{2}}\right). 
\end{align}
\end{theorem}

\begin{proof}

 For each outcome $x$ associated with a non-trivial effect $\E(x)$, we may coarse-grain $\E$ into a binary observable $\{\E(x), \E(\overline{x}) := \onesys - \E(x)\}$. By \corref{corollary:first-kind-fidelity-norm}, it holds that for any unit vectors $\psi \in \mathcal{K}_{\max}(x)$ and $\phi \in \mathcal{K}_{\min}(x)$, $\psi$ and $\phi$ are orthogonal---implying that  $\mathcal{K}_{\max}(x)$ and $\mathcal{K}_{\min}(x)$ are orthogonal subspaces---and $F(\ii_\xx(\pr{\psi}), \ii_\xx(\pr{\phi}))=0$. As such,  given the average conservation of $N$ by the interaction channel  $\ee$,   \lemref{lemma:WAY-distinguishability-inequality} implies that the following inequality must hold:
\begin{align*}
|\langle \psi | \nsys \phi \rangle | &\leqslant \| \nsys \| F\left(\Lambda(\pr{\psi}), 
\Lambda(\pr{\phi})\right) \\
& \leqslant \| \nsys \| \sum_{a=x,\overline{x}} \tr[\Z(a) \Lambda(\pr{\psi})]^{\frac{1}{2}}\tr[\Z(a) \Lambda(\pr{\phi})]^{\frac{1}{2}} \\
& = \| \nsys \| \sum_{a=x,\overline{x}} \langle \psi |\E(a) \psi\rangle^{\frac{1}{2}} \langle \phi|\E(a) \phi\rangle^{\frac{1}{2}} \\
& = \|\nsys\| \left( \| \E(x)\|^{\frac{1}{2}} (1 - \|\onesys - \E(x) \|)^{\frac{1}{2}} + (1 - \| \E(x) \|)^{\frac{1}{2}} \|\onesys - \E(x) \|^{\frac{1}{2}}\right).
\end{align*}
The second line uses the fact that for any states $\rho, \sigma$, it holds that $F(\rho, \sigma) \leqslant  \sum_a \tr[\F(a) \rho]^{\frac{1}{2}}\tr[\F(a) \sigma]^{\frac{1}{2}}$ for any  observable $\F$  \cite{Fuchs1995, Barnum1996}.  The third line uses the fact that $\Lambda$ is the conjugate channel to $\ii_\xx$ defined in \eq{eq:conjugate-channel}, and so it holds that $\tr[\Z(a)\Lambda(\rho)] =  \tr[\E(a) \rho]$ for all $\rho$ and  $a=x,\overline{x}$. To see how the final line is obtained, note that  we have $\langle \psi |\E(x) \psi\rangle = \| \E(x)\|$ and $\langle \phi| (\onesys - \E(x) ) \phi\rangle = \| \onesys - \E(x) \|$ by construction. For the first term, i.e.,   $a = x$,  we obtain $\langle \psi |\E(a) \psi\rangle = \| \E(x)\|$ and $\langle \phi|\E(a) \phi\rangle = \langle \phi|(\onesys - (\onesys - \E(x)) ) \phi\rangle = 1 - \langle \phi| (\onesys - \E(x) ) \phi\rangle = 1 - \| \onesys - \E(x) \|$. The second term for $a = \overline{x}$ is obtained in a similar manner.

\end{proof}

Let us note that if $\E$ commutes with $\nsys$, then \thmref{theorem:first-kind-WAY-distinguishability} imposes no restrictions on first-kindness. To see this,  let us  note that for any (possibly trivial) effect $\E(x)$,    and for any unit vectors $\psi, \phi$ satisfying  $\E(x)\psi = \|\E(x)\| \psi$ and $(\onesys - \E(x)) \phi = \|\onesys - \E(x)\| \phi$, it holds that $\<\psi| \E(x) \nsys \phi\> = \|\E(x)\| \<\psi|  \nsys \phi\>$ and $\<\psi|  \nsys \E(x) \phi\> = (1 - \|\onesys - \E(x)\|) \<\psi|  \nsys \phi\>$. It follows that if $[\E(x),\nsys]= \zero$ then either (i) $\<\psi|  \nsys \phi\> = 0$, or (ii)  $ \|\E(x)\| + \|\onesys - \E(x)\|=1$. Condition (i) implies that the lower bound of \eq{eq:distinguishability-WAY-bound} vanishes, and so no constraint is imposed. On the other hand, condition (ii) implies that $\|\E(x)\| + \|\onesys - \E(x)\| = 1 + (p_{\max}(x) - p_{\min}(x)) = 1$, where we recall from \corref{corollary:first-kind-fidelity-norm} that $p_{\max}(x) = \| \E(x)\|$ and $p_{\min}(x) = 1 - \| \onesys - \E(x)\|$ are the largest and smallest values from the set $\{p(x|z)\}$ given by \thmref{thm:first-kind-post-processing-norm-1}. Such equality is satisfied if and only if $p_{\max}(x) = p_{\min}(x) = \lambda$, in which case by \eq{eq:post-processing-first-kind} it follows that  $\E(x) = \lambda \onesys$ is a trivial effect. But \eq{eq:distinguishability-WAY-bound}  applies only to non-trivial effects, and so no constraints are imposed in such a case.    
 
Let us now show that   the first-kindness statement of  \thmref{theorem:Strong-WAY} may be recovered directly from \eq{eq:distinguishability-WAY-bound}.

\begin{corollary}\label{corollary:repeatable-WAY-distinguishability}
Consider a measurement scheme $\mm := (\ha, \xi, \ee, \Z)$ for  an $\E$-instrument $\ii$  acting in $\hs$, and assume that $\ee$  conserves an additive quantity $N = N\sub{\s} \otimes \oneapp + \onesys\otimes N\sub{\aa}$ on average, where $\nsys \in \los(\hs)$ and $\napp \in \los(\ha)$. If $\ii$ is a measurement of the first kind,  then for all effects $\E(x)$ that have both eigenvalue 1 and 0,
\begin{align*}
    \P [\E(x),  \nsys ] \P =  [\E(x),  \P \nsys \P] = \zero,
\end{align*}
 where  $\P = \P_0(x) + \P_1(x)$, with $\P_0(x)$ and $\P_1(x)$ orthogonal projections   onto the eigenvalue-0 and eigenvalue-1 eigenspaces of $\E(x)$, respectively.
\end{corollary}
\begin{proof}
 It is trivial to verify that $\supp(\P_1(x)) \equiv \kk_{\max}(x)$ and $\supp(\P_0(x)) \equiv \kk_{\min}(x)$ as defined in \thmref{theorem:first-kind-WAY-distinguishability}. Any $\varphi \in \supp(\P)$ may be written as $\varphi = \alpha \psi + \beta \phi$, where $\psi \in \kk_{\max}(x)$ and $\phi \in  \kk_{\min}(x)$ are unit vectors and $\alpha, \beta \in \co$. 
Since $\|\E(x)\| =  \|\onesys -\E(x)\|=1$ for any effect that has both eigenvalue 1 and 0, then by \thmref{theorem:first-kind-WAY-distinguishability} it follows that $\<\varphi| \P_1(x) \nsys \P_0(x) \varphi\> = \alpha^* \beta \<\psi | \nsys \phi\> = 0$ for all $\varphi \in \supp(\P)$. It follows that $ \P_1(x) \nsys \P_0(x) = \P_1(x) \nsys (\P - \P_1(x)) = \zero $, which implies that $\P_1(x) \nsys \P = \P_1(x) \nsys \P_1(x)$. Since the right hand side is self-adjoint, and $\P_1(x) = \P_1(x)\P = \P \P_1(x)$, it follows that $[\P_1(x), \P \nsys \P]=\zero$. But since $\P_1(x) = \E(x) \P = \P \E(x)$, we have $[\E(x), \P \nsys \P] = \P [\E(x) ,  \nsys ] \P =\zero$. This completes the proof.
\end{proof}

\bibliography{Projects-WAY}

\end{document}